\documentclass[11pt,a4paper,twoside, openany,openright]{book}
\usepackage[a4paper,includeheadfoot,margin=2.50cm]{geometry}
\renewcommand{\baselinestretch}{1.2}  
\usepackage{import}
\usepackage{example}
\usepackage{framed}
\usepackage{graphicx}
\graphicspath{{images/}}
\usepackage{pdfpages}
\usepackage{enumitem}
\usepackage{float}
\usepackage{caption}
\usepackage{subcaption}
\usepackage[toc,page]{appendix}
\usepackage{amsfonts}
\usepackage{amssymb}
\usepackage{mathrsfs}
\usepackage{euscript}
\usepackage{graphicx}
\usepackage{multirow}
\usepackage{algorithm}
\usepackage{algorithmic}
\usepackage{epstopdf}
\usepackage{epsf}
\usepackage{footmisc}
\usepackage{pdfpages}
\usepackage[utf8]{inputenc}
\usepackage[english]{babel}
\usepackage{amsmath,amsthm}

\usepackage{lgrind}
\usepackage{cmap}
\usepackage[T1]{fontenc}
\usepackage{longtable}
\usepackage{scrextend}
\usepackage{etoolbox}
\newtheorem{proposition}{Proposition}

\usepackage{verbatim}
\usepackage{amsmath,amssymb,amsfonts}
\usepackage{algorithmic}
\usepackage{graphicx}
\usepackage{textcomp}
\usepackage{algorithm}

\usepackage[T1]{fontenc}
\usepackage{mathtools}
\usepackage[framemethod=TikZ]{mdframed}

\surroundwithmdframed[
topline=false,
rightline=false,
bottomline=false,
leftmargin=\parindent,
skipabove=\medskipamount,
skipbelow=\medskipamount,
linewidth = 50 pt,
linecolor = lightgray
]{proof}

\usepackage{listings}
\usepackage{color}
\usepackage[colorinlistoftodos]{todonotes}
\usepackage{epstopdf}
\usepackage{csquotes}
\usepackage{wrapfig}
\usepackage{pdfpages}
\usepackage{todonotes}
\usepackage{multicol}
\usepackage{amsmath}
\usepackage{amssymb}
\usepackage{commath}      
\usepackage{smartdiagram} 
\usepackage{glossaries}   
\PassOptionsToPackage{hyphens}{url}
\usepackage{url}

\usepackage{quotchap}              

\usepackage[numbers]{natbib}       
\usepackage[nottoc]{tocbibind}     
\usepackage{tikz}

\usepackage{booktabs}
\usepackage{array}
\usepackage{ragged2e}  
\newcolumntype{L}[1]{>{\raggedright\let\newline\\\arraybackslash\hspace{0pt}}m{#1}}
\newcolumntype{C}[1]{>{\centering\let\newline\\\arraybackslash\hspace{0pt}}m{#1}}
\newcolumntype{R}[1]{>{\raggedleft\let\newline\\\arraybackslash\hspace{0pt}}m{#1}}

\usepackage{fancyhdr}
\renewcommand{\headrulewidth}{2pt}
\renewcommand{\footrulewidth}{0.1pt}
\fancypagestyle{plain}{ %
\fancyhf{} 
\renewcommand{\headrulewidth}{0pt} 
\renewcommand{\footrulewidth}{0pt}
}
\usepackage{hyperref}
\usepackage{hypcap}

\begin{document}

\frontmatter

\includepdf{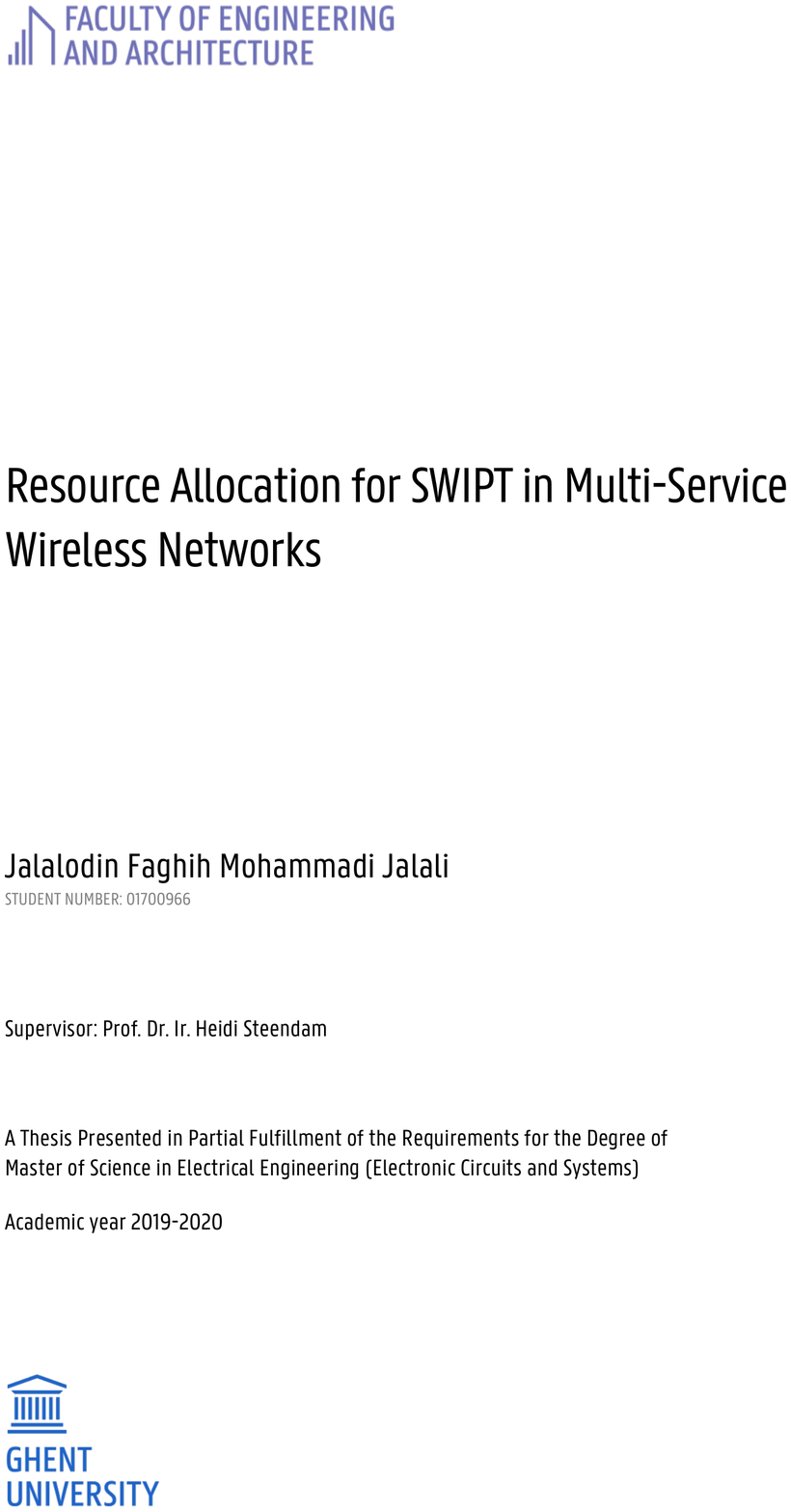}           
\newpage\thispagestyle{empty}\mbox{}  
\newpage\thispagestyle{empty}\mbox{}  
\newpage\thispagestyle{empty}\mbox{}  
\includepdf{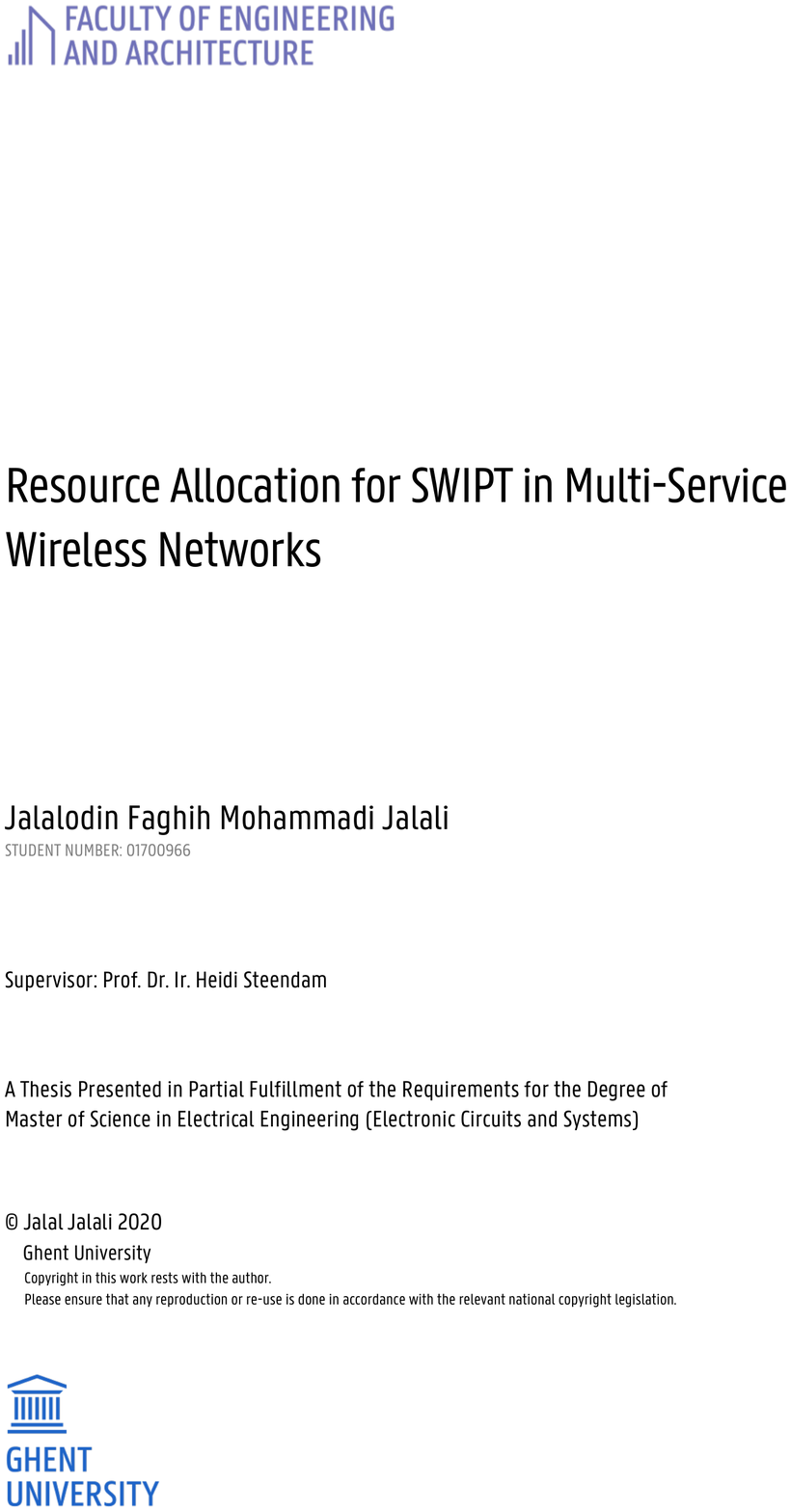}           
\newpage\thispagestyle{empty}\mbox{}  
\chapter*{Acknowledgments}

After an intensive year, I put the finishing touches on this Master's thesis. 
I have been able to discover new elements in various areas of the beautiful world of telecommunication. 
For this, I would like to sincerely thank a number of people who have supported me during this long journey.\\ 

First of all, I want to thank my supervisor, Prof. Dr. Ir. Heidi Steendam, for giving me the opportunity to work on this topic for my thesis. 
Her trust, guidance, and support helped me to successfully complete it.
I was fortunate for her supervision, especially her patience, understanding, and availability for discussing how to write scientific papers.\\

In addition, I express my gratitude to many friends and colleagues for our fruitful discussions and friendship. 
Among them, my special appreciation goes to my best friend Ata for his comments and invaluable knowledge, which reaped us significant research outputs. 
I am so lucky and honored to have amazing friends, including Pandidan, Sirus, Khasmakhu, Mamadu, Hani, Arakani, Mahboubeh, Keyvan, and Dennis, who gave me crucial feedback. 
Tinus, Wim, Jacques, and Junior count among those amazing friends whose help and concern I will always remember.\\

Finally, I would like to express my most profound gratefulness to my beloved parents, Afsaneh and Nouri, who are always sources of strength and perseverance. 
I also want to thank my spiritual mother, Nancy, for being part of my life, giving me guidance and offering me help far from my home country.   
Lastly, I would like to thank my partner,
Mary Grace, for her unconditional trust, care, admiration, and love.

\vspace{2cm}
\noindent
\rightline{Jalal Jalali}
\rightline{   June 2020}
\newpage\thispagestyle{empty}\mbox{}  
\chapter*{Permission to Use Content}
The author grants permission for this Master's thesis to be used for consultation and for parts of it to be copied for personal use.
For all other uses, however, the copyright must be respected, particularly with regard to the obligation to explicitly credit the source when quoting the results.

\vspace{2cm}
\noindent
\rightline{Jalal Jalali}
\rightline{   June 2020}

\newpage
\thispagestyle{plain} 
\mbox{}
\chapter*{}

{
\setlength{\baselineskip}{14pt}
\setlength{\parindent}{0pt}
\setlength{\parskip}{8pt}

\begin{center}
\vspace{-15em}
\noindent 
\textbf{\huge
Resource Allocation for SWIPT in}

\textbf{\huge
Multi-Service Wireless Networks}

by 

Jalal JALALI

A Thesis Presented in Partial Fulfillment of the Requirements for the Degree of  \\Master of Science in Electrical Engineering

Academic year 2019-2020

Supervisor: Prof. Dr. Ir. Heidi STEENDAM\\

Faculty of Engineering and Architecture\\
Ghent University

Department of Telecommunications and Information Processing\\
Chairman: Prof. Dr. Ir. Joris WALRAEVENS
\\ 

\end{center}

\vspace{2em}
\textbf{\Large{Summary}}

\vspace{0.7em}
The novel resource allocation for simultaneous wireless information and power transfer (SWIPT) is presented as a means of not only helping to communicate and access information with increasing efficiency in the next generation of mobile data networks, but also contributing to minimizing a network's overall power consumption by providing a green energy source.
First, a unique architecture is proposed that harvests energy from an access point (AP) without the receiver needing a splitter.
In the proposed system model, a portion of the spectrum is used for information decoding (ID) while the remaining portion is exploited for energy harvesting (EH) in an orthogonal frequency division multiple access (OFDMA) network. 
To investigate the performance gain, an optimization problem is formulated that maximizes the harvested energy of a multi-user single-cell OFDMA downlink (DL) network with SWIPT and also satisfies a minimum data-rate requirement for all users.
A locally optimal solution for the underlying problem, which is essentially non-convex due to the coupling of the integer variable, is obtained by using optimization tools. 
Second, the proposed system model is improved in order to investigate the resource allocation problem of needing to maximize throughput based on the separated receiver architecture in an OFDMA multi-user multi-cell system that uses SWIPT.
The resulting problem, which jointly optimizes the subcarrier assignment and power allocation, is a mixed-integer non-linear problem (MINLP) that is difficult to solve.
Third, a state-of-the-art harvesting technique at the receiver that is based on receiver antenna selection with a co-located architecture is explored to optimize the energy efficiency (EE) of a SWIPT-enabled multi-cell multi-user OFDMA network. 
This is referred to as a “generalized antenna-switching technique”.
Extensive simulation demonstrates the superiority of the proposed methodologies and presents interesting results.

\vspace{2em}
\textbf{\Large{Keywords}}

\vspace{0.7em}
D.C. programming, energy efficiency, energy harvesting, green communication, majorization minimization, MINLP, multi-user communication, non-convex optimization, OFDMA, resource allocation, SWIPT, throughput.

}
\newpage
\newpage\thispagestyle{empty}\mbox{}  
\tableofcontents                      
\listoffigures                        

\chapter{Abbreviations}
\fancyhf{}
\renewcommand{\headrulewidth}{2pt}
\fancyhead[LE,RO]{\thepage}
\fancyhead[RE]{\textit{ \nouppercase{Abbreviations}} }
\fancyhead[RE]{\textit{ \nouppercase{Abbreviations}} }
\fancyhead[LO]{\textit{ \nouppercase{\rightmark}} }
\renewcommand{\footrulewidth}{0.1pt}
\fancyfoot[CE,CO]{\nouppercase{Abbreviations}}
\fancyfoot[LE,RO]{JFMJ}

\begin{flushleft}
\vspace{-20mm}
\renewcommand{\baselinestretch}{1.5}
\normalsize
\hspace{20mm}
\begin{tabular}{ll}
5G    & Fifth Generation \\
6G    & Sixth Generation \\
ABCS  & Ambient Backscatter Communication System \\
AF    & Amplify-and-Forward \\
AP    & Access Point \\
AS    & Antenna Switching\\ 
AWGN  & Additive White Gaussian Noise \\
BS    & Base Station \\
CRN   & Cognitive Radio Network \\ 
CSI   & Channel State Information \\ 
D.C.  & Difference of Two Convex Functions \\
D2D   & Device-to-Device \\
DC    & Direct Current \\
DF    & Decode-and-Forward \\
DL    & Downlink \\
EE    & Energy Efficiency \\
EH    & Energy Harvesting\\
GI    & Guard Interval\\
ICT   & Information and Communication Technology\\
ID    & Information Decoding \\
IFFT  & Inverse Fast Fourier Transform \\
IoT   & Internet of Things \\
ISI   & Inter-Symbol Interference\\ 
KKT   & Karush-Kuhn-Tucker\\
LoS   & Line of Sight \\
LTE   & Long Term Evolution\\
MEC   & Mobile Edge Computing \\
\end{tabular}
\end{flushleft}

\newpage

\begin{flushleft}
\renewcommand{\baselinestretch}{1.5}
\normalsize
\hspace{20mm}
\begin{tabular}{ll}
MILP  & Mixed Integer Linear Programming \\
MIMO  & Multiple-Input Multiple-Output \\
MINLP & Mixed Integer Non-Linear Programming \\
MISO  & Multiple-Input Single-Output \\
MM    & Majorization Minimization \\
NOMA  & Non-Orthogonal Multiple Access \\
OFDM  & Orthogonal Frequency Division Multiplexing \\
OFDMA & Orthogonal Frequency Division Multiple Access \\
PS    & Power Switching \\
QoS   & Quality of Service\\
RF    & Radio Frequency \\
RFID  & RF Identification \\
SBS   & Small Base Station \\
SE    & Spectral Efficiency \\ 
SISO  & Single-Input Single-Output \\ 
SNIR  & Signal-to-Interference plus Noise Ratio \\
SNR   & Signal-to-Noise Ratio \\
SWIPT & Simultaneous Wireless Information and Power Transfer \\
SDMA  & Spatial Division Multiple Access \\
TDMA  & Time Division Multiple Access \\
TS    & Time Switching\\ 
UAV   &  Unmanned Aircraft Vehicles \\
UE    & User Equipment \\
UL    & Uplink \\
WDT   & Wireless Data Transfer \\
WPBC  & Wirelessly Powered Backscatter Communication \\
WPCN  & Wirelessly Powered Communication Network \\
WPT   & Wireless Power Transfer \\
WSN   & Wireless Sensor Network \\
\end{tabular}
\end{flushleft}

\chapter{Operators and Symbols}
\fancyhf{}
\renewcommand{\headrulewidth}{2pt}
\fancyhead[LE,RO]{\thepage}
\fancyhead[RE]{\textit{ \nouppercase{Operators and Symbols}} }
\fancyhead[RE]{\textit{ \nouppercase{Operators and Symbols}} }
\fancyhead[LO]{\textit{ \nouppercase{\rightmark}} }
\renewcommand{\footrulewidth}{0.1pt}
\fancyfoot[CE,CO]{\nouppercase{Operators and Symbols}}
\fancyfoot[LE,RO]{JFMJ}

\textbf{Operators}

\begin{flushleft}
\renewcommand{\baselinestretch}{1.5}
\normalsize
\hspace{5mm}
\begin{tabular}{ll}
    $|\mathcal{A}|$ & Cardinality of set $\mathcal{A}$\\
    $|x|$ & Absolute value norm of $x$\\
    $[x]^+$ & max$\{0,x\}$ \\
    $[x]^{-1}$ & Matrix inverse\\
    $[x]^{T}$ & Matrix transpose  \\
    $(\text{·})^{(t)}$ & Value of a variable under the $t^{th}$ iteration\\
    $f(\text{·}): \mathcal{A} \rightarrow \mathcal{B}$ & $f$ is a function from set $\mathcal{A}$ into $\mathcal{B}$  \\
    $\inf f$ & Infimum of function $f$  \\
    $\sup f$ & Supremum of function $f$  \\
    $\min f$ & Minimum of function $f$  \\
    $\max f$ & Maximum of function $f$  \\
    $\nabla f$ & Gradient of function $f$  \\
    $\mathcal{L}(\text{·})$ & Lagrangian function  \\
    $\mathcal{D}(\text{·})$ & Dual function  \\
    $\sum$ & Summation operator\\
    $\mathcal{CN}(\mu,\sigma^{2})$ & Circularly symmetric Gaussian distribution with mean $\mu$ and variance $\sigma^{2}$\\ 
\end{tabular}
\end{flushleft}

\newpage
\textbf{Symbols}

\begin{flushleft}
\renewcommand{\baselinestretch}{1.5}
\normalsize
\hspace{10mm}
\begin{tabular}{ll}
    $p^*$  & Primal optimal value \\
    $d^*$  & Dual optimal value \\
    $\mathscr{B}$ & Bandwidth \\
    $\mathcal{N}$ & Number or subcarriers \\
    $\mathcal{K}$ & Number or users in the network \\
    $\mathcal{J}$ & Number or cells in the network \\
    $h$ & Downlink channel gain for the wireless information transfer\\
    $g$ & Downlink channel gain for the wireless information transfer\\
    $a$ & Binary subcarrier indicators\\
    $x$ & Binary antenna selection indicators\\
    $\textrm{EH}$ & Amount of the harvested energy \\
    $R$ & Amount of the data-rate \\
    $\eta$  & Amount of the energy efficiency\\
    $\epsilon$ & Power conversion efficiency of an energy harvesting device\\
    $\kappa$ & Power amplifier efficiency\\
    $\sigma^2$ & Noise power\\
    $p_{max}$ & Maximum transmit power \\
    $R_{min}$ & Minimum data-rate requirement \\
    EH$_{min}$ & Minimum harvested energy requirement \\
\end{tabular}
\end{flushleft}

\newpage\thispagestyle{empty}\mbox{}  
\newpage\thispagestyle{empty}\mbox{}  

\mainmatter

\begin{savequote}[0.55\linewidth]
``Invention is the most important product of man’s creative brain. The ultimate purpose is the complete mastery of mind over the material world, the harnessing of human nature to human needs.''
\qauthor{ –  Nikola Tesla}
\end{savequote}

\chapter{Introduction}\vspace{5mm}
\fancyhf{}
\renewcommand{\headrulewidth}{2pt}
\fancyhead[LE,RO]{\thepage}
\fancyhead[RE]{\textit{ \nouppercase{\leftmark}} }
\fancyhead[LO]{\textit{ \nouppercase{\rightmark}} }
\renewcommand{\footrulewidth}{0.1pt}
\fancyfoot[CE,CO]{\nouppercase{\leftmark}}
\fancyfoot[LE,RO]{JFMJ}
\label{chap:intro}

\vspace{15mm}

Wireless communication is a broad and dynamic field that is enjoying the rapid development of new technologies needed to cope with the massive growth in the number of wireless communication devices and many practical applications. 
In the past few decades, wireless networks and communication devices have become an indispensable part of modern life. 
The next generation of wireless networks (called “fifth generation” (5G) by the standardization community) is designed to communicate and access information more efficiently. 
5G is being deployed to provide high data-rates for mobile users, massive internet of things (IoT) applications, device-to-device (D2D) communications, and low latency or extreme real-time tactile communications with high availability based on real-time immediate interactions. 
High data-rate access is extremely important, but the huge amount of power consumed by modern communication applications and networks is a major factor in global warming. 
This has inspired the notion of green and sustainable radio communication.
New wireless ecosystems and data networks supporting higher system throughput with greater energy efficiency are needed.
At the same time, they must provide a variety of services such as wireless power transfer, mobile node positioning, cooperative sensing of the surrounding environment, and distributed processing of wireless audio and video signals.

In this regard, a series of procedures, specifications, requirements, and constraints concerning the delivery of each of these services is accomplished through gradual design. 
At best, this leads to ad-hoc treatment of just one of the many services because the lack of structure causes the wireless infrastructure to be poorly managed. 
In contrast, a unique shared infrastructure provides plenty of room to holistically analyze and more systematically design the complete wireless ecosystem and optimize a multi-service wireless network. 
This systematic approach to multi-service wireless networks is called  “X-service design”.
In it, all the services and (computational and radio) resources are optimized adaptively and controlled cooperatively.

Today's cellular networks with approximately 6 million macro-cells worldwide consume a peak rate of almost 12 billion watts – the majority of the world’s wireless information and communication technology (ICT) power consumption \cite{SWIPT_New_Paradigm_for_Green_Communications}.
How can we reduce this enormous consumption of power? 
By remembering Nikola Tesla's late-nineteenth-century dream of a “wirelessly powered world”. 
Radio signals convey power and can potentially also be used to deliver energy to remote devices. 
Tesla's seemingly far-fetched idea has recently triggered interest in wireless power transfer (WPT), which when coupled with information transfer is referred to as simultaneous wireless information and power transfer (SWIPT)~\cite{4595260}. 
WPT can be a by-product of wireless data transfer (WDT) networks in which devices capture ambient power, and thus contribute to minimizing a network's overall power consumption using green energy. 
This partially responds to the urgent question, although it is far from an ideal answer. 
If a certain quality of power transfer must be secured, base stations (BSs) can play an active role in power delivery and reserve part or all of their resources for this specific objective of reducing a network's overall power consumption. 
This is one example of positioning services \cite{6477849}. 

Let's think for a moment of a hypothetical meeting of the minds between Claude Shannon, the father of information theory, and Nikola Tesla.
Tesla tried to build a circuit to deliver power to a load wirelessly, Shannon wanted to use such a circuit just for sending information \cite{5513714}.
Path loss, energy harvester sensitivities that require a significant signal level, and the limits of radio transmit power mean that WPT can be effective only over distances like those found in ultra-dense networks \cite{7476821}.
Thus, network densification could be a point of convergence for positioning, WPT, and WDT networks – with the ultimate goal that could be described as “zero RF power” SWIPT-enabled networks. 
This is where Shannon meets Tesla.

\section{The Architecture of an RF Energy Harvesting Network}
Wireless communication systems equipped with energy harvesting (EH) receivers have increasingly been attracting attention \cite{Proceeding}.
A radio frequency (RF) EH device has a sustainable power supply from a radio environment that provides harvestable energy from RF signals for information processing and transmission.
A practical example of this is wireless sensor networks (WSNs), in which ambient energy is converted to electrical energy via an EH device – not only to enable a long period of operation in WSNs but also as an alternative to replacing the battery~\cite{Ambient_Harvest}.
However, traditional EH devices may not be appropriate for many applications due to their complex mechanical constraints such as form factor and cost, and they may not always be available in indoor environments.
Moreover, conventional EH approaches usually depend on renewable energy sources like solar, tide, wind, thermal, geothermal, and vibrations, which are usually unpredictable and hard to control~\cite{K_S}.
Hence, proactive WPT as an EH method that only needs an RF EH circuit and has low cost and small form factor should be considered when studying how to jointly design WPT and WDT in a SWIPT-enabled network~\cite{Fundamental,23}.

In fact, WPT presents a viable solution for facilitating sustainable communication networks serving energy-limited communication devices in which wireless devices can communicate through electromagnetic waves in the RF band.
In an RF energy harvesting mode, the range from 3kHz to as high as 300GHz is designated for radio signals to carry energy as electromagnetic radiation.
Thus, by recalling that the RF signals carry both information and energy, RF energy radiated by the transmitter(s) can be recycled at the receivers to prolong the network's lifetime~\cite{1}. 

In order to be able to harvest energy from an RF source, a general network architecture aided with a means of harvesting modules must be present. 
Figure (\ref{fig:1}), which was adapted directly from \cite{S_Harvest}, illustrates the block diagram of a network equipped with an RF EH device. 
As can be seen in figure (\ref{fig:1.1a}), the application module provides the data to be processed by a low-power microcontroller, while the low-power transceiver is employed to either transmit the processed information or to receive data for further processing. 
The next major component in figure (\ref{fig:1.1a}) is the RF energy harvester that collects RF signals and converts them into electricity. 
The converted signal then goes to a power management module, 
which either stores the electricity obtained from the RF energy harvester in an energy storage device, such as a rechargeable battery, thus helping the users to save excess energy for future use (\textit{harvest-store-use} mode), 
or uses it directly to transmit information without saving the energy (\textit{harvest-user} mode) \cite{Modeling}. 

\begin{figure}[!t]
\vspace*{1mm}
\centering
\begin{subfigure}{0.75\textwidth}
\hspace*{-1cm}
\includegraphics[width=14.5cm,trim=4 4 4 4,clip]{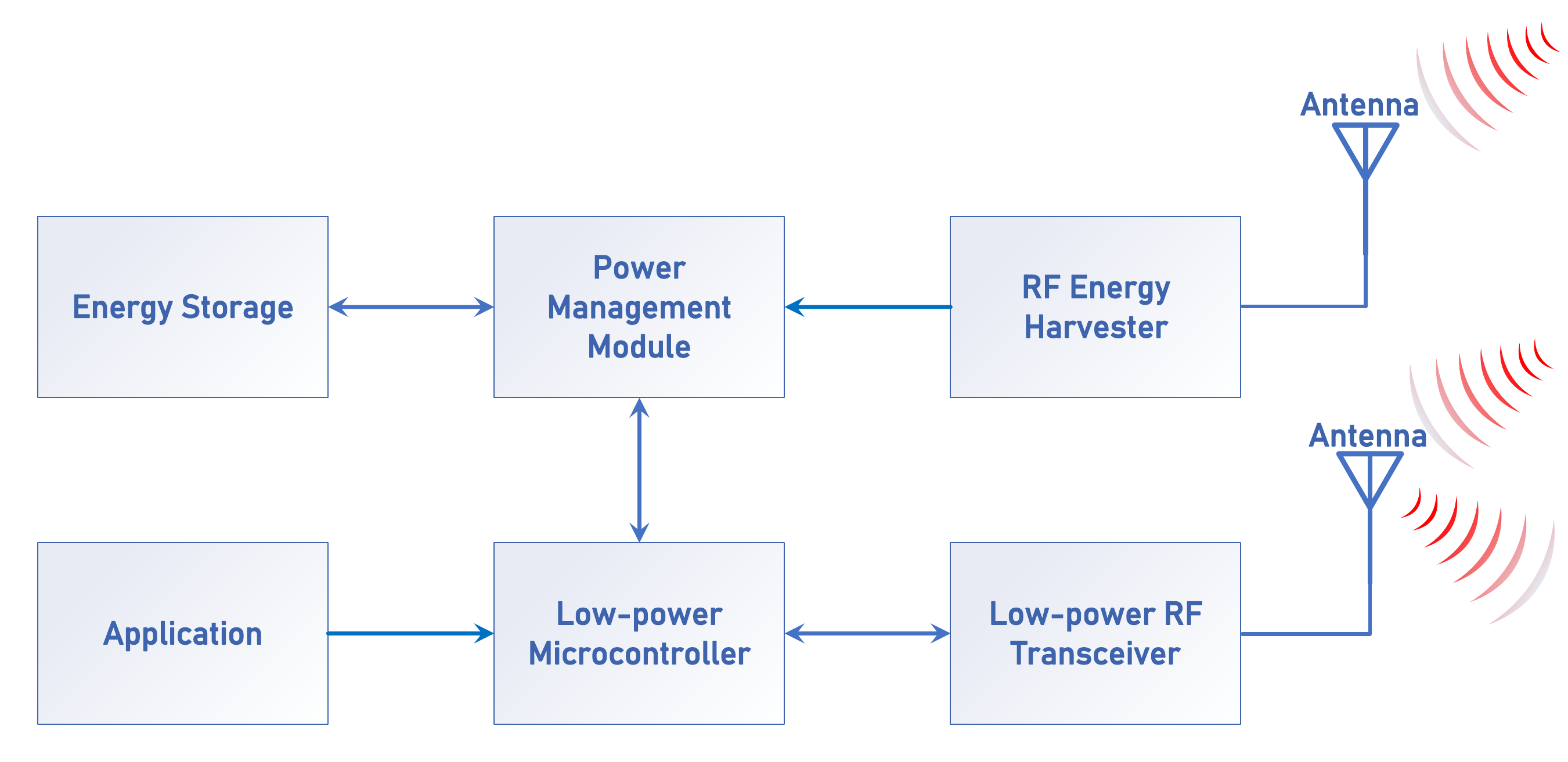}
\caption{General network architecture for energy harvesting}
\label{fig:1.1a}
\end{subfigure}
\begin{subfigure}{0.75\textwidth}
\hspace*{-2cm}
\vspace*{2mm}
\includegraphics[width=14.5cm,trim=4 4 4 4,clip]{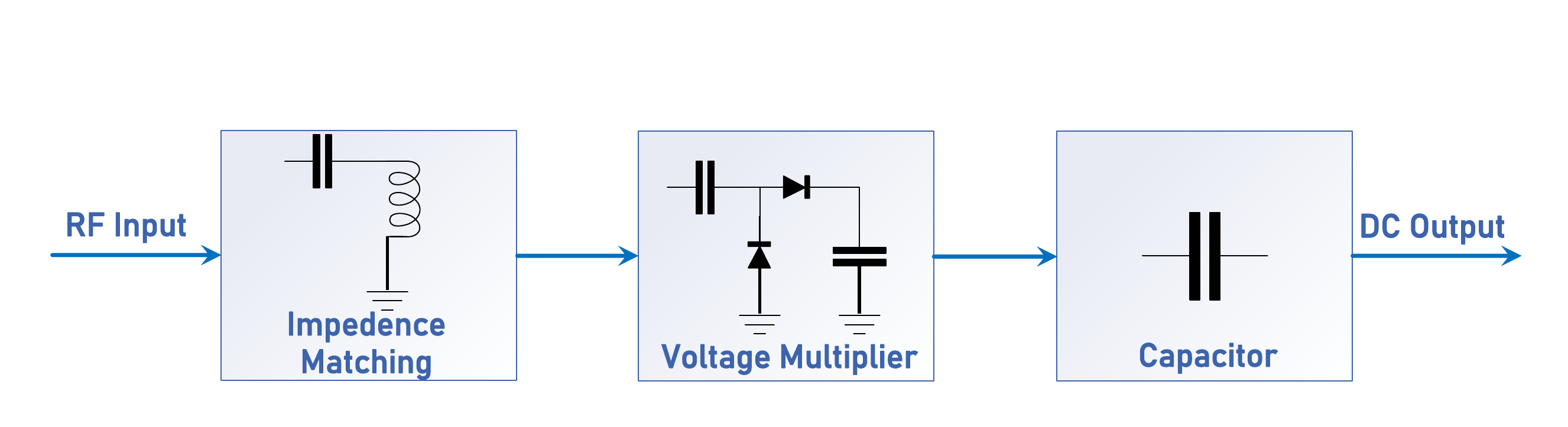}
\caption{RF energy harvester}
\label{fig:1.1b}
\end{subfigure}
\caption{The block diagram of an energy harvesting device [10]. }
\label{fig:1}
\end{figure}

Figure (\ref{fig:1.1b}) illustrates an RF harvester device with input from an RF antenna, an impedance matching circuit, a voltage multiplier, and a capacitor to create the output.
It is worth mentioning that an RF energy harvester typically operates over a range of frequencies:
The RF antenna that provides input to an RF EH unit can be designed to work on either single or multiple frequency bands, facilitating the energy harvesting from single or multiple sources at the same time.
To maximize the power transferred from the antenna to the voltage multiplier, an impedance matching in the form of a resonator circuit operating at the designed frequency, is utilized.
Figure~(\ref{fig:1.1b}) shows the diodes of the rectifying circuit that are the main component of the voltage multiplier that converts RF signals into direct current (DC) voltage levels that can be used to load an electronic circuit, 
where the capacitor ensures that the energy generated is smoothly delivered to the load~\cite{1435362,Fully_integrated,Recycling_ambient,Rate_Energy}.
Since RF signals carry both energy and information, an RF energy harvesting device like that shown in figure (\ref{fig:1}), could theoretically also simultaneously perform information decoding from the same RF signal input using the same antenna or antenna array. This concept has been defined as SWIPT. 

\section{Energy Harvesting Modes}
Future generation wireless networks not only have limited spectrum resources but also must operate with low-power batteries. 
Recently, energy-efficient communication systems, or “green radios”, have been increasingly attracting attention from the research community due to their ability to improve system performance while simultaneously diminishing the energy consumption of the communication devices \cite{An_Overview}.
Reducing wireless network energy consumption is not only essential for prolonging battery lives but also crucial for the environment. 

Although ICT is to blame for more than $2$ percent of CO$_2$ emissions worldwide, it also presents solutions for drastically diminishing the remaining $98$ percent of CO$_2$ emissions~\cite{webb2008smart}.
To significantly minimize ICTs' carbon footprint and environmental impact, we need new and efficient communication techniques~\cite{Green_radio}.
In recent literature, energy efficiency (EE), which measures the number of bits communicated per unit of energy consumed (bits/joules delivered to the receivers), has emerged as the performance metric to evaluate a communication system's energy consumption and guarantee green communication \cite{Tradeoff_Green,Toward_dynamic}.
However, today's galloping development of wireless communication technologies is increasing energy consumption and carbon emissions full tilt, further aggravating environmental concerns. 
According to \cite{Global_footprint}, the percentage of global carbon emissions due to ICTs is estimated to reach $5$ percent by the end of 2020, and the situation will only be exacerbated in coming years with the arrival of beyond 5G- and sixth generation (6G)-enabled networks.

In this regard, networks that harvest energy can definitely decrease the carbon consumption of high data-rate wireless systems by exploiting energy from the environment \cite{4}. 
In communication devices, energy harvested from RF signals is a random parameter that depends on the channel fading coefficients using circuitry like that discussed in the previous section~\cite{5513714,3}. 
Wireless communication energy harvesting by means of an RF EH device can be explored in a variety of ways, among which, the three most common configurations – WPT, WPCN, WPBC, and SWIPT – will be discussed below.

\subsection{WPT}

As shown in figure (\ref{fig:1.WPT}) in this model, RF energy is transferred in the downlink (DL) direction.
A WPT-enabled network includes a transmitter connected to the main power system, such as a base station (BS) or an access point (AP), that supplies power to an electrical component (or a portion of a circuit that consumes electrical power) without any wired interconnections.
The WPT-enabled network uses electromagnetic waves in the surrounding environment that can be obtained from \textit{near-field} (non-radiative) or \textit{far-field} (radiative) regions to power electrical components \cite{Radio_Science}.
In general, near-field techniques do not support the mobility of an energy harvesting device.
Therefore, it is preferable to transfer information through a far-field RF band in which 
the distance is much greater than the diameter of the transmit antenna,
instead of using methods such as inductive coupling, capacitive coupling, or magnetic resonant coupling for a near-field that corresponds to the area of just one wavelength of the transmitting antenna \cite{Near_field}.

WPT-enabled RF signals are anticipated to engender lots of applications and opportunities by providing cost-efficient, predictable, dedicated, on-demand, perpetual, and reliable energy supplies to energy-constrained wireless networks, where no wires, contacts, or batteries are needed. 
Numerous research activities are laying the groundwork for the future of wireless networking that transcends conventional communication-centric transmission.
For instance, the authors in \cite{3} proposed a randomly deployed power-beacon-based hybrid cellular network that wirelessly powers mobile users as they recharge their devices.
In \cite{7}, the total outage probability of an ad-hoc network overlaid with power-beacon was analyzed. The authors employed the stochastic geometry method in a harvest-store-use mode to study network performance in terms of the power and channel outage probability.
In \cite{One-Bit-Feedback}, multi-user multiple-input multiple-output (MIMO) WPT was considered with a new channel learning method that requires only one feedback bit from each EH receiver to the power transmitter per feedback interval.
The power transmitter uses the feedback information to coordinate the transmit beamforming in subsequent training intervals and concurrently obtains updated estimates of the MIMO channels to different EH users by solving an optimization problem. 
Building on that, \cite{Adaptively_Directional} presented an adaptive directional WPT methodology for prolonging the lifetime of a WSN by offering a sustainable power supply to the distributed sensor nodes.
Although today's wireless networks were designed solely for communication purposes, with their adoption of new technologies, they are evolving toward 6G.
Nonetheless, WPT development is still in its infancy and has not even entered its first generation: No single standard has yet been released on far-field WPT.
The WPT framework will be particularly suitable for future wireless networks with ubiquitous and autonomous low-power and energy-limited devices, massive IoT connections, and D2D {\text{communications}}.

\begin{figure}[!t]
\centering
\includegraphics[width=16cm,trim=4 4 4 4,clip]{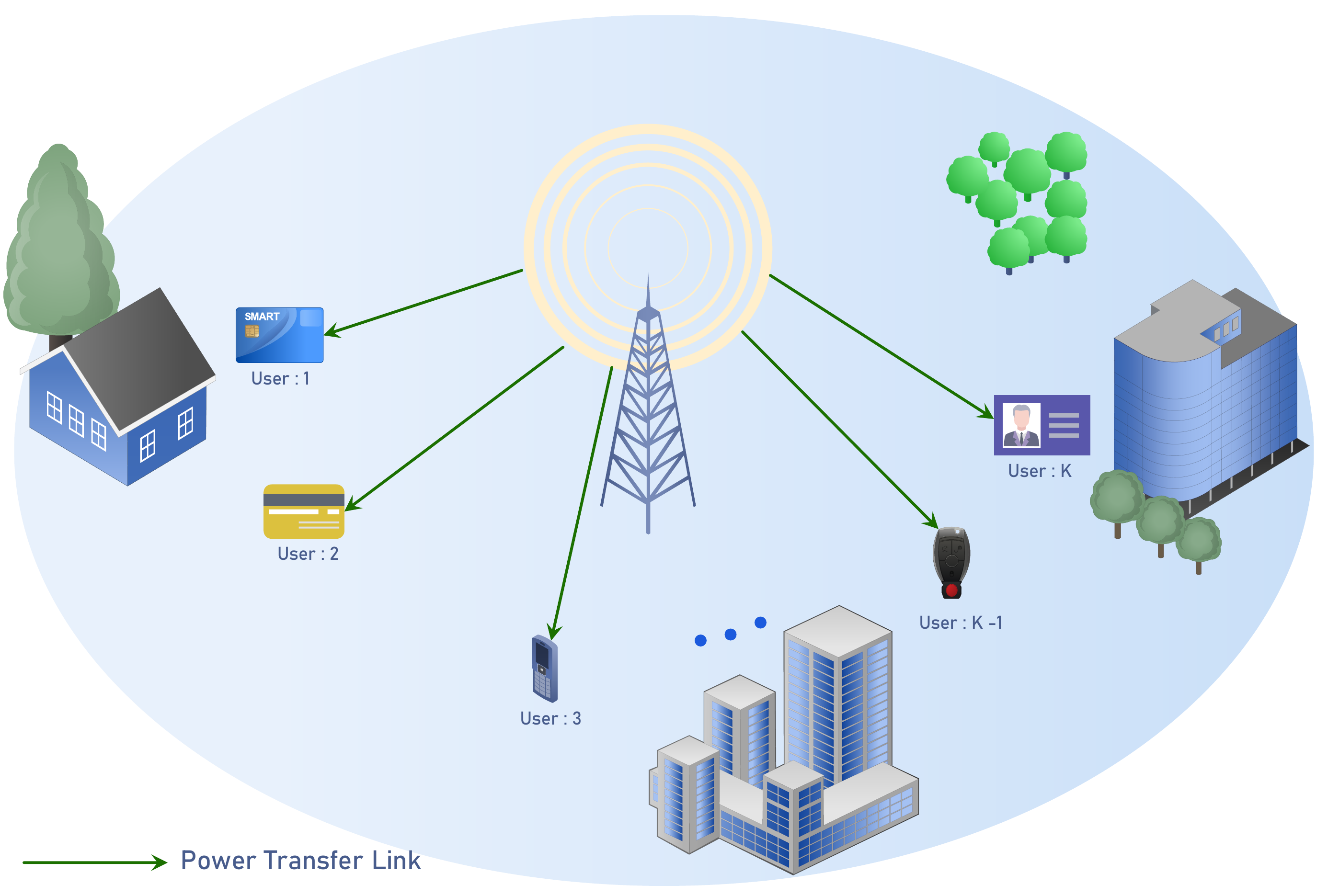}
\caption{The WPT system model. }
\label{fig:1.WPT}
\end{figure}

\subsection{WPCN}
Wirelessly powered communication networking (WPCN) is a new networking paradigm in which batteries for wireless communication devices can be remotely provisioned through microwave WPT technology.
Figure (\ref{fig:1.wpcn}) shows WPCN approach of transferring energy in the downlink (DL) direction while information is communicated in the uplink (UL) direction.
The receiver is a low-power device that harvests energy in the DL, and then uses the harvested energy to transfer data in the UL.
WPCN can effectively reduce cost and enhance communication performance by eliminating the battery charge limit. 
In addition, WPCN enjoys full control over its power transfer: Transmit power waveforms can be adjusted to provide a permanent energy supply under varying physical conditions and service requirements \cite{Wireless_powered}.

WPCN is intended to strike a balance between energy supply limitations and data transmission in order to optimize communication network performance.
WPCN could be suitable for a variety of low-power applications (up to several milliwatts) such as wireless sensor WSNs, IoT, and RF Identification (RFID) networks. 
WPCN enables low-power applications to operate longer while actively transmitting at much larger data-rates and from a greater distance than conventional backscatter-based communications \cite{8766912}.

\begin{figure}[!b]
\centering
\includegraphics[width=16cm,trim=4 4 4 4,clip]{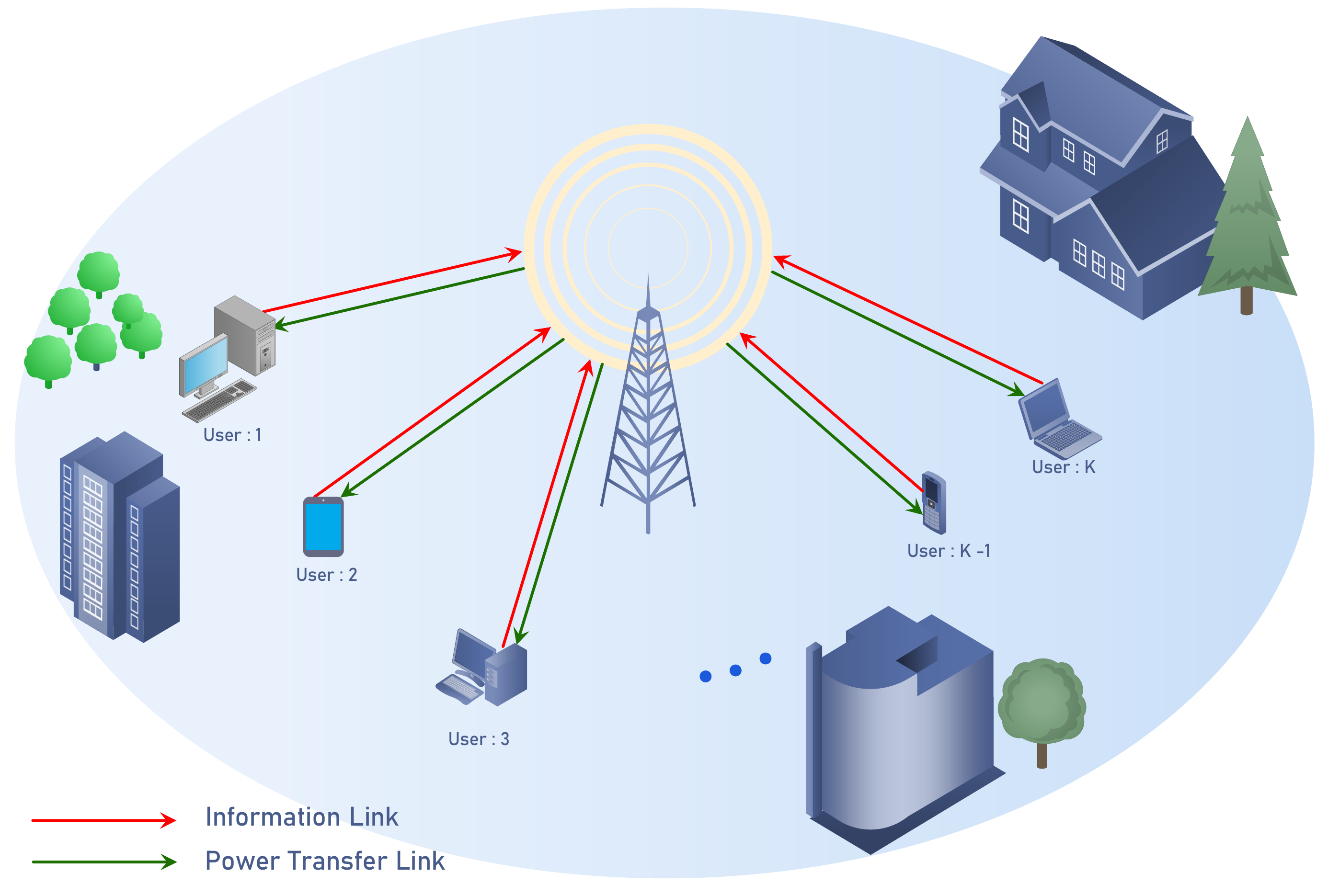}
\caption{The WPCN system model. }
\label{fig:1.wpcn}
\end{figure}

WPCNs promise to significantly enhance performance, but building an efficient WPCN is challenging.
In \cite{Throughput_WPCN}, the authors considered a hybrid AP based on the harvest-store-use mode with a constant power supply coordinating the wireless energy transmissions in the DL direction to a set of distributed users that have no other energy sources. 
Once the users have harvested enough energy, they send their independent information signals to the hybrid AP in the UL direction using time-division-multiple-access (TDMA). 
Because of WPCN's distance-dependent signal attenuation in both the DL WPT and the UL WDT, a user who is far from the hybrid AP receives less wireless energy than a user who is closer to the DL communication.
In order to have reliable information transmission, the distant user must transmit more power in the UL direction. 
This phenomenon is known as the \textit{doubly near-far} problem. 
\cite{Optimal_WPCN} studied an optimal resource allocation policy in which WPCN had simultaneous WPT in the DL and WDT in the UL. 
They addressed the doubly near-far problem in which the hybrid AP operates in full-duplex.
The authors in \cite{8580593} considered multiple-input single-output (MISO) WPCN, in which the single-antenna users harvest energy from a multi-antenna AP in DL direction and then retransmit information to the AP in the UL direction – using the TDMA or the spatial division multiple access (SDMA) technique. 
With multiple antennas, the AP can utilize energy beamforming in
the DL and employ the multiplexing-gain or receive beamforming-gain in the UL direction \cite{Energy_Beamforming_WPCN}.
In addition to all these techniques in various directions of research, WPCNs continue to present new research problems for future applications. 

\subsection{WPBC}
Although WPT and WPCN offer many advantages, these EH schemes still face critical limitations when adopted in low-power low-cost networks, such as WSN, RFID, IoT, and D2D applications. 
On one hand, in WPTs, users only harvest energy with no RF data transmission or reception, whereas in WPCNs, users may need lots of time to harvest the RF energy needed for data transmission, which limits the system's performance.  
To overcome these deficiencies, wirelessly powered backscatter communication (WPBC) networks, also known as ambient backscatter communication systems (ABCSs), are proposed as an alternative that can significantly improve network performance.
This innovative technique facilitates ubiquitous communication: Devices can interact among themselves at unprecedented scales and in previously inaccessible locations by using existing ambient RF signals, rather than generating their own radio waves \cite{8368232}.
In WPBCs, energy is transferred to a tag in the DL direction, and information is transferred to a reader in the UL direction, as shown in figure~(\ref{fig:1.WPBC}). 
More specifically, the information on a tag (a tagged object) is communicated to a nearby reader (e.g., an RFID reader) via backscatter modulation (tag-to-reader UL).
Since tags do not require oscillators to generate carrier signals and do not have any dedicated power infrastructure, backscatter communications – a passive method of communication – benefit from much less power consumption than conventional radio communications.

\begin{figure}[!t]
\centering
\vspace*{5mm}
\includegraphics[width=16cm,clip]{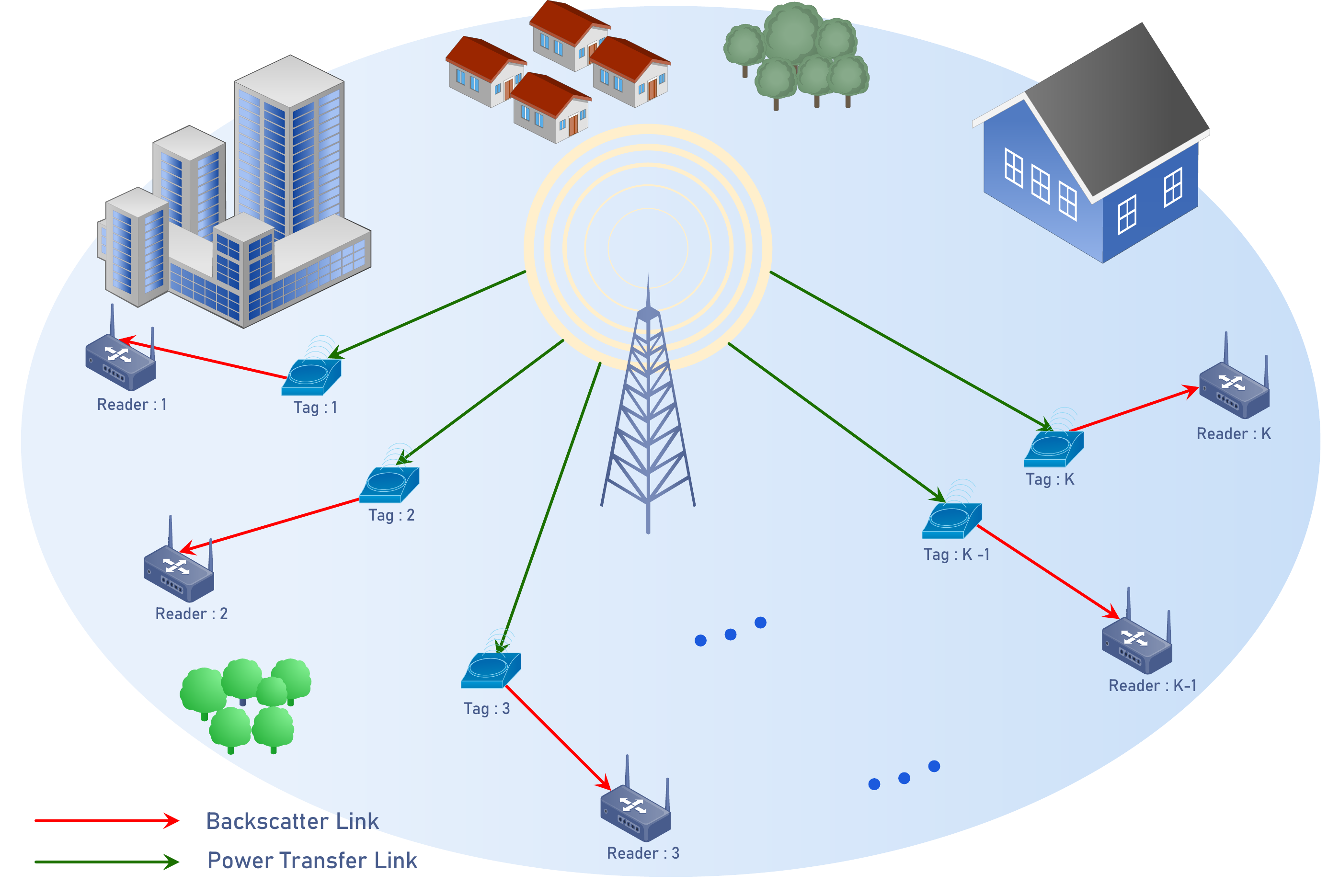}
\caption{The WPBC system model, where the backscatter modulation on a tag is used to reflect and modulate the incoming RF signal in order to communicate with a reader (e.g., a wireless router).}
\label{fig:1.WPBC}
\end{figure}
Quite a lot of research has been conducted on the topic of WPBC networks.
In \cite{8863924}, optimal resource allocation for the WPBC is investigated, where 
energy-constrained users harvest energy from RF signals transmitted by a multi-antenna source to power their future active wireless information backscatter transmission to an AP.
A theoretical trade-off in WPBC networks~–~between the reliability of the backscatter communication and the harvested energy at the tag, measured in terms of signal-to-noise ratio (SNR) at the reader – was studied in \cite{7950915}.
The authors in \cite{7876867} proposed a novel network architecture that enables D2D communication to be modeled as a WPBC. 
In \cite{6685977}, the authors aimed to integrate the WPBC and RFIDs to study perceptions of power and spectral efficiency with respect to energy constraints.
They did this by deriving the diversity-multiplexing trade-off for RFID MIMO channels.
Notwithstanding the progress gained through implementing WPBC networks, diverse issues such as security, reliability, the data-rate of communications, and the development of a global standard still need to be addressed. 
That is beyond the scope of this thesis.

\subsection{SWIPT}\vspace{-3mm}
Because RF signals can simultaneously carry information and energy, an intriguing new research domain has recently developed from multiple WPT technologies. 
What is called a simultaneous wireless information and power transfer (SWIPT) is a promising technique for wireless communication systems~\cite{K_S}.
The SWIPT paradigm enables energy-constrained wireless user equipments (UEs) to harvest energy and process the information simultaneously by utilizing RF signals transmitted from a BS, mobile BS (e.g., a drone), or AP.
In this model, energy and information signals are simultaneously transferred in the DL direction from one or multiple BSs or APs to one or multiple receivers for simultaneous information decoding (ID) and energy harvesting (EH).
The ideal receiver for enabling SWIPT has circuitry to perform ID and EH at the same time \cite{4595260}, rather than two dissimilar types of circuitry to perform EH and ID separately at the receiver~\cite{6133872}.

It is also worth noting that performing EH and ID operations at the same time does not necessarily mean that these operations are carried out on the same received signal: 
That is practically impossible since the information content of the RF domain signals would be entirely destroyed by harvesting energy on the signals.
Furthermore, a single antenna receiver may not be able to create a reliable energy supply due to its limited resources for collecting energy. 
Enabling SWIPT requires using separate antennas for both the EH and ID receivers, or splitting the received RF signal into two separate parts, one for EH and the other for ID operations, by using a \textit{splitter}.

EH and ID receivers for enabling SWIPT can be classified into two broad architectural categories: separated or co-located receiver. 
In a separated architecture, as shown in figure (\ref{fig:1.seperated}), the EH and ID receivers are two distinct devices with separate antennas that experience different channels from the transmitter.
The EH receiver is a low-power device capable of harvesting energy and the ID receiver is only able to process data. 
Since the ability to harvest energy deteriorates with distance, EH receivers ought to be closer to the BS or AP than ID receivers (which have to be spatially separated). 
This explains why an inner radius and outer radius are used to distinguish between EH and ID receivers in figure(\ref{fig:1.seperated}). 
On the other hand, with co-located SWIPT architecture, each receiver gets identical channels from the transmitter and is a low-power device that can perform EH and ID at the same time, as illustrated in figure (\ref{fig:1.colocated}). 

\begin{figure}[p]
\centering
\includegraphics[width=15.2cm,trim=4 4 4 4,clip]{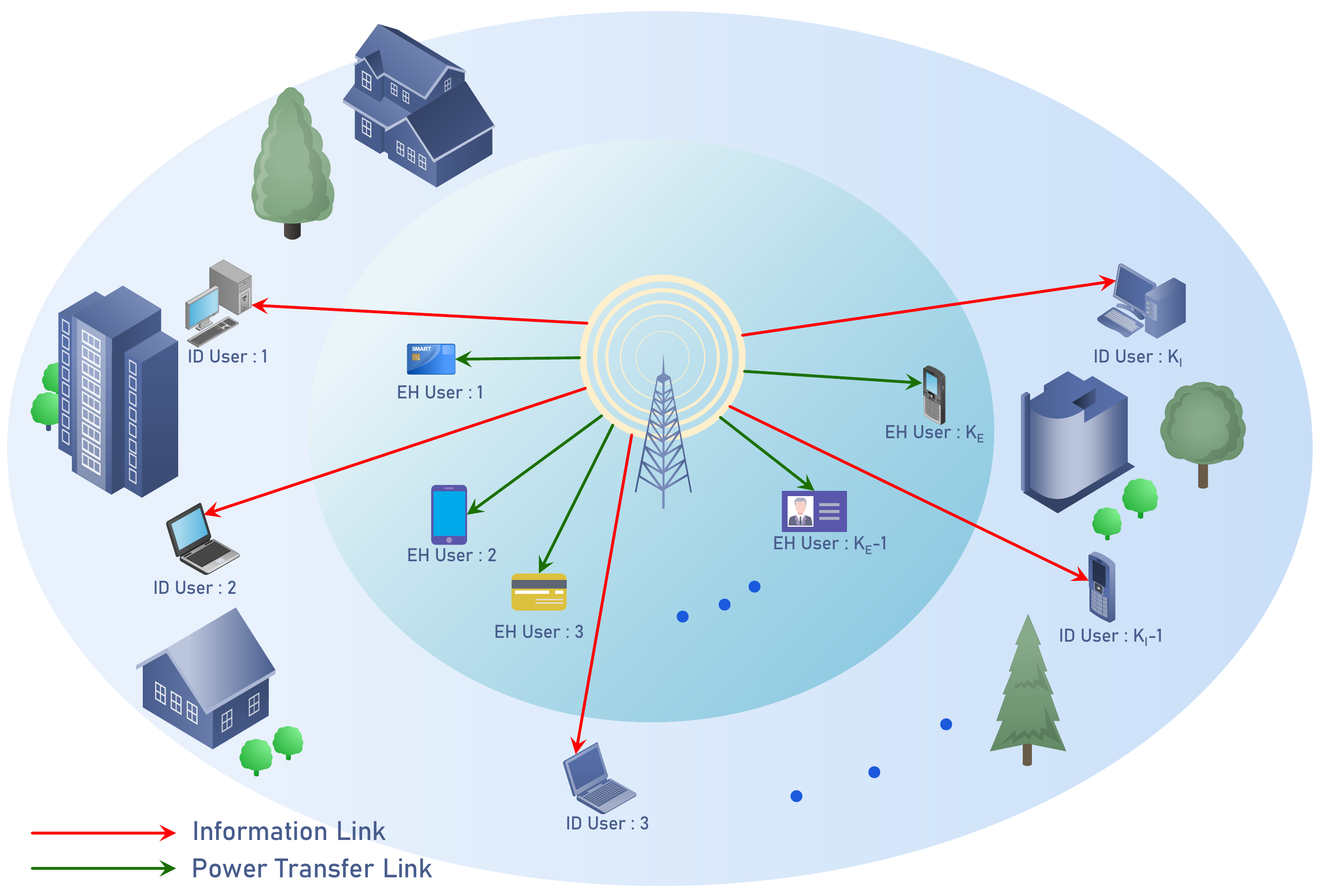}
\caption{The separated SWIPT system model. }
\label{fig:1.seperated}
\end{figure}

\begin{figure}[p]
\centering
\includegraphics[width=15.2cm,trim=4 4 4 4,clip]{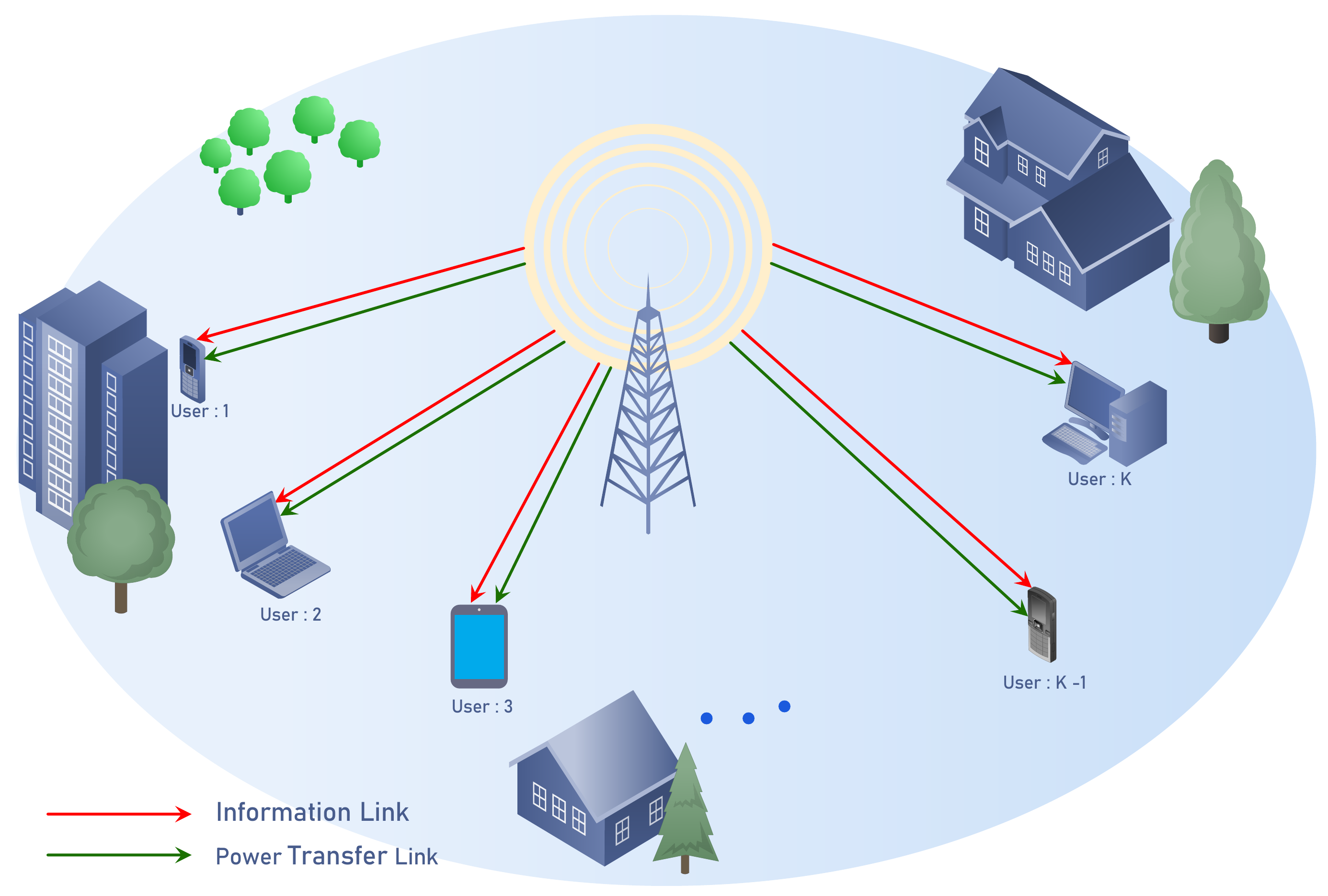}
\caption{The co-located SWIPT system model. }
\label{fig:1.colocated}
\end{figure}

\begin{figure}[p]
\vspace{-5mm}
\centering
\begin{subfigure}{0.75\textwidth}
\includegraphics[width=9.8cm,trim=4 4 4 4,clip]{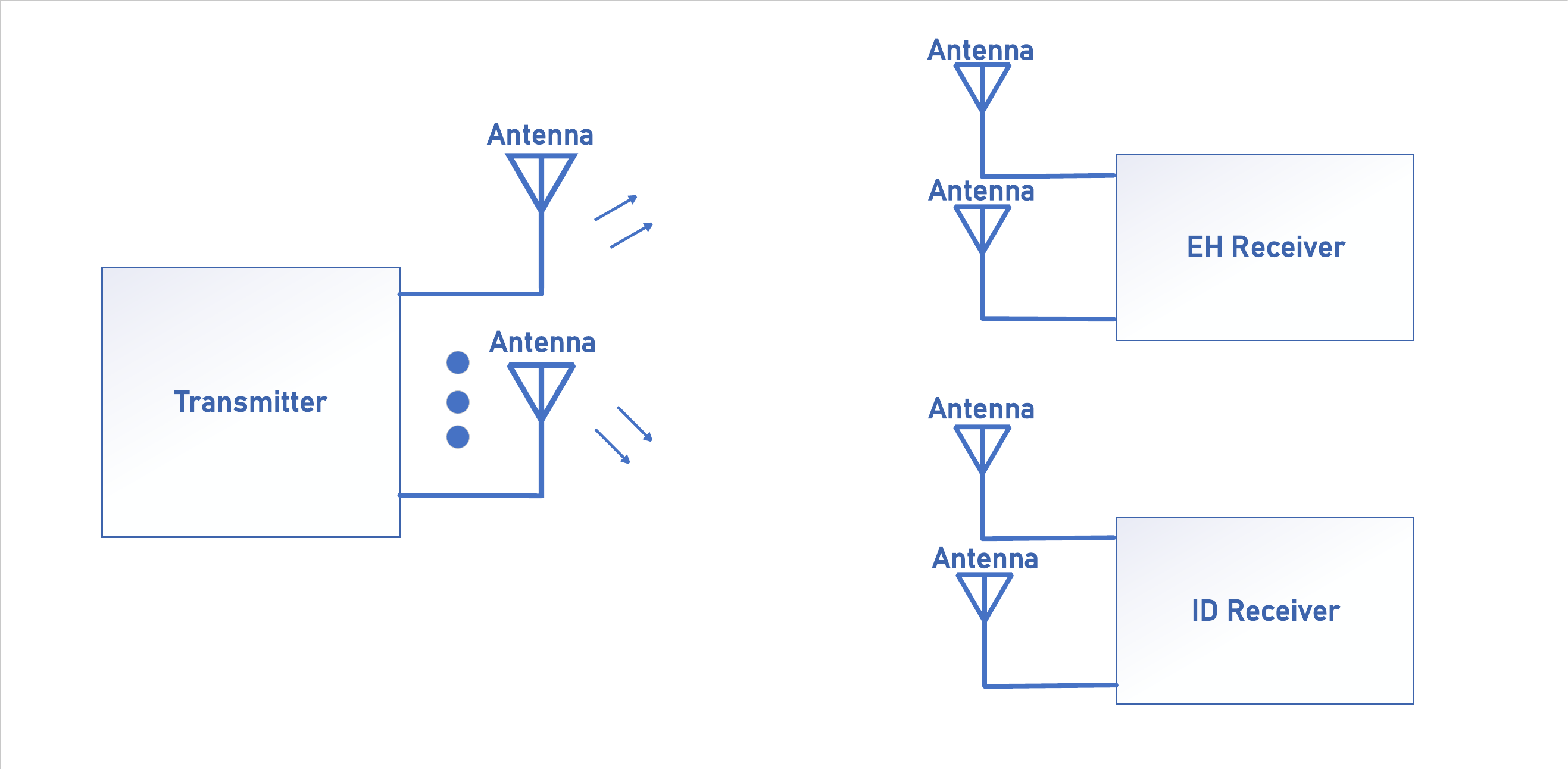}
\caption{Separated receiver architecture.}
\label{fig:1.4}
\end{subfigure}
\begin{subfigure}{0.75\textwidth}
\includegraphics[width=9.8cm,trim=4 4 4 4,clip]{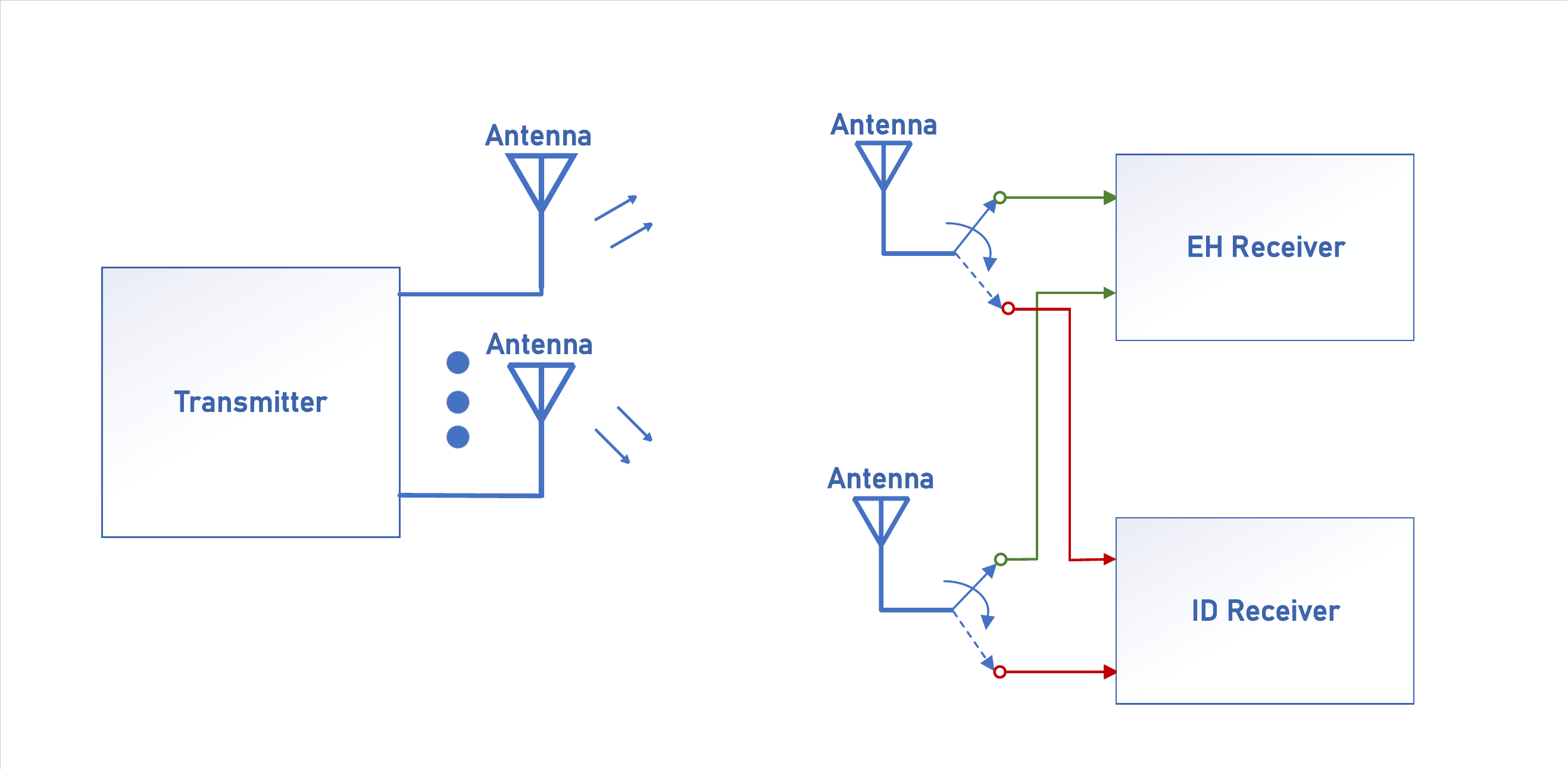}
\caption{Time switching (TS) approach to realize co-located SWIPT architecture.}
\label{fig:1.5}
\end{subfigure}
\begin{subfigure}{0.75\textwidth}
\includegraphics[width=9.8cm,trim=4 4 4 4,clip]{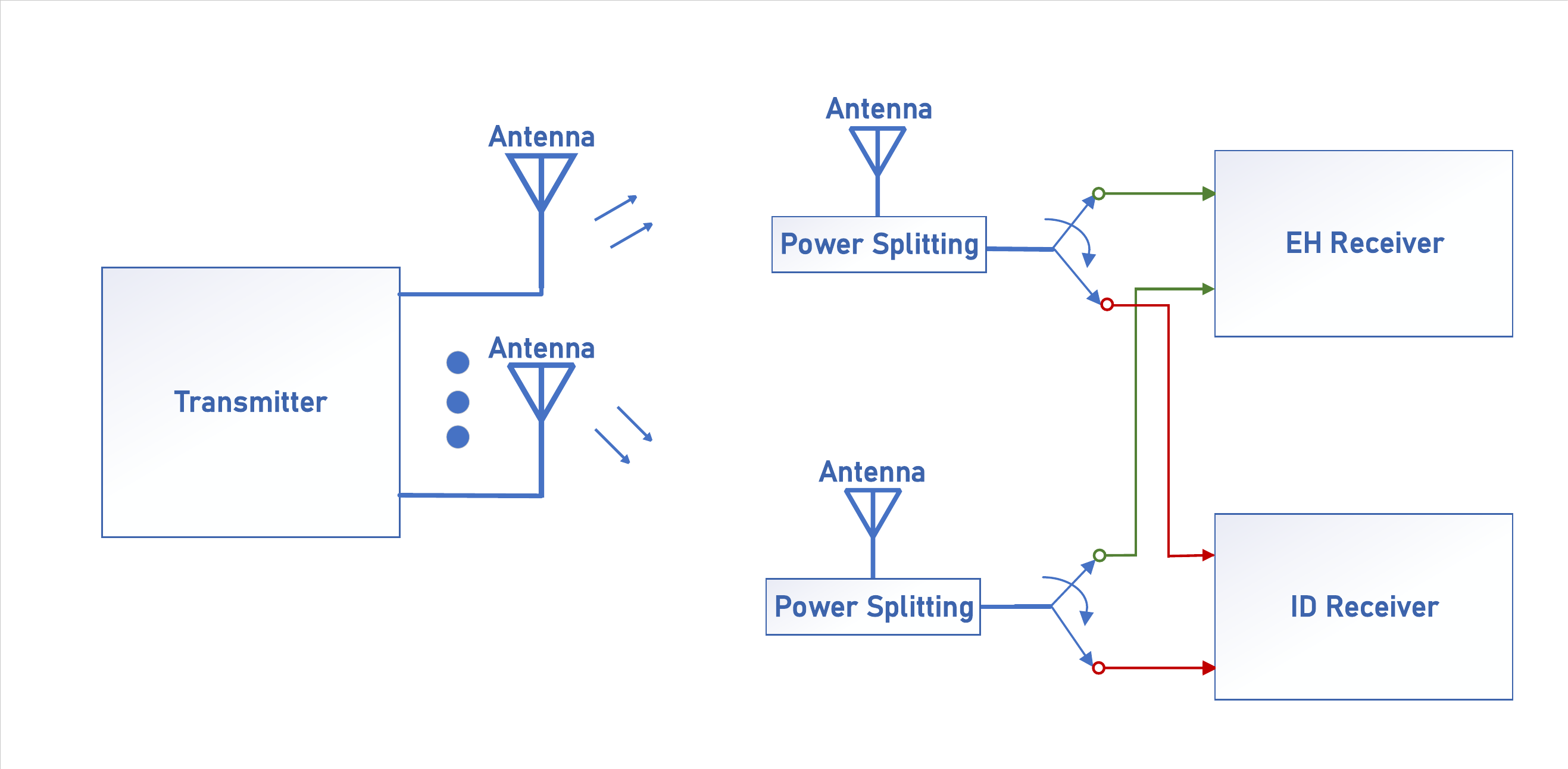}
\caption{Power splitting (PS) approach to realize co-located SWIPT architecture.}
\label{fig:1.6}
\end{subfigure}
\begin{subfigure}{0.75\textwidth}
\includegraphics[width=9.8cm,trim=4 4 4 4,clip]{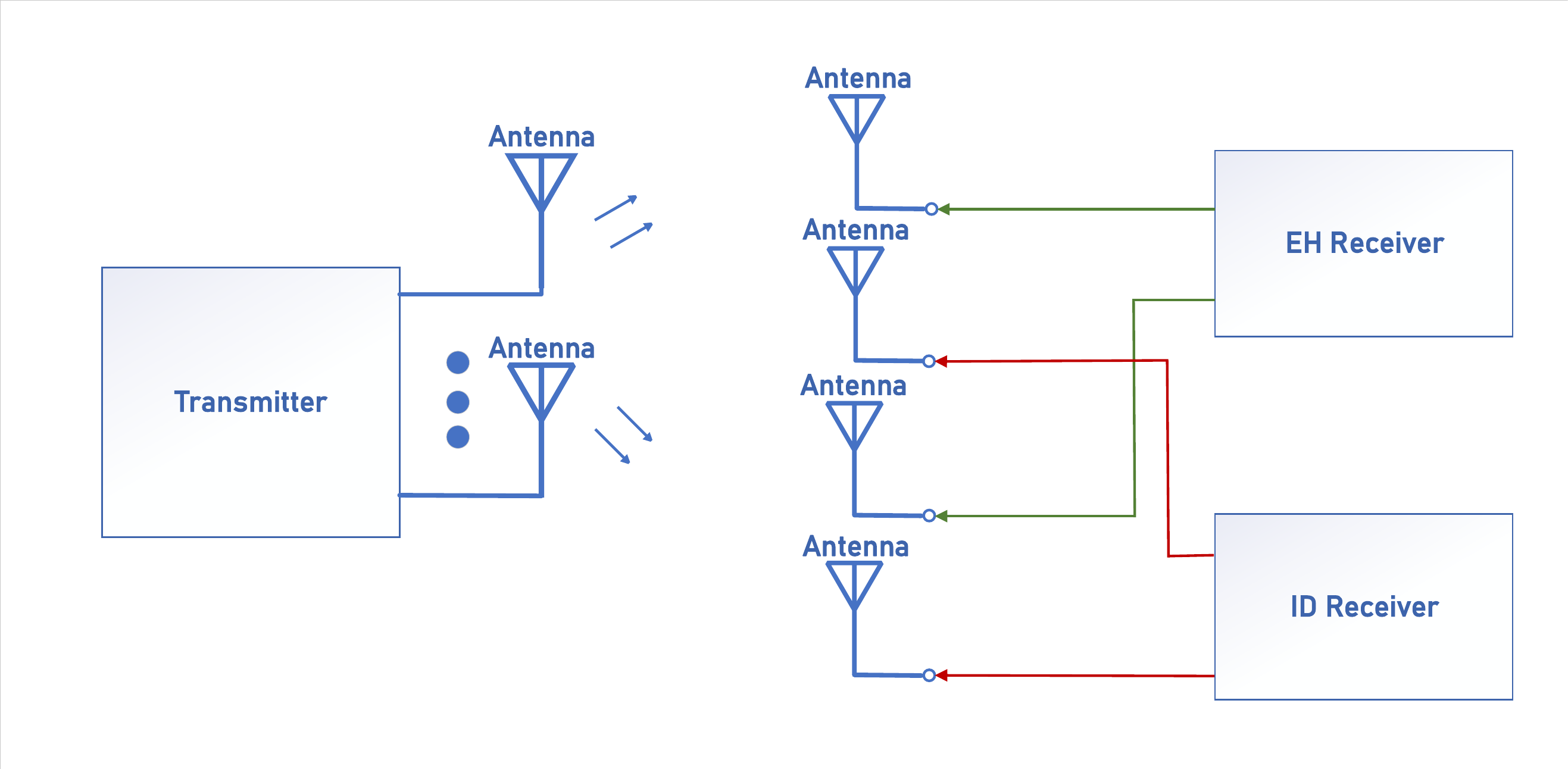}
\caption{{Antenna switching (AS) approach to realize~co-located SWIPT~architecture.}}
\label{fig:1.7}
\end{subfigure}
\caption{Integrated receiver architecture designs for SWIPT. }
\label{chap1:fig:co-located:methods}
\end{figure}

Three practical approaches to designing co-located receiver architecture for SWIPT are time switching (TS), power splitting (PS), and antenna switching (AS). 
The EH and ID receivers share the same antennas to realize the co-located receiver architecture, as shown in figure (\ref{chap1:fig:co-located:methods}). 
The receiver in the TS approach in figure (\ref{fig:1.5}) includes an EH module, an ID module, and a switch to periodically adjust the receiving antenna for particular operations.
The receiver switches between EH and ID modes based on a pre-defined, but optimizable, time factor or TS sequence. 
The TS approach necessitates careful information/energy scheduling and accurate time perception.
In the PS approach in figure (\ref{fig:1.6}), the receiver divides the received signal into two streams of different power levels for EH and ID operations based on an optimizable PS ratio.

Finally, figure (\ref{fig:1.7}) shows how the receiver is equipped with independent antennas for EH and ID operations in the antenna switching (AS) approach to enable SWIPT by means of a low complexity AS algorithm. 
In general, an antenna array is configured at the receiver end to take advantage of spatial multiplexing which divides the antennas into two subsets for EH and ID operations.
One subset of antennas operates on the EH mode while the rest executes the ID operation.
It should be stressed that the AS approach is somewhat easier and more suitable for practical SWIPT architecture designs than the TS and PS approaches \cite{6804407}. 
In addition, the AS approach can be similarly adopted to optimize the separated receiver architecture as shown in figure (\ref{fig:1.4}) \cite{6781609}.

\section{A SWIPT Literature Review}

In this section, we review some of the relevant work in literature on SWIPT, including multi-carrier SWIPT systems, SWIPT in cognitive radio networks, cooperative relaying in SWIPT networks, and multiple antenna communication in SWIPT systems.
All the application areas associated with SWIPT apply the far-field WPT technique for transferring power within communication systems.

\subsection{Multi-Carrier SWIPT Systems}
The basic idea of multi-carrier modulation is to divide the transmitted bitstream into several smaller blocks, each containing the original bitstream, that are sent over many different subcarriers. 
Under ideal propagation conditions, the subcarriers are orthogonal. 
In order to alleviate the effect of inter-symbol interference (ISI) on each subcarrier, each subcarrier should have a bandwidth lower than the channel coherence bandwidth in a multi-carrier network. 
Orthogonal frequency division multiplexing (OFDM) is one of the well-established, discrete implementations of a multi-carrier scheme for high data-rate wireless communications that is embraced in various standards such as IEEE 802.11a/g/n, IEEE 802.15 ultrawideband, WiMAX, and 3GPP-Long Term Evolution (LTE).
This stimulated us to study how SWIPT can manage high date-rate interference-free communication in multi-carrier-based systems.

In this regard, \cite{5513714} investigated SWIPT over a single-user OFDM channel in order to obtain the optimal trade-off between the achievable data-rate and the transferred power given a specific total available power. 
The authors studied a power allocation algorithm design assuming that the receiver is capable of using one received signal to simultaneously perform EH and ID operations.
A heuristic algorithm was proposed in \cite{IA} in a study of resource allocation policies for maximizing the harvested energy for a single user in an OFDM SWIPT system. 
In \cite{8548551}, the EE optimization problem for OFDM-based 5G wireless networks with SWIPT was studied,
with subcarrier and power allocation jointly optimized to maximize the system EE for single-user and multi-user cases using the Dinkelbach iterative and the Lagrange dual methods. 
The authors in \cite{chap4:6555184} investigated the sum data-rate maximization problem in a multi-user OFDM system with SWIPT capability 
that allows the user to either perform an ID operation when it is active (served by the transmitter) or to be in an EH mode when it is idle – but not simultaneously.
In order to realize SWIPT in a broadband system with transmit beamforming and PS-receiver architecture, \cite{Simultaneous_Information} presented a novel strategy for designing power control algorithms for both single-user and multi-user OFDM systems: These exploit channel diversity to concurrently enhance throughput and WPT efficiency in DL and UL directions with variable and fixed coding rates.
In \cite{18}, the weighted sum data-rate optimization problem was studied taking into account TS and PS approaches in a co-located architecture.
In \cite{18}, the optimal design for SWIPT in DL OFDM systems was investigated with users performing EH and ID operations on the same signals received from a fixed AP. 
Studying a system model very similar to that in \cite{18}, the authors of \cite{9Xu} sought to minimize the fraction of all users' outage when constrained by both minimum harvest energy and total transmit power. 

Two types of multiple access schemes – time division multiple access (TDMA) and orthogonal frequency division multiple access (OFDMA) – are studied for transmitting information in SWIPT networks in particular.
TDMA permits several users to share the same frequency channel by dividing the signal into different non-overlapping time slots. 
Since TDMA occupies the entire available bandwidth, ISI mitigation techniques are required to handle interference. 
On the other hand, OFDMA technology is a promising solution for effectively dividing the available bandwidth into orthogonal sub-channels so 
they can be flexibly allocated among existing users.
Although in OFDMA based networks, “intra-cell” interference does not exist, “inter-cell” interference arising from neighboring cells deteriorates the OFDMA networks' overall performance.
In this respect, resource allocation is highly significant in TDMA/OFDMA-based SWIPT cellular wireless networks: 
It mitigates the effect of interference with respect to limited bandwidth, stringent power constraints, and the growing demand for wireless communication services.
In \cite{18}, both TDMA and OFDMA were examined for information transmission. 
For TDMA-based information transmission, each user applies TS so that the ID operation is used throughout the user's scheduled information time slot and the EH operation is applied in all other time slots.
For the OFDMA-based transmission strategy, authors assumed that the PS approach was employed at each receiver with all subcarriers sharing the same PS ratio at each receiver.
These transmission scenarios were employed to address the problem of maximizing the weighted sum data-rate for all users by varying the power allocation in time and/or frequency and also the TS/PS ratios, subject to satisfying a maximum transmit power allowance and a minimum harvested energy constraint for each user.
Following Dinkelbach's method, the authors in \cite{8183429} investigated the secrecy EE maximization problem by proposing power allocation and PS optimization algorithms in an OFDMA SWIPT network with two users, where each user uses the PS receiver approach to enable SWIPT.
The resource allocation algorithm design for EE optimization was studied in \cite{Wireless_Information}, in which an EE maximization problem jointly optimizes subcarrier and power allocation as well as the PS ratios of PS hybrid receivers in an OFDMA system with SWIPT. 
The underlying problem in \cite{Wireless_Information}, which can be solved iteratively using the Dinkelback algorithm, was expressed as a non-convex optimization problem 
that had to take into account the maximum transmit power, the receivers' minimum amount of EH power, and the minimum data-rate requirements of both delay constrained services and the network.
Multi-carrier SWIPT is flourishing, with many exciting research directions yet to be explored.  

\subsection{SWIPT in Cognitive Radio Networks}
Radio spectrum underutilization results from conventional spectrum allocation with a cognitive radio network (CRN) introduced to enable highly reliable and efficient opportunistic spectrum access while increasing the number of wireless devices.
CRN allows secondary users to operate in the frequency bands allocated to primary users, as long as they cause no harmful uncontrollable interference to primary users \cite{hossain2009dynamic}.
OFDMA technology combined with a CRN can help allocate radio resources to secondary users more efficiently because it supports gained spectrum allocation by dividing the entire bandwidth into a set of subcarriers.
Two different types of cognitive radio based on spectrum-sensing and full-coordination are recognized in the literature~–~spectral-sensing to identify the channels in the radio frequency spectrum and full-coordination to assess all the attributes that a wireless network or node can be aware of during the communication process.
A CRN is subdivided into two main architectural designs: infrastructure-based and infrastructure-less.
In an infrastructure-based CRN, every unlicensed user transmits their data or routing parameters through a central BS, AP, or a relay, 
while in an infrastructure-less CRN, unlicensed users communicate with each other directly using the existing communication protocols (e.g., in a peer-to-peer fashion) that require no central entity.
The combination of infrastructure-based and infrastructure-less architectures is called a “hybrid” CRN.

In order to achieve both spectrum (spectral) efficiency (SE) and EE via dynamic spectrum access in CRN, secondary users in a CRN can be accompanied by the RF EH capability, as has been described very well in~\cite{Dynamic_spectrum}.
Focusing on the trade-off between spectrum sensing, data transmission, and RF energy harvesting, the study provides a detailed discussion of the dynamic channel selection problem in a multi-channel EH-based CRN.
The authors of \cite{7848925} investigated how hybrid EH cooperative spectrum sensing in heterogeneous CRNs can enable self-sustaining green communications by reducing energy cost while capitalizing on the idle spectrum.
\cite{8874991}~analyzed maximizing the throughput of cooperative CRNs with EH to obtain optimal time allocation between primary and secondary users, and balancing the trade-off between EH and packet transmission.
Secrecy outage performance for MISO CRNs with EH in order to maximize the secrecy and EE was studied in \cite{8910443}.
The same secrecy outage performance for an underlay MIMO CRN by means of transmitting antenna selection was investigated in \cite{7849037,7882641}. 
In both studies, EH from the primary transmitter drives the transmitter of secondary users in order to improve both EE and SE.
In \cite{8693979}, a joint optimization framework to distribute power among users at the secondary BS, allocate power for cooperative transmission, assign a time slot for the transmission of each user, and find the PS ratio for EH and ID receivers was explored for SWIPT-based cooperative CRNs.
The authors of \cite{7880677} tried to maximize the data-rate for underlay multi-hop CRNs using TS receivers. 
The problem of maximizing EH by taking into account an optimal resource allocation policy while also satisfying constraints of minimum data-rate, transmit power, and subcarrier was studied in \cite{Wideband_cognitive} with respect to max-min fairness in wideband CRN with SWIPT.

\subsection{Cooperative Relaying in SWIPT Networks}
Relay techniques are proposed for facilitating communication between the source and the destination \cite{Hierarchical}.
They offer a low-cost way of expanding coverage enlargements, gaining diversity, and enhancing data-rates. 
However, since the relay power supplies might be insufficient, it needs an additional, possibly green, energy source to cooperate. 
Luckily, WPT techniques have the potential to assist cooperative devices by recharging energy-limited relays to use as tokens for future cooperation. 
Recent research on using SWIPT for self-powered relaying has concentrated on two conventional cooperative relaying systems: amplify-and-forward (AF) and decode-and-forward~(DF)~\cite{8628978}.

The past few years have witnessed new research on cooperative relaying integrated with SWIPT networks.
The authors of \cite{Bidirectional_Wireless} studied a cooperative relaying SWIPT network with a bidirectional relay capable of transferring wireless power in the DL direction and relaying information by AF in the UL direction.
Two approaches were examined with PS and TS at the relay to maximize the user's data-rate. 
Energy harvesting and power consumption constraints at both the relay and the user were noted.  
In \cite{7362039}, the authors considered a cooperative PS-based SWIPT network with one source-destination pair and multiple EH DF relays that were intended to obtain a closed-form formulation of the outage probability achieved by the multi-relay cooperative protocol as well as its approximation at high SNR.
PS-based SWIPT in cooperative networks with spatially random DF relays was studied in \cite{6779694} in order to characterize the outage probability and diversity gain by applying the stochastic geometry principle.
The optimization problems of PS ratios at the relays were formulated in \cite{7419826} for both DF and AF relaying protocols in a multi-relay two-hop relay SWIPT system.
Joint relay selection and resource allocation optimization for EE maximization was studied in \cite{8954679} in order to provide a performance comparison between DF and AF relays without external power supplies in PS-enabled SWIPT.
The author of \cite{7953588} investigated whether a joint resource allocation would maximize the achievable data-rate in PS-based and AS-based SWIPT for a two-hop cooperative transmission in which a half-duplex multi-antenna relay adopts DF relaying strategy.
Relay selection strategies for both PS- and TS-based SWIPT in a cooperative AF relaying network were studied in \cite{8292374}.
The authors tried to maximize either the amount of EH at the user's end under the constraint of the minimum available data-rate in a PS-based SWIPT network or the overall user data-rate while guaranteeing a minimum EH in a TS-based SWIPT network. 
A novel relay selection and resource allocation in a two-hop relay-assisted multi-user PS-based SWIPT OFDMA network were analyzed in \cite{8668723} to optimize the subcarrier assignment, power allocation, and users' PS ratios as well as relays to ensure maximization of the system's data-rate while also satisfying minimal EH and maximal transmit power constraints. 

\subsection{Multiple Antenna Communication in SWIPT Systems}
Nowadays, new communication technologies can incorporate multiple antennas to increase data-rate through multiplexing, improve performance through diversity, and use WPT techniques to generate sufficient power for reliable communication.
These approaches require centralized or distributed antenna array deployment at the transmitter and/or receivers to achieve substantial array/capacity gains over single-input single-output (SISO) systems through spatial beamforming/multiplexing \cite{935916}. 
Multiple-input single-output (MISO) and MIMO are regarded as proper ways to boost the reliability and capacity of SWIPT-enabled communication networks.
Most work that has tried to integrate SWIPT in multiple antenna wireless networks assumes that two defined user groups need to be served, one for receiving information and the other for receiving power to recharge their power sources. 
However, there are also cases of co-located receiver architecture. 
The following paragraphs present interesting research studies conducted on SWIPT in MISO and MIMO systems to help the reader understand the challenges of multiple antenna configurations in SWIPT. 

In \cite{17}, the capacity region of a MISO broadcast channel featuring SWIPT was evaluated. 
The approach used was based on solving a sequence of weighted sum data-rate maximization problems subject to a maximum transmit power constraint for the AP, and a set of minimum harvested energy constraints for individual EH receivers. 
In this study a multi-antenna AP simultaneously delivers information and energy via RF signals to multiple single-antenna receivers in a separated receiver architecture.
The authors of \cite{8116441} studied joint power allocation and PS ratio optimization for a multi-user MISO SWIPT network with the aim of maximizing the users' minimum signal-to-interference plus noise ratio (SINR) under the maximum transmit power and the minimum energy harvesting constraints.

In \cite{8120246}, a multi-user MISO full-duplex system with PS-based SWIPT was proposed for a resource allocation policy design that jointly optimizes the PS ratios, the beamforming matrix, and the transmit power – subject to satisfying maximal SINR and harvested power constraints. 
A novel beamforming design to minimize transmission power in a multi-user MISO SWIPT system was proposed in \cite{8422097} for TS and PS receiver architectures. 
Joint transceiver optimization of beamforming vectors and TS ratios for MISO SWIPT systems was examined in \cite{8306833}.
In \cite{8481704}, the joint design of the beamforming vector and the artificial noise covariance matrix for a MISO SWIPT with multiple eavesdroppers was studied by analyzing a proportional secrecy EE maximization problem.

Two different scenarios for the MIMO broadcast system were investigated in \cite{6489506}, namely, separate and co-located receivers architecture in which all the transmitters and receivers were equipped with multiple antennas.
In the first scenario (separated receivers), the best transmission strategy for MIMO SWIPT systems was designed to attain different trade-offs between maximum information data-rate and
energy transfer in the boundary of a so-called rate-energy region. 
For the other scenario (co-located receivers), TS and PS approaches were applied to determine an outer bound for the achievable rate-energy region.
The authors of \cite{7934322} similarly described the trade-off between maximum energy transfer versus maximum information data-rate under the nonlinear EH model.
A novel power allocation algorithm for SWIPT-enabled multi-user MIMO DL systems with separated receiver architecture was proposed in \cite{6692369}, based on the block diagonalization precoding technique that completely suppresses multi-user interference while maximizing a network's data-rate.

In \cite{6777403}, the problem of antenna selection (at the transmitter) and transmit covariance matrix design was studied in the effort to jointly maximize the data-rate of a MIMO broadcast system with SWIPT, given the multidimensional trade-off between the minimum data-rate requirement of ID users and the minimum energy harvesting of EH users in a separated architecture. 
A general multi-objective optimization problem, in which data-rate and harvested power are simultaneously optimized for each user, is proposed in \cite{7581107} for a multi-user MIMO network implementing SWIPT.
A secrecy data-rate maximization problem for SWIPT in cognitive MIMO networks was studied in \cite{7132724}: 
An interference power constraint was imposed to protect the primary user while the secondary EH receiver had a minimum energy harvesting constraint.
Unlike the MIMO SWIPT systems, the weighted minimum mean squared error criterion – rather than the data-rate of the network – was investigated in \cite{7063588} for a separated SWIPT architecture.
The secure transmission issue for PS-enabled SWIPT in a multi-user MIMO system with multiple external eavesdroppers is presented in \cite{8171731}. Robust beamforming with imperfect channel state information (CSI) was designed with the maximum transmission power optimized subject to a minimum achievable secrecy data-rate and EH.
EE maximization in SWIPT-based MIMO broadcast channels for IoT applying the TS approach was studied in \cite{8233108} as a way of obtaining resource allocation policies that consider per-user minimum harvested energy constraints.

\section{Thesis Overview and Contributions}
Although studies have thus far concentrated on solving problems in wireless communication networks empowered by WPT technologies, various problems in the field remain to be answered.
In this thesis, we attempt to fill in several crucial gaps in the literature by designing, analyzing, and optimizing resource allocation problems with different objectives for SWIPT in multi-service energy-constrained wireless networks.
The overview and specific contributions of the main chapters are:

\textbf{Chapter \ref{CHAP2} {Optimization Techniques}}

In our study, we refer to several basic optimization techniques. 
Chapter \ref{CHAP2} discusses them and provides some approaches to optimization. 

\textbf{Chapter \ref{CHAP3} {SWIPT in Single Small-Cell Networks}}

In chapter \ref{CHAP3}, we provide a novel architecture for harvesting energy from an AP without needing a splitter at the receiver.
We propose a new system model in which a designated portion of the spectrum is used for ID operation while another portion is exploited for EH operation, and investigate how much performance is gained.
The main objective of this chapter is to design a resource allocation policy that maximizes the harvested energy of a multi-user DL OFDMA network with SWIPT while satisfying a minimum data-rate requirement for all users.
We then use optimization tools to obtain a locally optimal solution for the underlying problem, which is essentially non-convex.

The results of this chapter with slight improvements to the system model, as presented in the extended abstract, will be submitted to \textit{IEEE Communications Letters}:

\begin{addmargin}[2.5em]{0em}
——, “Optimal Resource Allocation for MC-NOMA in SWIPT-enabled Networks,” to be submitted to \textit{IEEE Communications Letters}.
\end{addmargin}

\textbf{Chapter \ref{CHAP4} {SWIPT in Multi-Cell Networks}}

In chapter \ref{CHAP4}, we study allocating resources to maximize the data-rate based on the separated receiver architecture in a SWIPT OFDMA multi-user multi-cell system.
We distribute the users between the inner and outer cells. 
Nearby users can harvest energy from the AP; 
users at a distance get their wireless information from the nearest AP.
Although it seems that the AP supports EH users by consuming some power, in this chapter, we discuss whether it is justifiable in light of the fact that this might cause the performance gain of a multi-user multi-cell OFDMA network to deteriorate. 
The resulting problem, which jointly optimizes the subcarrier assignment and the power allocation, is an intractable mixed-integer non-linear problem.
Because of that, we use a minorization maximization approach based on the difference of convex functions programming with a surrogate function to approximate the non-convex objective function.

\textbf{Chapter~\ref{CHAP5} {Antenna Selection Technique in SWIPT}}

Finally, in chapter~\ref{CHAP5}, we describe a new harvesting technique at the receiver that is based on the receiver antenna selection for a multi-user multi-cell SWIPT OFDMA system with a co-located architecture.
This we call a “generalized antenna switching technique”.
As a performance metric, we optimize EE, which is highly appreciated for investigating resource allocation in the next generation of wireless communication.
The underlying problem in this chapter is non-convex because we incorporate both interference and integer variables.
We relax the integer variable, and then apply the big-M formulation to make sure that relaxed variables take binary values.
After that, we use the minorization maximization approach employing a first-order Taylor approximation. 
As a last step to convexify the EE optimization problem, we apply the Dinkelback algorithm to transform the objective function into a non-fractional function.
The simulation result concludes the amount of performance gain that can be obtained through generalized antenna switching architecture.

The results of this chapter with an improvement in the system model will be submitted to \textit{IEEE JSTSP} (Special Issue on Signal Processing Advances in Wireless Transmission of Information and Power). A very simplified version of the system model without SWIPT has already been accepted by the VTC 2020 conference:

\begin{addmargin}[2.5em]{0em}
\textbf{Jalal Jalali}, Ata Khalili, and Heidi Steendam, “Antenna Selection and Resource Allocation in Downlink MISO OFDMA Femtocell Networks,” accepted for presentation at \textit{IEEE VTC} 2020.
\end{addmargin}

\begin{addmargin}[2.5em]{0em}
——, “Simultaneous Wireless Information and Power Transfer via a Joint Resource Allocation and Generalized Antenna Switching Strategy,” to be submitted to \textit{IEEE JSTSP}.
\end{addmargin}

\textbf{Chapter \ref{CHAP6} {Conclusion and Future Work}}

Chapter~\ref{CHAP6} summarizes the results and provides suggestions for future research.
\newpage\thispagestyle{empty}\mbox{}  
\newpage\thispagestyle{empty}\mbox{}  
\chapter{Optimization Techniques}
\label{CHAP2}
\fancyhf{}
\renewcommand{\headrulewidth}{2pt}
\fancyhead[LE,RO]{\thepage}
\fancyhead[RE]{\textit{ \nouppercase{\leftmark}} }
\fancyhead[LO]{\textit{ \nouppercase{\rightmark}} }
\renewcommand{\footrulewidth}{0.1pt}
\fancyfoot[CE,CO]{\nouppercase{\leftmark}}
\fancyfoot[LE,RO]{JFMJ}

\vspace{15mm}

In this chapter, we introduce some basic properties of convex functions that can be useful for better understanding this thesis \cite{boyd2004convex}.

\section{Convex Analysis}

\subsection{Definitions }
Let $f(\text{·}):\mathbb{R}^{n} \rightarrow\mathbb{R}$ be a convex function.
Then, $f(\text{·})$ is convex, if for each $\lambda \in [0,1]$, we have
\begin{equation}\label{F2:1}
f(\lambda \mathbf{x}_{1}+
(1-\lambda)\mathbf{x}_{2})\leq \lambda 
f(\mathbf{x}_{1})+
(1-\lambda)f(\mathbf{x}_{2}), 
\end{equation}
for all $\mathbf{x}_{1}$, $\mathbf{x}_{2} \in \mathbb{R}^{n}$.~Geometrically speaking, the above inequality states that the line segment between $(x_1, f(x_1))$ and $(x_2, f(x_2))$, which is the chord from $x_1$ to $x_2$, lies on top of the graph of $f(\text{·})$.~A function $f(\text{·})$ is said to be \textit{strictly convex} if (\ref{F2:1}) holds with strict inequality.~Moreover, supposing that $f(\text{·})$ is differentiable, i.e., its gradient exists, then each convex function satisfies the following inequality
\begin{equation}\label{F2:2}
f(\mathbf{{x}})\geq 
f(\tilde{\mathbf{{x}}}) + 
\nabla_{\mathbf{x}} 
f(\tilde{\mathbf{x}})^T (\mathbf{x}-\tilde{\mathbf{x}}),
\end{equation}
where $\nabla_{\mathbf{x}}$ is the gradient vector with respect to $\mathbf{x}$ at $\tilde{\mathbf{x}}$, and ${\square}^T$ is the transpose operation on ${\square}$.
The inequality (\ref{F2:2}) asserts that for a convex function, a global underestimator of the function can be easily derived via its first-order Taylor approximation.
Consequently, the first-order Taylor approximation of a convex function is always a global underestimator of the function.
This inequality additionally confirms that global information of a convex function can be obtained through its local information, i.e., its value and derivative at a point. 

\section{Duality Theorem}
We consider the following optimization problem, also known as the \textit{primal problem}, written in its general form as
\begin{align} \label{F(2-3)}
&\min_{\mathbf{x} \in \mathcal{X}}~f(\mathbf{x})
\\s.t.:~
& g_{i}(\mathbf{x})\leq 0,~ 
\forall i=1,...,I,\nonumber\\ 
& h_{l}(\mathbf{x})=0,~ 
\forall l=1,...,L,\nonumber
\end{align}
where $f(\text{·}):~\mathbb{R}^{n} \rightarrow\mathbb{R}$ is the objective function, and $\mathbf{x} \in \mathbb{R}^{n}$ is the vector of optimization variables inside the feasible set $\mathcal{X}$.
This optimization problem has $I$ inequality constraints and $L$ equality constraints.
Furthermore, we also refer to $p^*$ as the optimal value of the optimization problem in (\ref{F(2-3)}).

The Lagrangian duality of the objective function of (\ref{F(2-3)}) is given by 
\begin{equation}
\mathcal{L} (\mathbf{x},\boldsymbol{\mu},\boldsymbol{\nu}) = 
f(\mathbf{x}) + 
\sum_{i=1}^{I}\mu_{i}g_{i}(\mathbf{x}) + 
\sum_{l=1}^{L}\nu_{l}h_{l}(\mathbf{x}),
\end{equation} 
where $\boldsymbol{\mu}$ and $\boldsymbol{\nu}$ are called the vector of \textit{Lagrangian multiplier} or the \textit{dual variables} with respect to inequality and equality constraints associated with the problem (\ref{F(2-3)}) that have $\mu_{i}$'s and $\nu_{l}$'s as the elements of the corresponding vectors.
The essential purpose of Lagrangian duality is to get somehow rid of the constraints in (\ref{F(2-3)}) by adding a weighted sum of the constraint functions to the objective function.
We can now define the corresponding \textit{Lagrange dual function} (or just \textit{dual function}), which is formally stated as 
\begin{equation}\label{F2:5}
\mathcal{D}(\boldsymbol{\mu},\boldsymbol{\nu}) =
\inf_{\mathbf{x}} 
\mathcal{L} (\mathbf{x},\boldsymbol{\mu},\boldsymbol{\nu}). 
\end{equation}
Note that even though the primal problem could be non-convex, the dual problem is always a convex optimization problem since the dual function is a point-wise infimum.
This infimum can be seen as the greatest lower bound of a family of affine functions with respect to $\boldsymbol{\mu}$ and $\boldsymbol{\nu}$. 

The Lagrange dual function in (\ref{F2:5}) gives us a lower bound on the optimal value $p^*$ of the primal problem (\ref{F(2-3)}).
In order to find the best lower bound for the primal problem, the following optimization problem can be defined from the Lagrange dual function
\begin{equation}\label{F2:6}
\max_{\boldsymbol{\mu},\boldsymbol{\nu}} ~
\mathcal{D} (\boldsymbol{\mu},\boldsymbol{\nu}).
\end{equation}
This problem is known as the \textit{Lagrange dual problem} corresponding to the primal problem.
Moreover, if $\boldsymbol{\mu}^*$ and $\boldsymbol{\nu}^*$ are the optimal values for the Lagrange dual problem in (\ref{F2:6}), they are traditionally called \textit{dual optimal} or \textit{optimal Lagrange}.
It should also be noted since the objective to be maximized is concave in (\ref{F2:6}), the Lagrange dual problem is a convex optimization problem no matter the primal problem in (\ref{F(2-3)}) is convex or not. 

\subsection{Weak Duality and Duality Gap}
Let $\mathbf{x}^*$ be a feasible solution for the primal problem, i.e., $p^*$ and $(\boldsymbol{\mu}^*,\boldsymbol{\nu}^*)$ are a feasible solution to the dual problem, that is, $d^*$.
According to weak duality, we have the following inequality for a general (possibly non-convex) problem
\begin{equation}\label{F2:7}
d^* \leq p^*.
\end{equation} 
It must be noted that the weak duality inequality also holds when $d^*$ and $p^*$ are infinite.
On the other hand, the difference between the primal optimal value and dual primal value, i.e., $p^*-d^*$ is called the \textit{optimal duality gap}.
It should be stated that the optimal duality gap is always non-negative.
Since the dual problem is always convex, and often can be solved efficiently to determine $d^*$, the inequality in (\ref{F2:7}) is quite useful in finding a lower bound on the optimal value of a problem that is difficult to solve.

\subsection{Strong Duality and Slater condition}
If the duality gap is zero, i.e., $p^* = d^*$, the \textit{strong duality} holds.~The strong duality indicates that the best bound that can be achieved from the Lagrange dual function is tight.
Moreover, in strong duality, since the gap between primal and dual is zero, solving the dual problem is equivalent to solving the primal problem.

A sufficient condition for strong duality to hold for a convex optimization problem is the \textit{Slater condition} or \textit{Slater's condition}.
In particular, if the Slater condition holds for the primal problem, then the duality gap is zero, which implies strong duality for convex problems.
And if the dual optimal value is finite, then it is attained, i.e., a dual feasible $(\boldsymbol{\mu}^*,\boldsymbol{\nu}^*)$ exists that satisfies ${\mathcal{D}(\boldsymbol{\mu}^*,\boldsymbol{\nu}^*) = d^* = p^*}$.
In general, there exist so many results that establish conditions on the optimization problem, that yield strong duality.
These conditions are coined \textit{constraint qualifications}, where the Slater condition is only a simple specific example of many.

\section{Abstract Lagrangian Duality}
In order to study duality in optimization models, two approaches exist historically, and the duality results are manifested as referred to as: i) \textit{Classical Lagrangian} and  ii) \textit{Abstract Lagrangian}.
Among these two forms, the classical Lagrangian form is more extensively used in the literature.~What we have discussed so far is indeed the classical Lagrangian form of duality.
As seen, classical Lagrangian typically starts from a primal problem while the Lagrangian and the Dual Lagrangian problems are established subsequently.~However, at a more abstract level, an abstract Lagrangian function is used to derive the primal and dual optimization problems.
Here, we briefly discuss an abstract version of Lagrangian duality that is elaborated in more significant details in~\cite{goh2002duality}.
In this version, through a certain real-valued abstract Lagrangian function, the primal and dual costs are taken into account, such that
\begin{align*}
(\textrm {Primal problem})&~~ 
\min_{\mathbf{x}\in \mathcal{X}}~
\mathcal{F}(\mathbf{x})~~\textrm{where}~
\mathcal{F}(\mathbf{x})=~\sup_{\mathbf{y}\in Y}~
\mathcal{L}(\mathbf{x},\mathbf{y}),\\
(\textrm {Dual problem})~&~~
\max_{\mathbf{y}\in \mathcal{Y}}~
\mathcal{G}(\mathbf{y})~~\textrm{where}~
\mathcal{G}(\mathbf{y})=~
\inf_{\mathbf{x}\in \mathcal{X}}~
\mathcal{L}(\mathbf{x},\mathbf{y}),
\end{align*}
where $\mathcal{L}:\mathcal{X}\times \mathcal{Y} \longrightarrow R$ is the abstract Lagrangian function pertaining to $\mathcal{X}$ and $\mathcal{Y}$ as appropriate domains defined in some primal and dual spaces, respectively.
Moreover, the supremum can be seen as the least upper bound of a family of affine functions with respect to $\mathbf{x}$ and $\mathbf{y}$.
This approach to duality is based on conjugate duality, where a convexity assumption is always made~\cite{rockafellar1974conjugate}.~This approach also puts a strong emphasis on the minimax and saddle point theorems, which are given below.

$\bullet$~\textbf{Minimax Theorem}:
This theorem provides the condition that guarantees the \textit{strong max-min property} or the \textit{saddle point} as follows
\begin{equation}\label{F2:8}
\sup_{\mathbf{y}\in 
\mathcal{Y}}~\inf_{\mathbf{x}\in \mathcal{X}}~
\mathcal{H}(\mathbf{x},\mathbf{y})=
\inf_{\mathbf{x}\in \mathcal{X}}~
\sup_{\mathbf{y}\in \mathcal{Y}}~
\mathcal{H}(\mathbf{x},\mathbf{y}).
\end{equation}
It should be noted that the above equality, strong max-min property, holds only in special cases.
This is, in particular, true, when for example, $\mathcal{H}:\mathcal{X}\times \mathcal{Y} \longrightarrow R$ is the Lagrangian of a problem where the strong duality holds. 

$\bullet$~\textbf{Saddle Point Theorem}:
Under suitable conditions, there exists a \textit{saddle point} for $\mathcal{S}(\text{·})$ referred to as a pair $(\mathbf{x}^{*},\mathbf{y}^{*})\in \mathcal{X}\times \mathcal{Y}$ such that for all~$(\mathbf{x},\mathbf{y})\in \mathcal{X}\times \mathcal{Y}$,
\begin{equation}\label{F2-9}
\mathcal{S}(\mathbf{x}^{*},\mathbf{y})\leq
\mathcal{S}(\mathbf{x}^{*},\mathbf{y}^{*})\leq
\mathcal{S}(\mathbf{x},\mathbf{y}^{*}).
\end{equation}
In (\ref{F2-9}), $\mathcal{S}:\mathcal{X}\times \mathcal{Y} \longrightarrow R$ is the Lagrangian of a problem where the strong duality holds.
In other words, $\mathcal{S}(\textbf{x}^{*},\textbf{y}^{*}) = \sup_{\textbf{y}\in \mathcal{Y}} ~\mathcal{S}(\textbf{x},\textbf{y}^*)$, and $\mathcal{S}(\textbf{x}^{*},\textbf{y}^{*}) = \inf_{\textbf{x}\in \mathcal{X}}~\mathcal{S}(\textbf{x}^*,\textbf{y})$.
This indicates that the strong max-min property (\ref{F2:8}) holds with the common value of $\mathcal{S}(\textbf{x}^{*},\textbf{y}^{*})$.

\section{Complementary Slackness and KKT Optimality Conditions}
Suppose that both the primal and dual optimal values exist and are equal. This means the strong duality holds.
We further assume that $\mathbf{x^*}$ and $(\boldsymbol{\mu}^*,\boldsymbol{\nu}^*)$ to be a primal optimal and a dual optimal point, respectively. Therefore, we have
\begin{align}\label{F2:10}
f(\mathbf{x^*})=
\mathcal{D}(\boldsymbol{\mu}^*,\boldsymbol{\nu}^*) \leq
f(\mathbf{x^*}) + 
\sum_{i=1}^{I}\mu^*_{i}g_{i}(\mathbf{x^*}) +
\sum_{l=1}^{L}\nu^*_{l}h_{l}(\mathbf{x^*}) \leq 
f(\mathbf{x^*}).
\end{align}
The first inequality in (\ref{F2:10}) holds since the infimum of the Lagrangian over $\mathbf{x}$ is less than or equal to its value at $\mathbf{x} = \mathbf{x^*}$.
However, the last inequality follows from 
{${\mu_{i}^* \geq 0,~g_{i}(\mathbf{x^*})\leq0 ,~ \forall i=1,...,I}$}, 
and 
{${h_{l}(\mathbf{x^*})\leq0,~ \forall l=1,...,L}$}.
An important conclusion that one can make from (\ref{F2:10}) is that
\begin{equation}
\mu^*_{i}g_{i}(\mathbf{x^*})=0,~~~~
\forall i=1,...,I.
\end{equation}
This condition is called the \textit{complementary slackness}.~It confirms that one can go from the optimal primal solution to the optimal dual solution, and vice versa, if the strong duality holds.
Moreover, the complementary slackness verifies that a solution is optimal, by checking if there is a dual solution.

Now, we introduce the Karush-Kuhn-Tucker (KKT) conditions assuming that all the functions both in the objective and the constrains in (\ref{F(2-3)}) are differentiable.
Just same as was assumed in (\ref{F2:10}), let's also suppose the primal and dual variables at the optimum point, for which strong duality obtains, are $\mathbf{x}^{*}$ and $(\boldsymbol{\mu}^{*},\boldsymbol{\nu}^{*})$, respectively.
The KKT conditions have the following properties
\begin{subequations}
\begin{align}
g_{i}(\mathbf{x}^{*})&\leq 0,~ 
\forall i=1,...,I,\\ 
h_{l}(\mathbf{x}^{*})&=0,~ 
\forall l=1,...,L,\\
\mu_{i}^{*}&\geq 0, ~ 
\forall i=1,...,I,\\
\mu_{i}^{*}g_{i}(\mathbf{x}^{*})&=0,~ 
\forall i=1,...,I,\\
\nabla_{\mathbf{x}}f(\mathbf{x}^{*}) + 
\sum_{i=1}^{I}\mu_{i}^{*} \nabla_{\mathbf{x}}g_{i}(\mathbf{x}^{*}) +
\sum_{l=1}^{L}\nu_{l}^{*} \nabla_{\mathbf{x}}h_{l}(\mathbf{x}^{*})&=0,\label{2-12e}
\end{align} 
\end{subequations}
where $\mu_{i}^{*}$ and $\nu_{l}^{*}$ are the elements of Lagrangian vectors $\boldsymbol{\mu}^{*}$ and $\boldsymbol{\nu}^{*}$,~respectively.
Also, $\nabla_{\textbf{x}}$ denotes the gradient of a function with respect to \textbf{x} in (\ref{2-12e}).
Note that the KKT conditions are necessary and sufficient conditions for the optimality of the convex optimization problem with differentiable objective and constraint functions. 
However, if the problem is non-convex, the KKT conditions would only provide the necessary conditions for optimality given that the objective and constraints are differentiable.

\section{Interior-Point Methods}
The literature on interior-point methods is very extensive, 
and research is still flourishing. 
This paragraph can only serve as a very condensed introduction.
Interior-point methods can be seen as a branch in the classification of convex optimization algorithms that solve linear and nonlinear convex optimization problems.
KKT conditions, another branch of optimization algorithms in this classification, obtain a collection of linear equations 
that can be solved analytically, for example a quadratic optimization problem with linear equality constraints.
Interior-point methods, on the other hand, solve an optimization problem with linear equality and inequality constraints by the relaxation of a problem with only a set of linear equality constraints.
The motivation for calling such methods an “interior-point method” lies in the fact that 
these methods begin their search for an optimal solution in the interior of the feasible region, and travel on a path towards the boundary, converging at the optimum.

\section{MM Approach and D.C. Programming}
MM algorithms are an appropriate tool to reduce a given optimization problem into a series of simpler problems.
In this sense, an MM algorithm is not an algorithm, but rather an appropriate principal way of designing optimization algorithms for high dimensional settings, where 
the classical methods of optimization do not work well.
MM algorithms are not new.
The celebrated Expectation Maximization algorithm is a particular case of MM algorithms that is extensively used in electrical engineering applications and in other fields.
The reason for selecting the MM acronym is two-fold.
An MM algorithm operates on a more straightforward and simpler surrogate function that majorizes/minorizes (the first M of MM) the objective function in a minimization/maximization (the second M of MM) optimization problem.
Thus, the MM stands for either \textit{Majorization Minimization} or \textit{Minorization Maximization}, depending on the application.
In the next few paragraphs, we consider a majorization minimization problem to explain how the algorithm works.  

Consider the following optimization problem
\begin{align}
\min_{\textbf{x}\in \mathcal{X}}~&f(\textbf{x}),
\end{align}
where $\textbf{x}$ is the optimization variable vector belonging to the feasible set ${\mathcal {X}}$.
In order to majorize the function $f(\textbf{x})$ at $\textbf{x}^n$, there exists a surrogate function $g(\textbf{x}|\textbf{x}^n)$ that satisfies two conditions
\begin{align}
f(\textbf{x}^n) &=     
g(\textbf{x}^n|\textbf{x}^n), \label{F2:14}\\
f(\textbf{x})   & \leq 
g(\textbf{x}|\textbf{x}^n), ~~~\textbf{x} \neq \textbf{x}^n. \label{F2:15}
\end{align}
The first condition (\ref{F2:14}) is called the tangency condition at the current iteration step.
This condition grantees $g(\textbf{x}^n|\textbf{x}^n)$ is tangent to $f(\textbf{x})$ at $\textbf{x}^{n}$.~The second condition, on the other hand, (\ref{F2:15}) makes sure the $g(\textbf{x}|\textbf{x}^n)$ is dominant in a sense that it always lies above the surface of $f(\textbf{x})$ except at $\textbf{x}^{n}$.
Besides, if a function $g(\textbf{x}|\textbf{x}^n)$ majorizes the function $f(\textbf{x})$ at $\textbf{x}^{n}$, it can be easily perceived that -$g(\textbf{x}|\textbf{x}^n)$ minorizes -$f(\textbf{x})$.

Another very important result of the MM algorithms is the descent property.
Starting from $\textbf{x}^{0}\in \mathcal{X}$ as an initial point for the feasible set ${\mathcal {X}}$, an MM algorithm generates a sequence of feasible point $\textbf{x}^{n}$.
At point $\textbf{x}^{n}$ in the majorization step, a continuous surrogate function is constructed that satisfies the domination condition in (\ref{F2:15})
\begin{equation}\label{F2:16}
g(\textbf{x}|\textbf{x}^{n})\geq 
f(\textbf{x}) + 
g(\textbf{x}^{n}|\textbf{x}^{n})-
f(\textbf{x}^{n}),~\textbf{x} \neq 
\textbf{x}^{n}.
\end{equation}
Hence, in the minimization step, the following update rule can be applied
\begin{equation}\label{F2:17}
\textbf{x}^{n+1}\in \min_{\textbf{x}\in \mathcal{X}} g(\textbf{x}|\textbf{x}^{n}).
\end{equation}
It is easy to show that the generated sequence $f(\textbf{x}^{n})$ is non-increasing. 
Thus, we have
\begin{equation}\label{F2:18}
f(\textbf{x}^{n+1}) \leq 
g(\textbf{x}^{n+1}|\textbf{x}^{n}) - g(\textbf{x}^{n}|\textbf{x}^{n}) + 
f(\textbf{x}^{n}) \leq 
g(\textbf{x}^{n}|\textbf{x}^{n})-
g(\textbf{x}^{n}|\textbf{x}^{n}) + 
f(\textbf{x}^{n}) = f(\textbf{x}^{n}),
\end{equation}
where the first inequality comes from (\ref{F2:16}), and the second inequality is the direct consequence of (\ref{F2:17}).
The property in (\ref{F2:18}), the descent property, gives a remarkable numerical stability to MM algorithms.
Hence, instead of minimizing the cost function $f(\textbf{x})$ directly, the MM algorithms stably optimize a sequence of tractable approximate surrogate objective functions $g(\textbf{x}|\textbf{x}^{n})$ that minorize $f(\textbf{x})$ as tightly as possible.

The MM algorithms can easily be connected to other algorithmic frameworks \cite{MM,Sequential_convex,yuille2003concave}.
One of the application areas of the MM algorithms is in \textit{Difference of Convex functions} (D.C.) programming problems. The general form of D.C. functions is
\begin{align}\label{F2:19}
&\min_{\textbf{x}}~  
f_{0}(\textbf{x})- 
h_{0}(\textbf{x})\\ 
\textit{s.t.:}~ 
&f_{i}(\textbf{x})- h_{i}(\textbf{x})\leq 0, ~\forall i=1,...,m,
\end{align}
where $f_{i}$'s and $h_{i}$'s are all convex functions.~We further assume that $f_{i}$'s and $h_{i}$'s are twice differentiable, and are strictly convex without loss of generality according to (\ref{F2:1}).
Among various algorithms having desirable properties for the solution of D.C. problems, the MM scheme, which solves a sequence of convex problems acquired by linearizing non-convex parts in the objective function as well as the constraints, is preferred.
Accordingly, an approximate solution can be found that iteratively solves (\ref{F2:19}) through defining the following convex subproblem
\begin{align}
&\min_{\textbf{x}}~ 
g_{0}(\textbf{x}|\textbf{x}^{n})\\
\textit{s.t.:}~ 
&g_{i}(\textbf{x}|\textbf{x}^{n})\leq 0, ~\forall i=1,...,m,
\end{align}
where 
\begin{equation}
g_{i}(\textbf{x}|\textbf{x}^{n})=
f_{i}(\textbf{x})-
\Big(
h_{i}(\textbf{x}^{n}) + 
\nabla_{\textbf{x}} h_{i}
(\textbf{x}^{n})^{T}(\textbf{x}-\textbf{x}^{n})
\Big),~~~
\forall i \in \{0,...,m\}.
\end{equation}

The aforementioned approximation satisfies the MM principle and is a tight upper bound of $f_{i}-h_{i}$ with equality achieved at $\textbf{x}=\textbf{x}^{n}$. 
This technique is used several times throughout the thesis.
Moreover, the solution methodology for the MM algorithm is summarized in \text{\textbf{Algorithm~\ref{Alg:chap2}}}. 
A valid question to be asked at this point would be how good the convergence behaviors of the MM algorithms are.
For the answer, the interested reader is referred to \cite{SCA2014MR,General_inner,Local_convergence}. 

\begin{algorithm*}[t]
\caption{The MM Approach}
\begin{algorithmic}[1]
\STATE {$\mathbf{Initialize}$} \\{
\begin{addmargin}[1em]{0em}
{iteration index $n=0$ with the maximum number of iteration $N_{\max}$  \\
and find a feasible point $\mathbf{x}^{0}$.}
\end{addmargin}}
\STATE \textbf{repeat}
\STATE{
\begin{addmargin}[1em]{0em}
Find $\mathbf{x}^{n}$ by solving the optimization problem (\ref{F2:17}) and store as $\mathbf{x}$.
\end{addmargin}}
\STATE{
\begin{addmargin}[1em]{0em}
Set $n=n+1$ and $\mathbf{x}^{n}=\mathbf{x}$.
\end{addmargin}}
\STATE \textbf{until} some convergence criterion is met or $n=N_{\max}$
\STATE \textbf{return} optimal $\mathbf{x}$
\end{algorithmic}
\label{Alg:chap2}
\end{algorithm*}

\section{Optimization Packages}
Many optimization tools and packages exist for solving any given optimization problem.
This paragraph introduces a few of the most popular optimization packages.
The GLPK is a package designated for solving large-scale linear programming (LP), mixed integer programming (MIP), and other similar problems.
The Gurobi optimizer is a well-known commercial optimization solver for LP, quadratic programming (QP), and MIPs (including mixed-integer linear programming (MILP), mixed-integer quadratic programming (MIQP), and mixed-integer quadratically constrained programming (MIQCP)).
The Mosek solver is another widely used optimization package that solves LP, QP, MIP, second-order cone programming (SOCP), and semi-definite programming (SDP).
The SeDuMi and SDPT3 are two other leading solvers for SDPs.
The list of packages goes on and on, with new packages always being added and the well-known ones continuously being updated to respond to the ever-rising demand for higher performance speed to solve problems with unimaginably large dimensions.
The purpose of the thesis is not to develop such optimization solvers nor to improve them, but rather to use the existing ones to solve the problems at hand.
A whole different research domain investigates and designs optimization solvers for specific needs.
\newpage\thispagestyle{empty}\mbox{}  
\newpage\thispagestyle{empty}\mbox{}  
\chapter{SWIPT in Single Small-Cell Networks}
\label{CHAP3}
\fancyhf{}
\renewcommand{\headrulewidth}{2pt}
\fancyhead[LE,RO]{\thepage}
\fancyhead[RE]{\textit{ \nouppercase{\leftmark}} }
\fancyhead[LO]{\textit{ \nouppercase{\rightmark}} }
\renewcommand{\footrulewidth}{0.1pt}
\fancyfoot[CE,CO]{\nouppercase{\leftmark}}
\fancyfoot[LE,RO]{JFMJ}

\vspace{15mm}

As explained in the introductory chapter, wireless power transfer (WPT) provides wireless devices with continuous and stable energy.
Radio frequency (RF)-enabled simultaneous wireless information and power transfer (SWIPT) has the capability of using an innovative way to simultaneously transfer information over a single radio waveform as power.
Consequently, SWIPT-based networks contribute exceptional benefit to users by conveniently utilizing radio signals to transfer both energy and information.
SWIPT is attracting ever more attention due to its ability to provide green communication services more efficiently by broadcasting information and power on orthogonal and non-orthogonal resources.
Furthermore, resource allocation for SWIPT-enabled networks carefully takes into account the performance of both information and power transfer, unlike traditional wireless systems.
In this sense, resource allocation in SWIPT systems can improve network performance and make maximum use of network resources while satisfying the quality of service requirement by flexibly allocating and dynamically adjusting the network's available resources.

A great deal of research has been conducted on resource allocation for SWIPT in various types of wireless communication networks.
In \cite{Simultaneous_Information} for instance, the focus was on resource allocation based on the orthogonal frequency division multiplexing (OFDM) for different system configurations in a SWIPT-enabled mobile network 
that maximizes throughput by using separated receiver architecture.
The optimal design for maximizing the weighted sum data-rate over all users for SWIPT in downlink (DL) OFDM systems was explored in \cite{18}.
In it, users harvest energy and decode information using the same receiving signals 
that carry both energy and information from a fixed access point (AP). 
Each user also applies either power splitting (PS) or time switching (TS) receiver architecture to coordinate the energy harvesting (EH) and information decoding (ID) operations.
The authors in \cite{chap3:8276564} investigated resource allocation for maximizing the effective capacity and effective energy efficiency (EE) of a DL multi-user OFDM system by considering both PS and TS architectures for receivers using SWIPT. 
In \cite{Wireless_Information}, the ideal resource allocation algorithm design to maximize EE of data transmission was studied by considering PS hybrid receivers in an orthogonal frequency division multiple access (OFDMA) network.
Most of the previous works had investigated a SWIPT system based on the splitter using PS or TS architecture to separate the received signal for either EH or ID operation. 
However, \cite{IA} considered a SWIPT-assisted joint subcarrier and power allocation based on an OFDM system using a heuristic solution methodology in which neither time nor power splitter was used to maximize EH.
In this chapter, we aim to fill the knowledge gap in existing literature by improving the paradigm of resource allocation policy in SWIPT networks without the need for a splitter in a co-located architecture.

We study a simple scenario often found in the literature that captures the essential characteristics of SWIPT-enabled networks in this chapter: the optimal design for resource allocation in a multi-user OFDMA network with SWIPT using the same signals received from a fixed AP to perform both EH and ID operations.
The complexity of the receiver is significantly reduced because there is no need for a splitter to perform appropriately: The receiver has no time or power splitter.
The resource allocator only needs to know which group of subcarriers is allocated for EH and which for ID operation.
This information is derived from the channel state information (CSI) and the relevant algorithm.
We do not address the effect of interference in this chapter;
instead we investigate the problem of resource allocation design in a single small-cell network, 
which allows us to address the problem of maximizing the energy harvested by all users.
More specifically, the problem of joint power allocation and subcarrier assignment is investigated in order to maximize harvested energy while fulfilling the minimum data-rate requirement for each user.
Since each subcarrier can be configured independently in OFDMA systems, information and power are transferred separately on different subcarriers with different waveforms.
We thus try to determine a resource allocation policy that includes joint power allocation and subcarrier assignment algorithms based on an OFDMA network for a DL of a single small-cell multi-user system by means of SWIPT.
The underlying optimization problem calls for mixed-integer linear programming (MILP).
To tackle the problem, we employ the majorization minimization (MM) optimization approach to obtain a close-to-optimal resource allocation policy.
Simulation results demonstrate that our proposed algorithm achieves excellent performance as compared to other possible solutions presented in the literature.

\section{System Model}
We consider a DL of an OFDMA network in a single small-cell scenario consisting of an indoor AP and $K$ co-located users as shown in figure~(\ref{fig:3.1}).
In particular, users receive the intended signal from the AP for EH and ID operations, simultaneously.
Furthermore, the AP and all the receivers are equipped with a single antenna in this configuration.
We additionally assume that the entire frequency band of $\mathscr{B}$ is partitioned into $N$ subcarriers, each having a bandwidth of $\mathscr{W}$.
It also needs to be stated that a portion of the spectrum is used for ID while the remaining portion is exploited for EH proclaiming a demand for the use of two separate filters at receivers~\cite{4623916}.
Moreover, the set of users and subcarriers are denoted by $\mathcal{K}=\{1,2,...,K\}$ and $\mathcal{N}=\{1,2,...,Z,Z+1,...,N\}$, respectively.
Let $\mathcal{N}_i=\{1,2,,...,Z\}$ denote the set of subcarriers for ID, whereas the remaining subcarriers {$\mathcal{N}_{e}=\mathcal{N}-\mathcal{N}_{i}=\{Z+1,Z+2,...,N\}$} is used to indicate the set of subcarriers for EH. Note that the optimal value of $Z$, i.e. the cardinally of the set $\mathcal{N}_i$, can be obtained \cite{IA}. However, we choose not to adapt the problem of finding the optimal $Z$ in our proposed problem design for the sake of avoiding repetition. 
\begin{figure}[h]
\centering
\includegraphics[width=17cm,trim=4 4 4 4,clip]{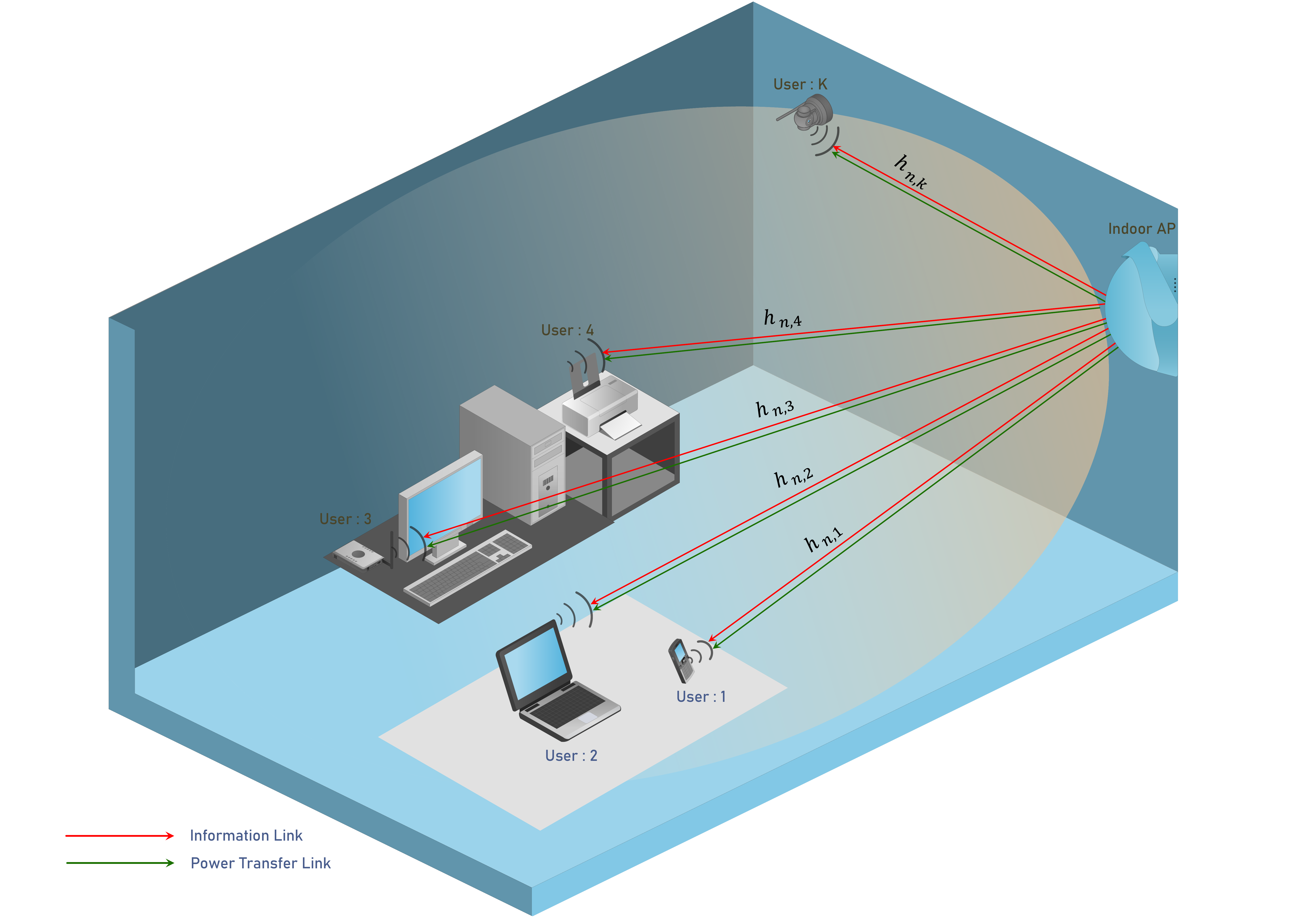}
\caption{SWIPT in a DL of a co-located multi-user small single-cell OFDMA network.}
\label{fig:3.1}
\end{figure}

Moreover, it is further assumed that all the subcarriers are considered to be perfectly orthogonal to one another, and no inter-subcarrier interference exists.
Hence, the subcarrier assignment variable is given by
\begin{equation*}
a_{n,k} = 
\begin{cases}
1, & \text{if subcarrier $n$ is assigned to user \textit{k}}, \\
0, & \text{otherwise}.
\end{cases}
\end{equation*}
Let $h_{n,k}$ denote the DL channel coefficient from the AP to the $k^{th}$ user over the subcarrier $n$.
We assume that the perfect CSI is available at a centralized resource allocator to design resource allocation policy.
Specifically, it is presumed that the AP broadcasts orthogonal preambles, pilot signals, in the DL to the users.
Then, through a feedback channel, each user estimates the CSI and transfers this information back to the AP.
Afterward, the corresponding AP listens to the sounding reference signals communicated by the users and sends the CSI to the centralized controller for resource allocation design.
Now, by denoting $p_{n,k}$ as the DL transmit power of AP to the $k^{th}$ user over the subcarrier $n$, the DL signal-to-noise ratio (SNR) of user $k$ in subcarrier $n$ is defined as
\begin{equation}
\label{F3-1}
\gamma_{n,k}=
\frac{a_{n,k}p_{n,k}|h_{n,k}|^2}
{\sigma^{2}_{n,k}},   
\end{equation}
where $\sigma^{2}_{n,k}$ denotes the additive noise power.~More specifically, at receiver $k$, the received signal on each subcarrier is corrupted by noise $n_{n,k}$.
This noise is modeled as an additive white Gaussian noise (AWGN) random variable with zero mean and variance $\sigma^{2}_{n,k}$, denoted by a circularly symmetric Gaussian distribution referred to as $n_{n,k} \sim  \mathcal{CN}(0,\sigma^{2}_{n,k})$.
However, for the sake of simplicity, we consider $\sigma^{2}_{n,k}$ = $\sigma^{2}$ throughout this chapter meaning that the variance of the noise is the same over all subcarriers for all users.
According to the Shannon capacity formula, the data-rate of the $k^{th}$ user over the subcarrier $n$ can be expressed as
\begin{equation}\label{F3-2}
R_{n,k}=
\log_2
\bigg(1+\gamma_{n,k}\bigg).
\end{equation}
For facilitating the presentation, we denote $\textbf{p} \in \mathbb{R}^{1\times KN}$ and $\textbf{a} \in \mathbb{Z}^{1\times KN}$ as vectors of optimization problem for power allocation and subcarrier assignment, respectively.
Consequently, the data-rate of the $k^{th}$ user in DL is given as
\begin{equation} \label{F3-3}
R_{k}(\textbf{a},\textbf{p})=
\sum_{n\in\mathcal{N}_i}R_{n,k}.
\end{equation}
Furthermore, to guarantee the quality of service (QoS) of users, a minimum data-rate denoted by $R_{min}$, should be provided for each user. That is
\begin{equation} \label{F3-5}
R_{k}(\textbf{a},\textbf{p})\geq R_{min}, ~
\forall k \in \mathcal{K}.	
\end{equation}
Moreover, the amount of the harvested energy can be stated as 
\begin{equation} \label{F3-4}
\textrm{EH}(\textbf{a},\textbf{p})=
\sum_{k\in\mathcal{K}}
\epsilon_k
\bigg(
\sum_{n\in\mathcal{N}_e}
a_{n,k}p_{n,k}|h_{n,k}|^{2}
\bigg),
\end{equation}
where $\epsilon_k$ is the power efficiency of the $k^{th}$ user capable of energy harvesting, which takes its value from the interval $0 < \epsilon_k < 1$.
It should be noted that $\sigma^{2}$ also contributes to the EH formula in (\ref{F3-4}).
However, since its value is very small, it is neglected in the representation of the EH formula.

\section{Optimization Problem Formulation}
The main objective persuaded in this chapter is to assign subcarrier(s) and set the transmit power(s), for each user, such that the total harvested energy in (\ref{F3-4}) is maximized.
Thus, we can formulate the optimization problem as
\vspace{-5mm}
\begin{subequations}
\begin{align}
&\max_{\textbf{a},\textbf{p}} ~ \textrm{EH}(\textbf{a},\textbf{p})
\label{F3-6}
\\
s.t.: 
&~C_{1}:
\sum_{k \in \mathcal{K}} 
a_{n,k}\leq 1,  ~~~~~~ 
\forall  n \in \mathcal{N}_i, 
\label{F3-7}
\\ 
&~C_{2}:
\sum_{k \in \mathcal{K}} 
a_{n,k}\leq 1,  ~~~~~~ 
\forall  n \in \mathcal{N}_e, 
\label{F3-7-1}
\\ 
&~C_{3}:
\sum_{k\in \mathcal{K}}~
\sum_{n\in \mathcal{N}}
a_{n,k}p_{n,k}\leq
p_{max} ,
\label{F3-8}
\\
&~C_{4}:
R_{k}(\textbf{a},\textbf{p})\geq R_{min}, ~
\forall k \in \mathcal{K}, 
\label{F3-9}
\\
&~C_{5}:
a_{n,k}\in\{0,1\} , ~~~~~~ 
\forall k \in \mathcal{K} ,~ 
\forall n \in \mathcal{N} .
\label{F3-10}
\end{align}
\label{F3-6new}%
\end{subequations}
In this optimization problem, $C_{1}$ and $C_{2}$ indicate that each subcarrier can be assigned to at most one user for ID and EH operation, respectively.
$C_{3}$ states that power constraint for the AP with a maximum transmit power allowance of $p_{max}$.
The forth constraint, $C_{4}$, guarantees the QoS for each user.~Finally, by keeping in mind that the subcarrier assignment variable is binary, the last constraint, $C_{5}$, makes sure that different subcarriers take their values from a binary set.
This means whether a given subcarrier is going to be selected to maximize the energy harvesting, i.e., the objective of the optimization problem at hand.

One can readily conclude that the optimization problem in~(\ref{F3-6new}) is a non-convex mixed-integer linear programming (MILP) problem \cite{MINLP} due to the binary constraint for the subcarrier assignment in $C_{5}$ and the non-linearity of the QoS constraint.
In general, it is impossible to find an optimal solution for a non-convex MILP in a polynomial-time.
However, in the next section, we exploit an approach to find a locally optimal solution for the considered system.
Furthermore, we propose a suboptimal resource allocation algorithm, which has a polynomial-time computational complexity to strike a balance between complexity and system performance.
	
\section{Solution to the Optimization Problem}  
In this section, we try to solve the problem in (\ref{F3-6new}) using the MM approach by constructing a sequence of surrogate functions to approximate the non-convex problem.
The MM procedure in our setting, i.e., the problem (\ref{F3-6new}), consists of two major steps.
In the first minorization step, a surrogate function is found that locally approximates the transformed objective function with their difference maximized at the current point.
That is to say, the transformed objective function is lower-bounded by the surrogate function up to a constant value.
Next, the surrogate function is maximized in the maximization step.
By inspiring this approach, we make a convex approximation which can be globally and efficiently solved using optimization packages incorporating the MM algorithm.
Nonetheless, finding a proper surrogate function that yields a low-complexity algorithm is not a straightforward task.
On the one hand, a surrogate function that attempts to mimic the form of the objective function and obtains a fast convergence speed is preferable.
On the other hand, the surrogate function should be easy to maximize in such a way that the computational cost per iteration remains low.
Realizing the appropriate trade-off between the above-mentioned conflicting goals necessitates experiences in applying inequalities to particular problems, as becomes evident in the next subsection.
For more details please refer to~\cite{MM,SCA,ASM,SCA2,SCA3} and the references therein.

\subsection{Joint Power Allocation and Subcarrier Assignment Algorithm}
In this subsection, we propose a locally optimal solution for the optimization problem in (\ref{F3-6new}).
It is worth mentioning that the multiplication of two variables in the objective function of the problem in (\ref{F3-6new}) as well as the constrain $C_4$ are the obstacles for the design of a computationally efficient resource allocation algorithm.
Since the multiplication of two variables in (\ref{F3-8}), i.e., $a_{n,k}p_{n,k}$, is non-convex, we define the product terms as $\tilde{p}_{n,k}=a_{n,k}p_{n,k}$.
In order to handle this difficulty of the non-convexity, we adopt the big-M formulation \cite{big_M} to decouple the product terms.
Therefore, the following additional constraints are imposed accordingly as 
\begin{align}
&C_{6}:
\tilde{p}_{n,k}\leq p_{max}a_{n,k},~~~~~~~~~~~~~~~~~ 
\forall k \in \mathcal{K},~ 
\forall n \in \mathcal{N}, 
\\ 
&C_{7}:
\tilde{p}_{n,k}\leq p_{n,k},~~~~~~~~~~~~~~~~~~~~~~~ 
\forall k \in \mathcal{K},~ 
\forall n \in \mathcal{N}, 
\\ 
&C_{8}:
\tilde{p}_{n,k}\geq  p_{n,k}-(1-a_{n,k})p_{max}, ~ 
\forall k \in \mathcal{K},~ 
\forall n \in \mathcal{N}, 
\\ 
&C_{9}:
\tilde{p}_{n,k}\geq  0,~~~~~~~~~~~~~~~~~~~~~~~~~~~ 
\forall k \in \mathcal{K},~ 
\forall n \in \mathcal{N},  
\end{align}
where $\tilde{\textbf{p}}\in \mathbb{R}^{1\times KN}$ is the collection of all $\tilde{p}_{n,k}$'s.~Therefore, the original optimization problem in (\ref{F3-6new}) can be recast in equivalent form as
\begin{subequations}
\begin{align}
&\max_{{\textbf{a}},\textbf{p},\tilde{\textbf{p}}} 
~\widehat{\overline{\textrm{EH}}}(\textbf{a},\textbf{p},\tilde{\textbf{p}})
\\
s.t.: &~C_{1}-C_2,C_5-C_9, 
\\ &~C_{3}:
\sum_{k\in \mathcal{K}}~
\sum_{n\in \mathcal{N}}
\tilde{p}_{n,k}\leq p_{max},
\label{F3-22}
\\ &~C_{4}:
\sum_{ n \in \mathcal{N}}\log_2(1+\frac{\tilde{p}_{n,k}|h_{n,k}|^2}{\sigma^{2}})\geq   R_{min},~ \forall k \in \mathcal{K},
\end{align}
\label{F3-new20}%
\end{subequations}
where, by revisiting the definition of the harvested energy in (\ref{F3-4}), the objective function of (\ref{F3-new20}) is as follows
\begin{equation}
\widehat{\overline{\textrm{EH}}}(\textbf{a},\textbf{p},\tilde{\textbf{p}})=
\sum_{k\in\mathcal{K}}
\epsilon_k\bigg(
\sum_{n\in\mathcal{N}_e}
\tilde{p}_{n,k}|h_{n,k}|^{2}\bigg).
\end{equation}
The optimization problem of (\ref{F3-new20}) is still a non-convex MILP problem, which is complicated to solve.
To facilitate the solution design, we restate the integer variable in constraint $C_{5}$ as the intersection of the following regions \cite{TWC_Ata,che2014joint}
\begin{align} 
&\dot{C}_{5}:
0\leq a_{n,k}\leq 1,~
\forall k \in \mathcal{K},~ 
\forall n \in \mathcal{N}, 
\\
&\ddot{C}_{5}:
\sum_{k \in \mathcal{K}}
\sum_{n \in \mathcal{N}}
a_{n,k}-(a_{n,k})^{2}\leq 0.
\end{align}
The above transformations make integer optimization variables continuous with values between zero and one.
Therefore, the original problem in (\ref{F3-new20}) can be rewritten as 
\begin{align}
&\max_{{\textbf{a}},\textbf{p},\tilde{\textbf{p}}} ~ 
\widehat{\overline{\textrm{EH}}}(\textbf{a},\textbf{p},\tilde{\textbf{p}})
\label{F3-30}
\\
s.t.: &~C_1-C_4,\dot{C}_{5},\ddot{C}_{5},C_6-C_9. 
\nonumber
\end{align}
Note that the optimization problem in (\ref{F3-30}) is a continuous optimization problem with respect to all variables.
However, our concern is to obtain integer solutions for $a_{n,k}$.
To this end, we add a penalty term to the objective function. 
The integer variable is now relaxed to take any values between zero to one.
Thereby, the problem can be restated as follows
\begin{align}
&\max_{\textbf{a},\textbf{p},\tilde{\textbf{p}}}
\mathcal{L}(\textbf{a},\textbf{p},\tilde{\textbf{p}},\lambda)
\label{F3-31}
\\
s.t.: &~C_1-C_4,\dot{C}_{5},C_6-C_9, 
\nonumber
\end{align}
where $	\mathcal{L}(\textbf{a},\textbf{p},\tilde{\textbf{p}},\lambda)$ is the \textit{abstract Lagrangian duality} \cite{rockafellar1974conjugate} associated to (\ref{F3-30}), and is defined as  
\begin{align}
\label{F3-32}	
\mathcal{L}(\textbf{a},\textbf{p},\tilde{\textbf{p}},\lambda)=
\widehat{\overline{\textrm{EH}}}(\textbf{a},\textbf{p},\tilde{\textbf{p}})-
\lambda\bigg(\sum_{k \in \mathcal{K}}\sum_{n\in \mathcal{N}}a_{n,k}-(a_{n,k})^{2}\bigg).
\end{align}
In (\ref{F3-32}), $\lambda$ acts as a penalty factor to penalize the objective function when $a_{n,k}$ is not an integer variable.
	
\begin{proposition}\label{weak_duality}
In the abstract Lagrangian of (\ref{F3-31}),~the $\lambda$ acts as a penalty factor to penalize the objective function when the subcarrier allocation variable $a_{n,k}$ does not have a binary value.
For sufficiently large values of $\lambda$, the optimization problem in (\ref{F3-31}) is equivalent to (\ref{F3-30}) with both problems yielding the same optimal results.
\end{proposition}
\vspace{-2mm}
\begin{proof}
By using the abstract Lagrangian duality, the proof of~\textbf{Proposition~\ref{weak_duality}} is presented.
Accordingly, the primal and dual problem of (\ref{F3-30}) can be written respectively as
\vspace{-2mm}
\begin{align}
(\textrm {Primal problem})&
~~ 
p^{*}=
\max_{\textbf{a},\textbf{p},\tilde{\textbf{p}}}~
\min_{\lambda}
\mathcal{L}(\textbf{a},\textbf{p},\tilde{\textbf{p}},\lambda),\\
(\textrm {Dual problem})~&
~~
d^{*}=
\min_{\lambda}~\max_{\textbf{a},\textbf{p},\tilde{\textbf{p}}}
\mathcal{L}(\textbf{a},\textbf{p},\tilde{\textbf{p}},\lambda).
\end{align}
Now, by defining $\mathcal{G}(\lambda)\triangleq\max\limits_{\textbf{a},\textbf{p},\tilde{\textbf{p}}} \mathcal{L}(\textbf{a},\textbf{p},\tilde{\textbf{p}},\lambda)$,~the following inequality holds according to the weak duality theorem 
\begin{eqnarray}\label{F3-36}
p^*= 
\max_{\textbf{a},\textbf{p},\tilde{\textbf{p}}}\min_{\lambda} 
\mathcal{L}(\textbf{a},\textbf{p},\tilde{\textbf{p}},\lambda) 
\leq 
\min_{\lambda\geq 0}\mathcal{G}(\lambda)=
d^*.
\end{eqnarray}
Additionally, it should be pointed out that for $\textbf{a},\textbf{p},\tilde{\textbf{p}}~s.t.:C_1-C_4,\dot{C}_{5},\ddot{C}_{5},C_6-C_9$, two cases can be observed.

\textbf{\quad\emph{Case~1}}: In the first case, it is assumed that at the optimal subcarrier allocation point, we have
\begin{equation}
\ddot{C}_{5}:~
\sum_{k \in \mathcal{K}}
\sum_{n\in \mathcal{N}}
a_{n,k}-(a_{n,k})^{2}= 0.
\end{equation}
Consequently, $d^*$ becomes a feasible solution of (\ref{F3-30}).
Afterward, substituting the optimal value of $\lambda$, i.e., $\lambda^{*}$, into the abstract Lagrangian duality problem of (\ref{F3-31}) results in 
\begin{equation}
d^*=
\mathcal{G}(\lambda^*) =
\max_{\textbf{a},\textbf{p},\tilde{\textbf{p}}}
\textrm{EH}(\textbf{a},\textbf{p},\tilde{\textbf{p}})=
p^*.
\end{equation}
Furthermore, considering the optimization problem in (\ref{F3-31}) and assuring that the subcarrier allocation variable, $a_{n,k}$, takes its values in the region \textbf{a}$~\in \dot{C}_5,\ddot{C}_5$, one can conclude that $\mathcal{G}(\lambda)$ is a monotonically decreasing function with respect to $\lambda$.
On the other hand, it can be asserted that 
$d^*$=$\min_{\lambda\geq 0} \mathcal{G}(\lambda)$.
Therefore, we have
\begin{equation}
\mathcal{G}(\lambda)= d^*, 
\forall ~\lambda\geq \lambda_{0},
\end{equation}
where $\lambda_0$ is a given value.
This confirms that for any value of $\lambda_{}\geq\lambda^*$, the solution of (\ref{F3-31}) can return the optimal solution of (\ref{F3-30}).

\textbf{\quad\emph{Case~2}}: In the second case, it is assumed that the subcarrier allocation variable, $a_{n,k}$, takes some value from the interval between zero and one.
This is equivalent to satisfying the following inequality 
\begin{equation}
\ddot{C}_{5}:\sum_{k\in \mathcal{K}}\sum_{n\in \mathcal{N}}a_{n,k}-(a_{n,k})^{2}> 0.
\end{equation}
In this regard, one may conclude $\mathcal{G}(\lambda^*)$ tends to $-\infty$ at the optimal point by referring to (\ref{F3-31}) and (\ref{F3-36}).~Nevertheless, this cannot happen.
Since the primal solution must always be greater than zero, this would be a contradiction as it states that $\mathcal{G}(\lambda^*)$ would be limited from below by the solution of (\ref{F3-30}).
Hence, $\sum_{k \in \mathcal{K}}\sum_{n\in \mathcal{N}}a_{n,k}-(a_{n,k})^{2}= 0$.~Therefore, the solution of (\ref{F3-31}) yields the optimal solution of (\ref{F3-30}).
\end{proof}
\vspace{-5mm}
Now, we can express the optimization problem in (\ref{F3-30}) in terms of difference of convex functions (D.C.) as follows
\begin{align}\label{F3-41}
&\max_{\textbf{a},\textbf{p},\tilde{\textbf{p}}}~
\widehat{\overline{\textrm{EH}}}
(\textbf{a},\textbf{p},\tilde{\textbf{p}})-
\lambda\big(\mathcal{V}(\textbf{a})-\mathcal{W}(\textbf{a})\big)
\\
s.t.: &~C_1-C_4,\dot{C}_{5},C_6-C_9,
\nonumber
\end{align}
where
\vspace{-5mm}
\begin{align}
\mathcal{V}(\textbf{a})&=
\sum_{k\in\mathcal{K}}
\sum_{n\in\mathcal{N}}
a_{n,k}, \\
\mathcal{W}(\textbf{a})&=
\sum_{k\in\mathcal{K}}
\sum_{n\in\mathcal{N}}
(a_{n,k})^2,
\end{align}
are all convex functions, and $\lambda$ is the penalty factors to penalize the objective function when $a_{n,k}$ is not integer values.~Although all terms in the objective function of (\ref{F3-41}) are convex, the subtraction of two convex functions are not necessarily convex
\cite{yuille2003concave,che2014joint,DC}.~To make a convex approximation for the objective function, we adopt the MM algorithm by constructing a surrogate function via a first-order Taylor approximation as
\begin{align}\label{F3-42}
\mathcal{W}(\textbf{a})
\simeq 
\mathcal{W}(\textbf{a}^{t-1})+
\nabla_{\textbf{a}}\mathcal{W}
(\textbf{a}^{t-1})^{T}.(\textbf{a}-\textbf{a}^{t-1})
\triangleq 
\tilde{\mathcal{W}}(\textbf{a}), 
\end{align}
where $t$ denotes the iteration number, the $\textbf{a}^{t-1}$ is the solution of the problem at $(t-1)^{th}$ iteration, and $\nabla_{\square}$ represents the gradient with respect to ${\square}$.~Approximation (\ref{F3-42}), satisfies the MM criteria and is a tight upper bound of $\mathcal{W}(\textbf{a})$ \cite{MM,DC}.~Therefore, using the MM approach while constructing a sequence of surrogate functions at the $t^{th}$ iteration, we can solve the following convex problem instead of dealing with the non-convex optimization problem in (\ref{F3-41}). Thus, we have 
\begin{align}\label{F3-43}
&\max_{\textbf{a},\textbf{p},\tilde{\textbf{p}}}~
\widehat{\overline{\textrm{EH}}}
(\textbf{a},\textbf{p},\tilde{\textbf{p}})-
\lambda\big(
\mathcal{V}(\textbf{a})-
\tilde{\mathcal{W}}(\textbf{a})
\big)\\
s.t.: 
&~C_1-C_4,\dot{C}_{5},C_6-C_9.
\nonumber
\end{align}
It is easy to demonstrate that the optimization problem~(\ref{F3-43}) is convex and can be solved efficiently via D.C. approximation based on the interior point methods~\cite{MM}.~As a consequence, the solution of (\ref{F3-43}) would be an approximation to the solution of the original problem given~in~(\ref{F3-41}). However, in D.C. programming, the iteration begins from a feasible initial point and solves the optimization problem iteratively until it eventually approaches a close-to-optimal solution~\cite{che2014joint,DC,CL_Ata}.~Besides, it must be noticed that the MM approach produces a sequence of improved feasible solutions with the adopted D.C. approximation, which would ultimately converge to a locally optimal solution $(\textbf{a}^{*},\textbf{p}^*,\tilde{\textbf{p}}^*)$ using standard convex program solvers such as CVX.

\begin{proposition}
By incorporating D.C. approximation, the solution of the problem~(\ref{F3-43}) {\text{becomes}} a tightly lower-bounded solution from below for the original problem (\ref{F3-41}) at the end of each {\text{iteration.}}
\end{proposition}
\vspace{-2mm}
\begin{proof}
In the $t^{th}$ iteration, the objective function of (\ref{F3-41}) is 
\begin{equation*}
\widehat{\overline{\textrm{EH}}}
(\textbf{a}^t,\textbf{p}^t,\tilde{\textbf{p}}^t)-
\lambda\big(\mathcal{V}(\textbf{a}^t)-
\mathcal{W}(\textbf{a}^t)\big).    
\end{equation*}
Subsequently, in the next iteration, we have
\begin{align}
\widehat{\overline{\textrm{EH}}}
(\textbf{a}^{t+1},&\textbf{p}^{t+1},\tilde{\textbf{p}}^{t+1})-
\lambda\big(\mathcal{V}(\textbf{a}^{t+1})-
\mathcal{W}(\textbf{a}^{t+1})\big) 
\nonumber\\&
\geq
\widehat{\overline{\textrm{EH}}}
(\textbf{a}^{t+1},\textbf{p}^{t+1},\tilde{\textbf{p}}^{t+1})-
\lambda\big(\mathcal{V}(\textbf{a}^{t+1})-
\mathcal{W}(\textbf{a}^{t})\big)+
\lambda\nabla_{\textbf{a}}\mathcal{W}
(\textbf{a}^{t})^{T}.(\textbf{a}-\textbf{a}^{t})
\nonumber\\
\nonumber&=	
\max_{{\textbf{a}},{\textbf{p}}, \tilde{\textbf{p}}}
\widehat{\overline{\textrm{EH}}}
(\textbf{a},\textbf{p},\tilde{\textbf{p}})-
\lambda\big(\mathcal{V}(\textbf{a})-
\mathcal{W}(\textbf{a}^{t})-
\nabla_{\textbf{a}}\mathcal{W}
(\textbf{a}^{t})^{T}.(\textbf{a}-\textbf{a}^{t})\big)
\\
&\geq
\widehat{\overline{\textrm{EH}}}(\textbf{a},\textbf{p},\tilde{\textbf{p}})-
\lambda\big(\mathcal{V}(\textbf{a})-
\mathcal{W}(\textbf{a}^{t})-
\nabla_{\textbf{a}}\mathcal{W}
(\textbf{a}^{t})^{T}.(\textbf{a}^{t}-\textbf{a}^{t})\big)
\nonumber\\
\nonumber
& =
\widehat{\overline{\textrm{EH}}}
(\textbf{a}^{t},\textbf{p}^{t},\tilde{\textbf{p}}^{t})-
\lambda\big(\mathcal{V}(\textbf{a}^{t})-
\mathcal{W}(\textbf{a}^{t})\big).
\end{align}
This completes the proof. 
\end{proof}
\vspace{-5mm}
One can readily verify that the objective function of (\ref{F3-43}) takes larger values as the iteration continues.
Hence, we adopt an iterative solution to tighten the obtained upper bound based on the \textbf{Algorithm~\ref{alg2D.C.}}. 

\begin{algorithm}[H]
\caption{Proposed Iterative Method via D.C. Programming Based on the MM Approach}
\label{alg2D.C.}
\begin{algorithmic}[1]
\STATE {$\mathbf{Initialize}$} \\{
\begin{addmargin}[1em]{0em}
{MM iteration index $t=0$ with maximum number of MM iteration $T_{max}$, \\
feasible set vector $\mathbf{a}^{0}$, $\mathbf{p}^{0}$, and $\tilde{\mathbf{p}}^{{0}}$},\\
and the penalty factor $\lambda\gg1$.
\end{addmargin}}
\STATE{\textbf{repeat} }
\STATE{
\begin{addmargin}[1em]{0em}
Update $\tilde{\mathcal{W}}(\textbf{a})$ based on (\ref{F3-42}). 
\end{addmargin}}
\STATE{
\begin{addmargin}[1em]{0em}
Solve optimization problem of (\ref{F3-43}) and store the intermediate resource allocation policy $\textbf{a}^{t}$, $\textbf{p}^{t}$, and $\tilde{\textbf{p}}^t$.
\end{addmargin}}
\STATE{
\begin{addmargin}[1em]{0em}
Set $t=t+1$.
\end{addmargin}}
\STATE{
\begin{addmargin}[1em]{0em}
Set \{$\mathbf{a}^t,\mathbf{p}^t$,$\tilde{\mathbf{p}}^t$\}= \{$\mathbf{a},\mathbf{p}$,$\tilde{\mathbf{p}}$\}.
\end{addmargin}}
\STATE \textbf{until} Convergence or $t=T_{max}$
\STATE \textbf{return} \{$\mathbf{a}^{*},{\mathbf{p}}^{*}$,$\tilde{\mathbf{p}}^*$\} $=$ \{$\mathbf{a}^{t},{\mathbf{p}}^{t}$,$\tilde{\mathbf{p}}^t$\}
\end{algorithmic}
\end{algorithm}
		
\section{Complexity Analysis}
As can be observed, the joint optimization problem in (\ref{F3-43}) involves $KN$ variables and $N+K+5NK$ linear constraints.
Consequently, it can be concluded that the overall computational complexity of optimization problem is $\mathcal{O}(NK)^{2}(N+K+5NK)$.
This is asymptotically equal to $\mathcal{O}(NK)^{3}$, exhibiting a polynomial-time complexity.

\section{Simulation Results}
\begin{table}[t]
\centering
\caption{Simulation Parameters}
\label{chap3:Simulation_Parameters}
\begin{tabular}{|c|c|}\hline
{\bf Parameter} & {\bf Value} \\ \hline \hline
{Cell coverage ($d_{max}$)} & {$20$ m}\\
{Reference distance ($d_{0}$)}&{$5$ m}\\
{The number of user ($K$)} & {$4$}\\
{The number of subcarriers ($N$)} & {$16$} \\
{Noise power ($\sigma^2$}) & {$-120$} dBm \\
{The bandwidth of each subcarrier} & {$180$} kHz \\
{Path loss exponent ($\alpha$)} & {$2.76$} \\
{Path loss model for cellular links} & {$31.7+27.6\log(\frac{d}{d_0})$} \\
{Multi-path fading distribution} & {Rician fading with factor $3$ dB} \\
{Power conversion efficiency ($\epsilon$)}  & {$30\%$} \\
{The maximum transmit power of the AP ($p_{max}$)} & {$30$ dBm} \\
{The minimum data-rate requirement for the $k^{th}$ user (${R}_\textnormal{min}$)} & {$1$ bps$/$Hz} \\
{Channel realization number} & $100$\\
\hline 
\end{tabular}
\end{table}	
In this section, the performance gain of the proposed joint subcarrier assignment and power allocation algorithm for SWIPT in the DL direction of a single-cell multi-user OFDMA system is evaluated through extensive simulations.
The radius of the cell, $d_{max}$, is 20 meters, with a reference distance, $d_{0}$, of 5 meters.
Moreover, there are $K = 4$ uniformly and randomly located users between, $d_{0}$, the reference distance and maximum coverage of the small-cell, $d_{max}$.
We further assume a frequency-selective fading channel and consider the central carrier frequency is set to 3 GHz with the bandwidth of each subcarrier being {\text{180 kHz}}.
The number of subcarriers is $N=16$, where the optimal set cardinality of subcarriers for ID and EH is determined based on~\cite{IA}.
The variance of the background noise at the receiver is equal to ${\sigma_{n,k}^2=\sigma^2=-120}$ dBm throughout the simulations.~Since a line-of-sight (LoS) signal is expected in the received signal, the small-scale fading channel is modeled as Rician fading with Rician factor $\rho=3$ dB.
Besides, the Rician flat fading channel gains include a distance-dependent path loss model of $31.7+10 \alpha \log(\frac{d}{d_0})$ [dB] (where $d$ is the distance between the transmitter and the receiver) and a log-normal shadowing component with~$8$ dB standard deviation where the path loss exponent is equal to $\alpha = 2.8$~\cite{339880}.
These parameters for propagation modeling and simulations follow the suggestions in 3GPP evaluation methodology~\cite{chap3:3GPP}.~The power conversion efficiency of all users, $\epsilon_k$, is assumed to be the same and is equal to  $\epsilon_k = \epsilon = 0.3 $.~The target transmission rate {\text{$R_{min}=1$ bit/second/Hz (bps/Hz)}} unless otherwise stated.~Moreover, we conduct Monte Carlo simulations by generating random realizations of the channel gains to obtain the average harvested energy of the network.
Finally, the setting summarized in \textbf{Table}~(\ref{chap3:Simulation_Parameters}) are used unless otherwise specified.

\subsection{Total Harvested Energy versus the Maximum Transmit Power}
Figure~(\ref{plot:3.1}) shows the total harvested energy versus the maximum transmit power $p_{max}$.
As can be observed, the average total harvested energy grows monotonically as the maximum transmit power increases.
Besides, harvested energy for the lower values of the maximum transmit power is low as compared to higher values. This is due to the inability of the AP to contribute to energy harvesting as it is forced to ensure the QoS requirements.
Consequently, for the higher values of the maximum transmit power, yet with the same data-rate requirement, the AP can help users to harvest more energy since fewer subcarriers are assigned to ID and more to EH.
Hence, excess subcarriers are utilized to provide more harvested energy. 
\begin{figure}[!b]
\centering
\includegraphics[width=12cm]{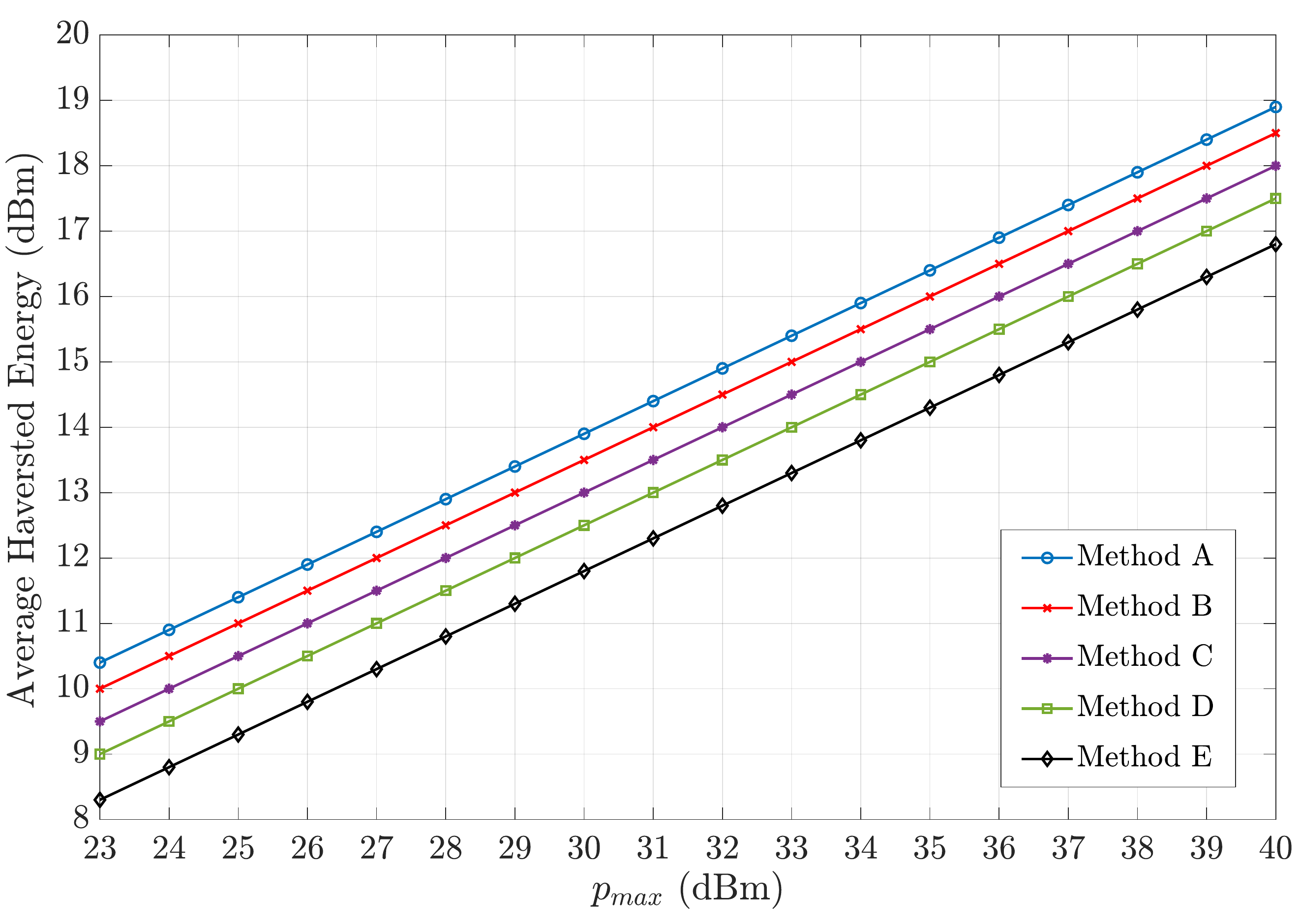}
\caption{Total harvested energy versus the maximum transmit power}
\label{plot:3.1}
\end{figure}

For comparison, the performance of the proposed joint optimization algorithm, Method~A, is compared with the following benchmark algorithms as of four Methods~B-E.
Method B is the proposed method in \cite{IA}.~In this method, a subset of subcarriers is assigned to EH to maximize the harvesting energy based on the resource allocation design.
On the other hand, the remaining ones are allocated for ID with an imposed QoS requirement.
Method~C examines the proposed method in~\cite{18} in which each subcarrier set is divided into two groups.
Specifically, one group is employed for EH, and the other one is utilized for ID while considering fixed PS ratios.
Method~D is our proposed method with equal power allocation.
This method is based on the proposed algorithm for the subcarrier assignment, while equal power is allocated across subcarrier for each user.
Method~E is our proposed method with an equal power allocation alongside a random subcarrier assignment. This method considers equal power, where the subcarrier assignment is done randomly for meeting the data-rate requirement.
Once the minimum data-rate requirement is satisfied, the remaining subcarriers are assigned to EH.

It can be seen that our proposed method outperforms the other benchmark algorithms due to the joint optimization framework.
As can be seen, Method B performs better than Method~C.
This is because Method B considers the power allocation and subcarrier assignment based on a heuristic search, while in Method~C, each subcarrier is split into two streams with a fixed PS ratio.
In particular, in Method C, the receiver does not adapt to the channel condition of each subcarrier, and each subcarrier splits the same ratio of power for ID and EH.
This results in the degradation of the performance gain and less harvested energy comparatively.
Another interesting observation is that considering a joint subcarrier assignment and power allocation in the optimization problem leads to a remarkable enhancement of the performance gain due to mitigating the harmful effect of deep fades experienced in the wireless channel as the resources allocator carefully takes this phenomenon under consideration.
It should be noted that the reason the Method~E performs worse than the rest is that the subcarriers are assigned randomly together with an equal power assignment to users.
Although the performance can improve if the subcarriers are assigned according to the designed resource allocation policy as it is the case in Method~D.

\subsection{Average Harvested Energy versus Distance}
Figure~\ref{plot:3.2} depicts the harvested energy versus distance between transmitter and receiver.
It is observed that as the distance increases, the harvested energy decreases.
The reason for this is that by increasing the distance, the channel strength would become weak, and subsequently, more subcarriers with more power are needed to be assigned to meet the minimum required data-rate.
Hence, less energy would be harvested by the users since the AP first needs to consume power to secure the quality of experience provisioning.
Once the minimum data-rate is met, the rest of the subcarriers would be exploited for EH. It should be noted that the Methods A-E are the same as defined in the last subsection.
The explanations and superiority of Method A also stay the same and is because of our joint optimization network resource allocation framework.
\begin{figure}[!h]
\centering
\includegraphics[width=12cm]{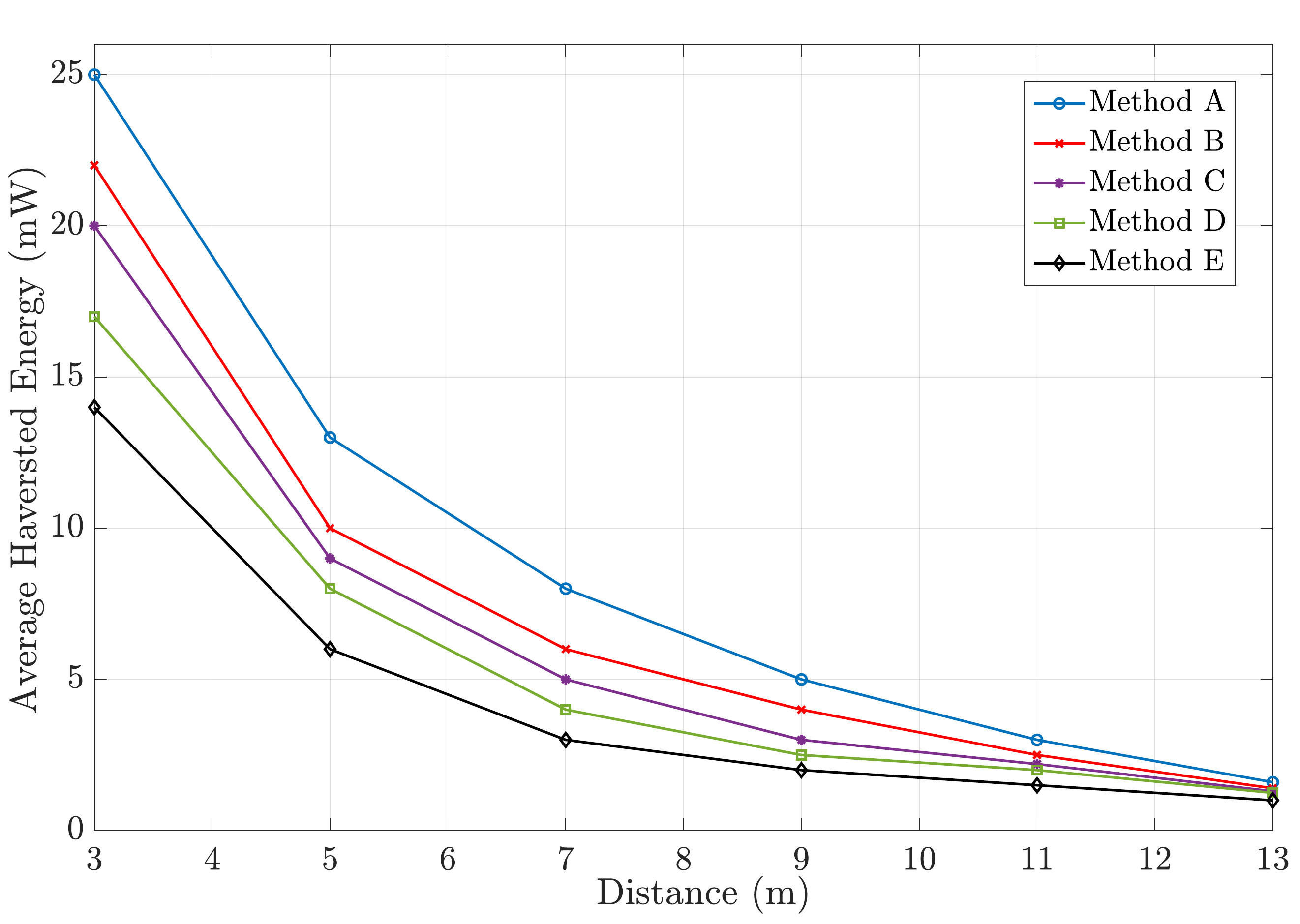}
\caption{Average harvested energy versus distance.}
\label{plot:3.2}
\end{figure}

\subsection{Average Harvested Energy versus Number of Iterations}
Figure~(\ref{plot:3.3}) illustrates the convergence of the proposed algorithm for different initialization for the power allocation.
It can be observed that the average harvested energy rises up gradually as the number of iterations increases and tends to a stationary point no matter how the power is allocated among users over all subcarriers initially.
This is because a more accurate, locally optimal solution can be found by increasing the number of iterations.
Likewise, this is equivalent to saying that the algorithm is getting closer to satisfying convergence requirements.
In the figure, we compare three initial selections for the power control; 
$\mathbf{p}^0(i) = \frac{p_{max}}{N}$, i.e., equal power is assigned to all users over all subcarriers, be it to an EH or an ID operation, $\mathbf{p}^0(i) = p_{max}$, i.e., 
all subcarriers have the maximum power, and $\mathbf{p}^0(i) = 0$, no power is assigned to different subcarriers to start with.
Note that the last two power assignments do not satisfy the constraint ${C_{3}}$ in (\ref{F3-22}).
Moreover, even though these curves are obtained for different initial power allocations, they all converge to almost the same value.
However, the convergence rate differs significantly from the different power allocation settings: while the initial setting $\mathbf{p}^0(i) = 0$ converges the slowest to the optimal harvested energy, the setting $\mathbf{p}^0(i)=\frac{p_{max}}{N}$, which satisfies the constraints of the optimization problem, converges the fastest, i.e., less than nine iterations are required.\\
\begin{figure}[!t]
\centering
\includegraphics[width=12cm]{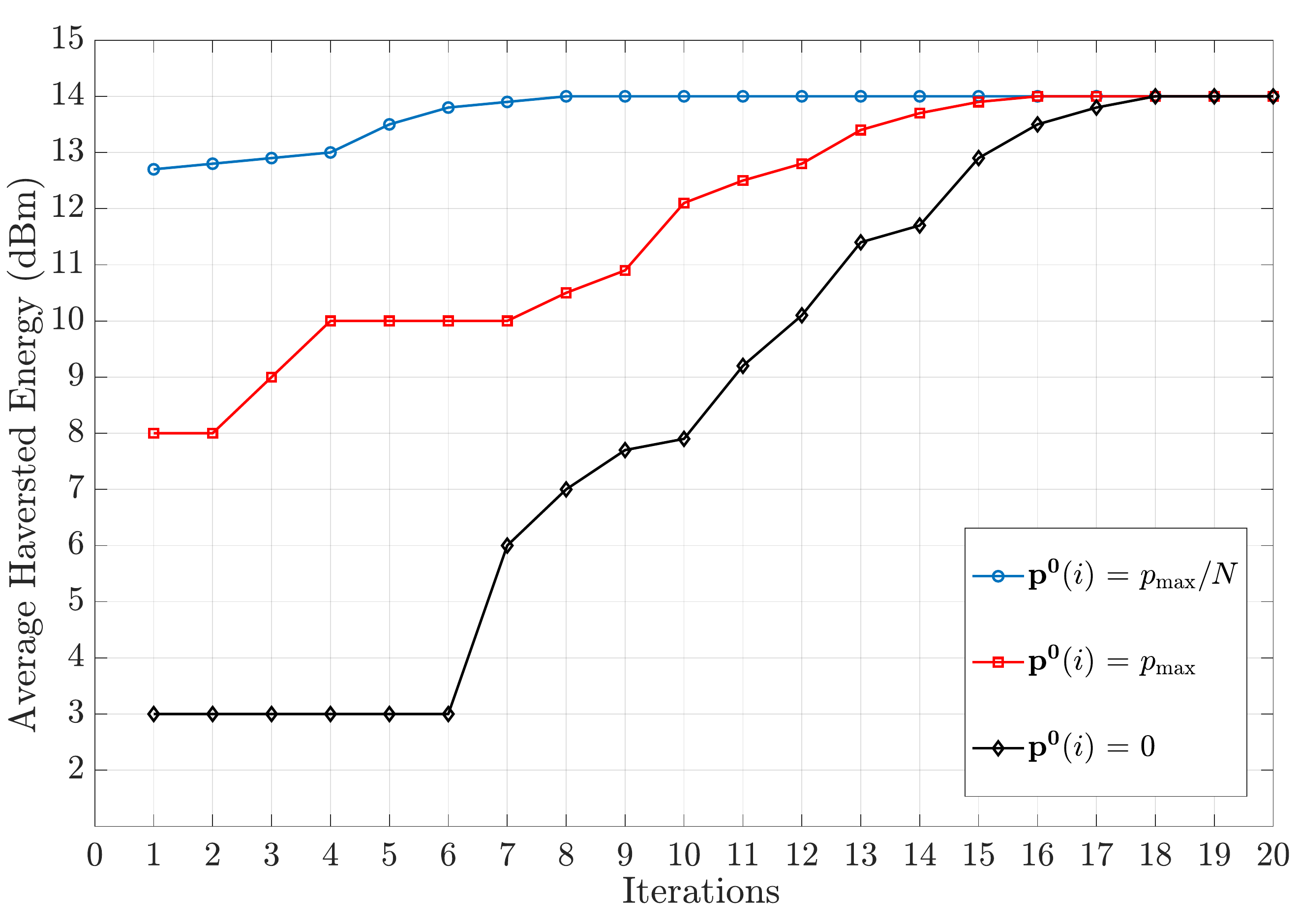}
\caption{Average harvested energy versus number of iterations.}
\label{plot:3.3}
\end{figure}

\begin{figure}[!t]
\centering
\includegraphics[width=12cm]{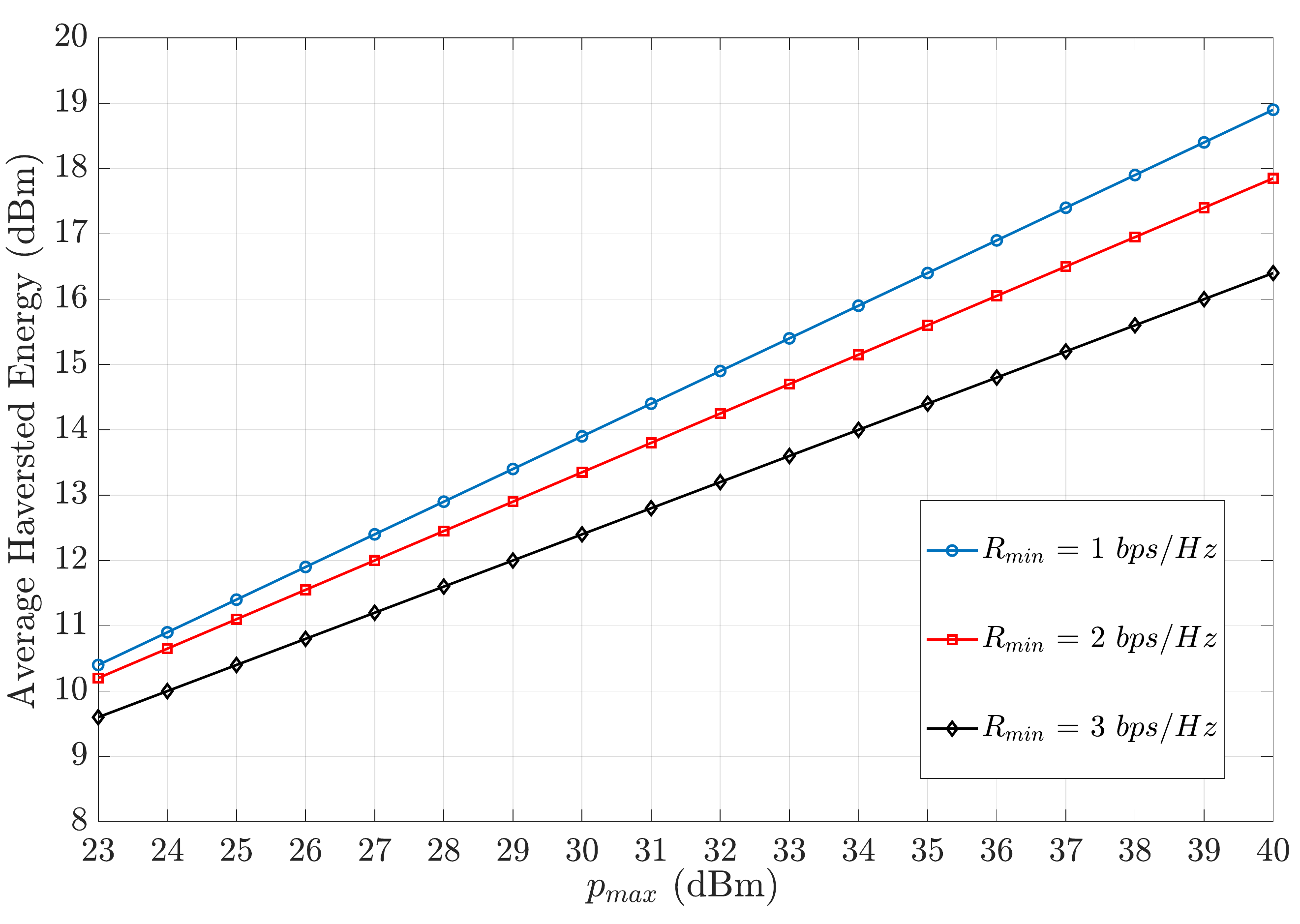}
\caption{Harvested energy versus sum transmit power at different target data-rates.}
\label{plot:3.4}
\end{figure}
\subsection{Harvested Energy versus Maximum Transmit Power at Different Target Data-rates}
Figure (\ref{plot:3.4}) illustrates the harvested energy versus the maximum transmit power, $p_{max}$, for different data-rate requirements.
It can be seen that less energy would be harvested by increasing the minimum data-rate requirement, $R_{min}$.
The reason behind it is quite evident.
Considering a fixed maximum transmit power at the AP, more subcarriers with more power would be required to be assigned to each user when the minimum data-rate requirement becomes larger.
This means that fewer subcarriers would be assigned to EH operation with lower allocated power.
This, in turn, causes less energy to be harvested by the users over the EH subcarriers as the minimum data-rate requirement increases.

\section{Summary}
In this chapter, we investigated the problem of joint subcarrier assignment and power control for SWIPT in a multi-user single small-cell network to maximize the harvesting energy while respecting the minimum required data-rate for each user.
We reduced the complexity of the receiver in our proposed algorithm as the receiver does not need a splitter to perform appropriately; that is, neither time nor power splitter is utilized at the receiver.
The only knowledge that the resource allocator needs to know is which group of subcarriers is allocated for EH and which group for ID operation, where this knowledge is derived based on the CSI and the designed resource allocation policy.
The problem considered in this chapter was a non-convex MILP problem, which is generally difficult to solve.
In order to circumvent this difficulty, we relaxed the integer variables and employed the MM approach to convexify the problem.
Afterward, the problem was solved in an iterative manner to obtain a locally optimal solution.
Simulation results demonstrated that the designed algorithms performed better than other algorithms that have been addressed in the literature.
\newpage\thispagestyle{empty}\mbox{}  
\newpage\thispagestyle{empty}\mbox{}  
\chapter{SWIPT in Multi-Cell Networks}
\label{CHAP4}
\fancyhf{}
\renewcommand{\headrulewidth}{2pt}
\fancyhead[LE,RO]{\thepage}
\fancyhead[RE]{\textit{ \nouppercase{\leftmark}} }
\fancyhead[LO]{\textit{ \nouppercase{\rightmark}} }
\renewcommand{\footrulewidth}{0.1pt}
\fancyfoot[CE,CO]{\nouppercase{\leftmark}}
\fancyfoot[LE,RO]{JFMJ}

\vspace{15mm}

Most of the literature on the subject of simultaneous wireless information and power transfer (SWIPT) considers a single-cell network, in which users’ data-rate is a function of their signal-to-noise ratio (SNR).
Although the absence of interference simplifies problem solving, it is clear that these kinds of system models do not accurately represent real-world networks.
In contrast to chapter \ref{CHAP3}, here we examine a multi-cell network that incorporates interference in the data-rate function.
It is not always the best practice to maximize energy harvesting when it could result in the network having a lower total data-rate or throughput. 
Data-rate is conventionally defined as the number of delivered information bits per channel use measured in~bit/second/Hz~(bps/Hz).
In this chapter, we use a smart resource allocation policy to maximize the data-rate, subject to a minimum harvested energy constraint for users. 

The problem of resource allocation in multi-cell networks has already been addressed by other researchers, for instance, the authors of~\cite{06294504}, who studied an energy-efficient communication in the downlink (DL) of an orthogonal frequency division multiple access (OFDMA) multi-cell network with large numbers of single antenna cooperative base stations (BSs). 
They based their work on joint BS zero-forcing beamforming with full channel state information (CSI) while looking at the trade-off between energy efficiency, backhaul capacity, and network capacity.
The authors of \cite{08885781} studied an adaptive fairness scheduling algorithm in co-located SWIPT-enabled multi-cell DL networks that employs an adjustable weighted sum method as the utility function of each user to study fairness among users.
With the aim of maximizing energy harvesting efficiency, the authors of~\cite{08690892} investigated the beamforming design for multi-cell multi-user networks with SWIPT in a co-located receiver architecture.
A max-min signal-to-interference plus interference ratio (SINR) problem in the DL of a dense multi-cell SWIPT-enabled network was explored in~\cite{7552733}, in which the users near their serving BS can perform both energy harvesting (EH) and information decoding (ID) using the power splitting receiver approach, whereas far-field users can only perform ID.
In this chapter, we present a resource allocation policy in a {\text{multi-cell multi-user}} SWIPT-enabled network with a separated receiver architecture.
In our system model design, each cell has two ring-shaped boundary regions, in which EH users are placed inside the inner boundary near small base stations (SBSs) and ID users are located in the outer region.

As described in the first chapter, the literature presents several interesting receiver designs for enabling SWIPT, of which the four most viable designs are time switching~(TS), power splitting~(PS), antenna switching~(AS), and the separated receiver architecture.
The complex circuitries associated with TS, PS, and AS receiver approaches for SWIPT add one or more additional optimization parameter(s) to the design of a proper resource allocation policy in order to segregate the received signal so it can carry out distinct or simultaneous ID and EH operations.
The separated receiver architecture as an enabler of SWIPT networks could be very useful for dramatically reducing the architectural requirements within the transmit-receive operations.
Separated design architecture in itself effectively reduces the design complexity of a receiver capable of SWIPT.
Here, we investigate resource allocation for a SWIPT-enabled multi-cell multi-user OFDMA network with a separated receiver architecture.
More precisely, we focus on designing a resource allocation algorithm for an OFMDA scheme with SWIPT, in which each SBS serves multiple ID and EH users with a single antenna.
For this configuration, we propose a joint subcarrier assignment and power allocation optimization problem that maximizes the total data-rate of the network with the aim of satisfying a minimum data-rate requirement, fulfilling a minimum amount of harvested energy, and respecting the maximum transmit power allowance constraints. 
The users’ data-rate is proportional to their SINR due to the shared frequency spectrum between the cells.
Interference in the data-rate function would make the optimization problem non-linear and non-convex – and therefore, very challenging.
Next we propose an efficient algorithm via the majorization minimization (MM) approach based on the difference of convex functions (D.C.) programming and variable relaxation. 
This algorithm is intended to deal with these issues and effectively help to obtain a locally optimal solution for the original problem.
Simulation results confirm that our proposed algorithm achieves excellent performance as compared to other studies described in the literature.

\section{System Model}
In this section, we aim at extending the scenario of having only one small-cell, as in \text{chapter \ref{CHAP3}}, to a multi-cell case in an indoor application.
As shown in figure~(\ref{fig:4-1}), we consider a DL OFDMA-based network with $J$ cells each having one serving SBS.
The set of all SBSs is denoted by ${j \in \mathcal{J}=\{1,2,...,J\}.}$
We further assume that each SBS is equipped with a single antenna.~The single antenna assumption also holds for all receivers.
Furthermore, two sets of users are distinguished in each cell.
Let us define the set of ID and EH users at the $j^{th}$ cell as $\mathcal{K}^{ID}_{j}=\{1,2,...,K^{ID}_{j}\}$ and $\mathcal{K}^{EH}_{j}=\{1,2,...,K^{EH}_{j}\}$,~respectively, where~$\mathcal{K_I}=\sum_{j\in\mathcal{J}}\mathcal{K}_j^{ID}$ indicates the total number of ID users, {$\mathcal{K_E}=\sum_{j\in\mathcal{J}}\mathcal{K}_j^{EH}$} shows the total number of EH users, and $\mathcal{K}=\mathcal{K_I}+\mathcal{K_E}$ gives the total number of users in the network.
The receivers in this chapter are particularly configured based on the separated receiver topology in a multi-cell system using SWIPT, where each cell has two ring-shaped boundary regions.
The near users inside the inner boundary can only harvest energy from the SBSs, whereas the far users only receive information signals from the SBSs in the outer region.
We additionally assume that the entire frequency band of $\mathscr{B}$ is divided for $N$ orthogonal subcarrier, each having a bandwidth of $\mathscr{W}$.
Furthermore, we consider that each subcarrier is assigned to at most one ID user.
Moreover, we assume that the perfect CSI is available at the resource allocator to design the resource allocation policy so as to unveil the performance upper bound of the considered network.
For the sake of readability, we first introduce some of the essential parameters that are used to describe the system model:
\begin{itemize}
\item $h^{ID}_{j,n,k}$: The DL channel gain for the wireless information transfer from the $j^{th}$ SBS to the $k^{th}$ ID user over the $n^{th}$ subcarrier.
\item $g^{EH}_{j,n,k}$: The DL channel gain for the wireless power transfer from the $j^{th}$ SBS to the $k^{th}$ user over the $n^{th}$ subcarrier.
\item $a_{j,n,k}$: Binary subcarrier indicators from the $j^{th}$ SBS  corresponding to the $k^{th}$ user when the $n^{th}$ subcarrier is selected.
\item $p_{j,n,k}$: The corresponding transmit power from the $j^{th}$ SBS to the $k^{th}$ user over the $n^{th}$ subcarrier.
\end{itemize}
\begin{figure}[p]
\centering
\includegraphics[width=21cm,trim=4 4 4 4,clip,angle=90]{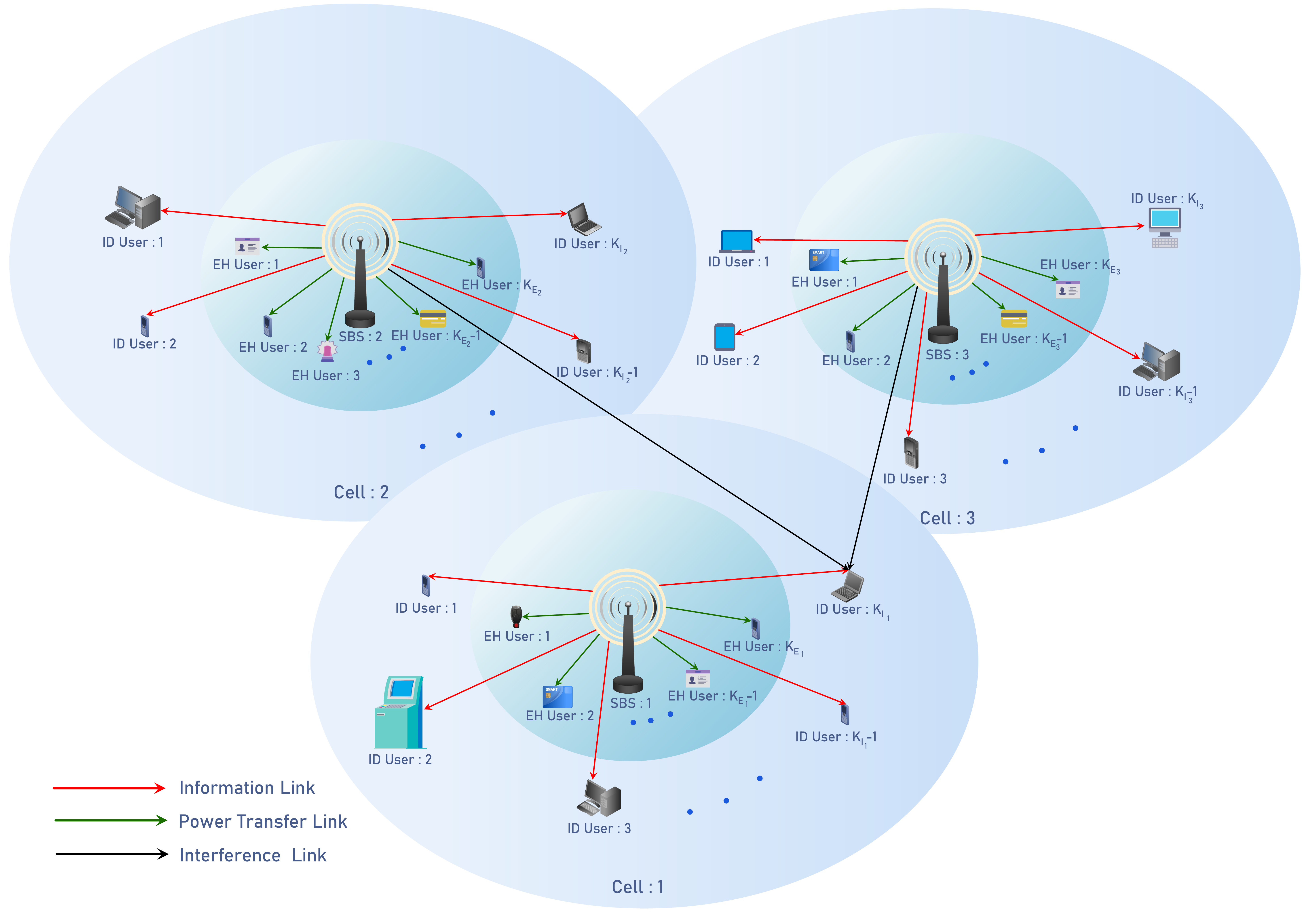}
\caption{SWIPT in a DL of a multi-user multi-cell OFDMA network with a separated receiver architecture.}
\label{fig:4-1}
\end{figure}
Then, the received signal at the $j^{th}$ SBS over the $n^{th}$ subcarrier corresponding to the $k^{th}$ ID user is given by
\begin{equation}\label{F4-1}
y^{ID}_{j,n,k}=
a_{j,n,k}\sqrt{p_{j,n,k}}h^{ID}_{j,n,k}+
\sum_{\substack{j'\neq j \\ j' \in \mathcal{J}}}
\sum_{\substack{k'\neq k \\ k' \in \mathcal{K}}}
a_{j',n,k'}\sqrt{p_{j',n,k'}}h^{ID}_{j',n,k}+ 
z^{ID}_{j,n,k},
\end{equation}
where $z^{ID}_{j,n,k}$ is the additive white Gaussian noise at an ID user.~More specifically, $z^{ID}_{j,n,k}$ is an additive white Gaussian noise (AWGN) random variable with zero mean and variance $|\sigma_{j,n,k}^{{ID}}|^{^2}$ denoted by $z^{ID}_{j,n,k} \sim  \mathcal{CN}(0,|\sigma_{j,n,k}^{{ID}}|^{^2})$.
Furthermore,~the harvested signal from the $j^{th}$ SBS for the EH user $k$ over the subcarrier $n$ is given by
\begin{equation}
y^{EH}_{j,n,k}=
\sqrt{p_{j,n,k}}g^{EH}_{j,n,k} + 
z^{EH}_{j,n,k},
\end{equation}
where $z^{EH}_{j,n,k}$ is the additive white Gaussian noise at the EH user with a circularly symmetric Gaussian distribution referred to as $z^{EH}_{j,n,k} \sim  \mathcal{CN}(0,|\sigma_{j,n,k}^{{ED}}|^{^2})$.~According to the famous Shannon capacity formula, the data-rate of the $k^{th}$ ID user over the subcarrier $n$ inside the cell $j$ can be written as
\begin{equation}\label{F4-3}
R_{j,n,k}=\log_2
\Bigg(1+
\frac{a_{j,n,k}p_{j,n,k}|h_{j,n,k}^{ID}|^2}{|\sigma_{j,n,k}^{{ID}}|^{^2}+I_{j,n,k}}
\Bigg),
\end{equation}
where 
\begin{equation}
I_{j,n,k}=
\sum_{\substack{j'\neq j \\ j' \in \mathcal{J}}}
\sum_{\substack{k'\neq k \\ k' \in \mathcal{K}}}
a_{j',n,k'}{p_{j',n,k'}}|h^{ID}_{j',n,k}|^2, 
\end{equation}
is the interference term arising from the co-channel effect on the subcarrier $n$ that is emitted by unintended transmitters sharing the same frequency channel.

On the other hand, the transferred wireless energy can be harvested at each EH user when a simple linear EH model is considered.
Therefore, the amount of the harvested energy of the $k^{th}$ EH user in the cell $j$ can be calculate according to 
\begin{equation} \label{F4-4}
\textrm{EH}_{j,k}=
\epsilon_{j,k} \sum_{n\in \mathcal{N}}a_{j,n,k}{p_{j,n,k}}|g^{EH}_{j,n,k}|^{2},
\end{equation}
where $\epsilon_{j,k}$ is the power conversion efficiency as introduced in the previous chapter.~It should be noted that the contribution of the noise power, i.e., $|\sigma_{j,n,k}^{{ID}}|^{^2}$, to EH formula in (\ref{F4-4}) is ignored due to being reportedly small compared to the other existing term in (\ref{F4-4}).

Let us now define $\textbf{p}_{kj}=[p_{j,1,k},...,p_{j,N,k}]$ as a row vector encompassing transmit power values assigned to the $k^{th}$ user in the $j^{th}$ cell for subcarriers, and $\textbf{a}_{kj}=[a_{j,1,k},...,a_{j,N,k}]$ as a row binary vector for the same user indicating the selected subcarriers in the $j^{th}$ cell.~Furthermore, $\textbf{p}=[\textbf{p}_{1,1},...,\textbf{p}_{|\mathcal{K}|,j}]^T$
and $\textbf{a}=[\textbf{a}_{1,1},...,\textbf{a}_{|\mathcal{K}|,j}]^T$, 
denote a vector containing total transmit power in the $j^{th}$ cell for all user and a binary vector representing the assigned subcarrier(s) to each user,~respectively. 
Therefore, we can define the total data-rate as
\begin{equation}
R^{\text{Total}}(\textbf{a},\textbf{p}) = 
\sum_{j\in \mathcal{J}}
\sum_{k\in\mathcal{K_{I}}}
\sum_{n\in \mathcal{N}} R_{j,n,k}.
\label{chap4:F4:rate}
\end{equation}
Similarly, we can define the total harvested energy in the network as 
\begin{equation}\label{4F4-7} 
\textrm{EH}^{\text{Total}}(\textbf{a},\textbf{p})=
\sum_{j\in \mathcal{J}}
\sum_{k\in \mathcal{K_E}}\text{EH}_{j,k}.
\end{equation}
\begin{figure}[p]
\centering
\includegraphics[width=21cm,trim=4 4 4 4,clip,angle=90]{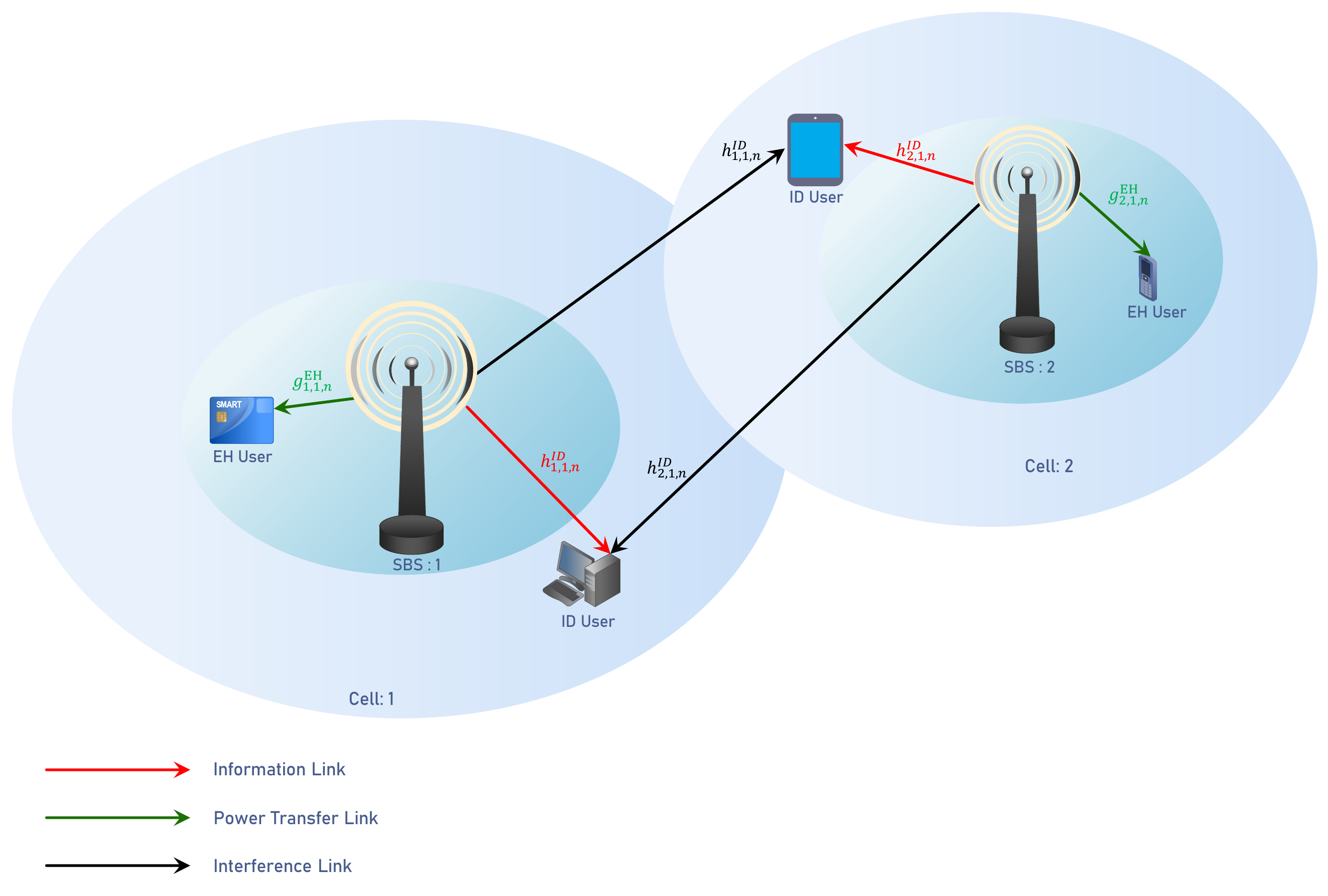}
\caption{SWIPT in a DL of an OFDMA network consisting of $J=2$ small cells, where there is one user of each type in each cell, i.e., $\mathcal{K}^{ID}_1=\mathcal{K}^{ID}_2=\mathcal{K}^{EH}_1=\mathcal{K}^{EH}_2$=1.}
\label{fig:4-2}
\end{figure}

Furthermore, to guarantee the quality of service (QoS), a minimum data-rate denoted by $R_{min}$ should be provided for ID users. That is
\begin{align}
\sum_{n \in \mathcal{N}}R_{j,n,k}\geq R_{min}, ~\forall j \in \mathcal{J}, k \in \mathcal{K_I}.
\end{align}
Moreover, a minimum harvested energy referred to as EH$_{min}$ is also considered for each EH user
\begin{align}
\text{EH}_{j,k}\geq \text{EH}_{min}, ~\forall j \in \mathcal{J}, k \in \mathcal{K_E}.
\end{align}
\section{Optimization Problem Formulation}
In this section, we aim at finding a subcarrier assignment and power allocation policy via formulating the system data-rate or throughput maximization problem,
while fulfilling a minimum harvested energy requirement for EH receivers and a minimum data-rate requirement for ID receivers.
Consequently, we introduce the following optimization problem to maximize the total data-rate of the network 
\begin{subequations}
\begin{align}
&\max_{\textbf{a},\textbf{p}}R^{\text{Total}}(\textbf{a},\textbf{p})\label{F4-5}\\
s.t.: 
&~C_{1}:
\sum_{k\in \mathcal{K_I}} 
a_{j,n,k}\leq 1,~~~~~~~~~~~~~~~~~~~
\forall j \in \mathcal{J},~
\forall n \in \mathcal{N},
\label{F4-6}\\ 
&~C_{2}:
\sum_{k\in \mathcal{K}}~
\sum_{n\in \mathcal{N}}
a_{j,n,k}\ p_{j,n,k}
\leq p_{max},~
\forall j \in \mathcal{J},
\label{F4-7}\\
&~C_{3}:
\sum_{n\in \mathcal{N}}
R_{j,n,k}\geq R_{min},~~~~~~~~~~~~~~
\forall j \in \mathcal{J},~
\forall k \in \mathcal{K_{I}},
\label{F4-8}\\
&~C_{4}:
\text{EH}_{j,k}\geq \text{EH}_{min},~~~~~~~~~~~~~~~~~~
\forall j \in \mathcal{J},~
\forall k \in \mathcal{K_{E}},
\label{F4-8-}\\
&~C_{5}:a_{j,n,k}\in\{0,1\},~~~~~~~~~~~~~~~~~~~~
\forall j \in \mathcal{J},~
\forall n \in \mathcal{N},~
\forall k \in \mathcal{K}.  
\label{F4-9} 
\end{align}
\label{chap4:f6:main}%
\end{subequations}

\vspace{-8mm}
In the optimization problem (\ref{chap4:f6:main}),~$C_{1}$ indicates that each subcarrier can be allocated to at most one ID user in each cell.
$C_{2}$ indicates that the total transmit power of SBSs should not exceed their maximum threshold, which is denoted by $p_{max}$.
In $C_{3}$, a minimum rate requirement, $R_{min}$, is guaranteed for each ID user in each cell and
$C_{4}$ makes sure of a minimum harvested energy, EH$_{min}$, is satisfied for each EH user.
Lastly, $C_{5}$ represents that the subcarrier indicator variable takes only binary values.

Due to the multiplication of two variables, the binary subcarrier allocation variables, and the interference included in the data-rate function, the problem (\ref{chap4:f6:main}) is mixed integer non-linear programming (MINLP), which is generally difficult to solve.
These challenges that make the above optimization problem complicated are explained further below:
\begin{itemize}
\item Multiplication of two variables: 
Since the multiplication of two variables is non-convex, the term~$a_{j,n,k}p_{j,n,k}$ poses a challenge in tackling the optimization problem in (\ref{chap4:f6:main}).
More specifically, because the maximum total transmit power constraint in $C_2$, the minimum data-rate requirement constraint in $C_3$, the minimum harvested energy requirement constraint in $C_4$, and also the total data-rate objective function, are multiplied by a function of both the transmit power and subcarrier allocation variables (as given in (\ref{F4-7}), (\ref{F4-8}), (\ref{F4-8-}), and (\ref{F4-5})), these mentioned constraints together with the objective function are non-convex. 
\item Interference: 
The inter-cell interference incorporated in data-rate functions, makes both the constraint $C_3$ and the objective function non-convex.
\item Binary subcarrier assignment variable: 
Discrete subcarrier assignment in $C_5$ turns (\ref{chap4:f6:main}) into a complex MINLP problem.
\end{itemize}

In the following section, we first restate the problem (\ref{chap4:f6:main}) as a mathematically tractable form to maximize the total data-rate by considering the interference in the data-rate function.
We also guarantee a minimum data-rate for ID users and a minimum harvested energy for EH users.
Furthermore, we propose a suboptimal resource allocation algorithm, which has a polynomial-time computational complexity to compromise between complexity and system performance.

\section{Solution to the Optimization Problem}
In order to address the mixed non-convex and combinatorial optimization problem (\ref{chap4:f6:main}), we first deal with the problem of variables multiplication in constraints $C_{2}$, $C_{3}$, and $C_{4}$.
In order to handle this difficulty, we adopt the big-M formulation \cite{big_M} to decouple the product terms.
Therefore, we impose the following additional constraints
\begin{align}
&C_{6}:
\tilde{p}_{j,n,k}\leq p_{max}a_{j,n,k},~~~~~~~~~~~~~~~~~~~
\forall j \in \mathcal{J},~\forall n \in \mathcal{N},~\forall k \in \mathcal{K},
\\
&C_{7}:
\tilde{p}_{j,n,k}\leq p_{j,n,k},~~~~~~~~~~~~~~~~~~~~~~~~~
\forall j \in \mathcal{J},~\forall n \in \mathcal{N},~\forall k \in \mathcal{K},
\\
&C_{8}:
\tilde{p}_{j,n,k}\geq  p_{j,n,k}-(1-a_{j,n,k})p_{max},~
\forall j \in \mathcal{J},~\forall n \in \mathcal{N},~\forall k \in \mathcal{K},
\\
&C_{9}:
\tilde{p}_{j,n,k}\geq 0,~~~~~~~~~~~~~~~~~~~~~~~~~~~~~~
\forall j \in \mathcal{J},~\forall n \in \mathcal{N},~\forall k \in \mathcal{K},
\end{align}     
where $\tilde{\textbf{p}}\in \mathbb{R}^{1\times JNK}$ is the collection of all $\tilde{p}_{j,n,k}$'s, and ${\textbf{a}}\in \mathbb{Z}^{1\times JNK}$ is the collection of all $a_{j,n,k}$'s.
Now, by revisiting the definition of the data-rate function, the data-rate of the $k^{th}$ ID user over the subcarrier $n$ inside the $j^{th}$ cell in (\ref{F4-3}), and also, the total data-rate of the network in (\ref{chap4:F4:rate}) can be rewritten respectively as 
\begin{align}
&\overline{R}_{j,n,k} = \log_2
\Bigg(1+\frac{\tilde{p}_{j,n,k}|h^{ID}_{j,n,k}|^2}
{|\sigma_{j,n,k}^{{ID}}|^{^2}+
\sum\limits_{\substack{j'\neq j \\ j' \in \mathcal{J}}}~\sum\limits_{\substack{k'\neq k \\ k' \in \mathcal{K}}}~{\tilde{p}_{j',n,k'}}|h^{ID}_{j',n,k}|^{2}}\Bigg),\\
&\widehat{\overline{R}}^{\text{Total}}(\textbf{a},\tilde{\textbf{p}}) = 
\sum_{j\in \mathcal{J}}\sum_{k\in\mathcal{K_{I}}}\sum_{n\in \mathcal{N}} \overline{R}_{j,n,k}.
\end{align}
Furthermore, the amount of the harvested energy of the $k^{th}$ EH user in the cell $j$ in (\ref{F4-4}) and the total amount of the harvested energy of the network in (\ref{4F4-7}) can be restated as 
\begin{align}
&\overline{\textrm{EH}}_{j,k}=
\epsilon_{j,k} \sum_{n\in \mathcal{N}}\tilde{p}_{j,n,k}|g^{EH}_{j,n,k}|^{2},\\
&\widehat{\overline{\textrm{EH}}}^{\text{Total}}(\textbf{a},\tilde{\textbf{p}})=
\sum_{j\in \mathcal{J}}
\sum_{k\in \mathcal{K_E}}
\overline{\textrm{EH}}_{j,k}.
\end{align}
Hence, the original optimization problem in (\ref{chap4:f6:main}) can be modified as
\begin{subequations}
\begin{align}
&\max_{\textbf{a},\textbf{p},\tilde{\textbf{p}}}
\widehat{\overline{R}}^{\text{Total}}
(\textbf{a},\tilde{\textbf{p}})\\
s.t.: 
&~C_{1},C_{5}-C_{9},\\ 
&~C_{2}:
\sum_{k\in \mathcal{K}}~
\sum_{n\in \mathcal{N}}
\tilde{p}_{j,n,k} \ \leq p_{max},~
\forall j \in \mathcal{J},\\
&~C_{3}:
\overline{R}_{j,n,k}\geq  R_{min},~~~~~~~~~~~~~
\forall j \in \mathcal{J},~
\forall k \in \mathcal{K_{I}},\\
&~C_{4}:
\overline{\textrm{EH}}_{j,k} \geq \text{EH}_{min},~~~~~~~~~~~
\forall j \in \mathcal{J},~
\forall k \in \mathcal{K_{E}}.
\end{align}
\label{chap4:f6:main2}%
\end{subequations}
Through this method, we have easily and efficiently dealt with the non-convex constraints $C_{2}-C_{4}$ by using their equivalent convex forms.
Another challenge in solving the above optimization problem is due to the incorporating interference in the data-rate functions in the constraint $C_{3}$ and objective function.
This makes the resulting optimization problem in (\ref{chap4:f6:main2}) still non-convex. 
To facilitate the solution design, we first rewrite the optimization problem in terms of the D.C. functions.
Mathematically speaking, the problem can be restated as
\begin{subequations}
\begin{align}
&\max_{\textbf{a},\textbf{p},\tilde{\textbf{p}}} \ 
\sum_{j \in \mathcal{J}}
\sum_{ k \in \mathcal{K_I}}
\mathscr{U}(\textbf{a},\tilde{\textbf{p}})-
\mathcal{V}(\textbf{a},\tilde{\textbf{p}})
\label{F4-17}\\
s.t.:~& 
\dot{C}_3: 
\mathscr{U}(\textbf{a},\tilde{\textbf{p}})-
\mathcal{V}(\textbf{a},\tilde{\textbf{p}})
\geq R_{min},~\forall j \in \mathcal{J},~\forall k \in \mathcal{K_{I}},
\label{F4-20b}\\
& C_{1}-C_{2},C_{4}-C_{9}.
\end{align}
\label{F2:20}%
\end{subequations}
where $\mathscr{U}(\textbf{a},\tilde{\textbf{p}})$~and $\mathcal{V}(\textbf{a},\tilde{\textbf{p}})$ are defined as
\begin{align}
\mathscr{U}(\textbf{a},\tilde{\textbf{p}}) & =
\sum_{n \in \mathcal{N}}\log_2\Big(
\tilde{p}_{j,n,k}|h^{ID}_{j,n,k}|^2+{|\sigma_{j,n,k}^{{ID}}|^{^2}}+
\sum_{\substack{j'\neq j \\ j' \in \mathcal{J}}}
\sum_{\substack{k'\neq k \label{4F4:21}\\ k' \in \mathcal{K}}}
\tilde{p}_{j',n,k'}{|h^{ID}_{j',n,k}|^{2}}\Big),\\
\mathcal{V}(\textbf{a},\tilde{\textbf{p}}) & =
\sum_{n \in \mathcal{N}}
\log_2\Big({{|\sigma_{j,n,k}^{{ID}}|^{^2}}+
\sum_{\substack{j'\neq j \\ j' \in \mathcal{J}}}
\sum_{\substack{k'\neq k \\ k' \in \mathcal{K}}}
\tilde{p}_{j',n,k'}|h^{ID}_{j',n,k}|^{2}}\Big).
\label{4F4:22}
\end{align}
It should be emphasized that $-\mathscr{U}(\textbf{a},\tilde{\textbf{p}})$ and $-\mathcal{V}(\textbf{a},\tilde{\textbf{p}})$ are now convex functions.
In addition, the problem in (\ref{F2:20}) belongs to the class of D.C. programming problems.
Consequently, the first-order Taylor approximation can be applied to approximate the D.C. components.
This facilitates the design of a computationally efficient iterative resource allocation algorithm for obtaining a locally close-to-optimal solution.
Thus, for any feasible point ${\textbf{a}}^{t-1}$ and $\tilde{\textbf{p}}^{t-1}$, the following approximation holds \cite{SCA,DC}
\begin{align}
\mathcal{V}(\textbf{a},\tilde{\textbf{p}}) 
\simeq
\mathcal{V}(\textbf{a}^{t-1},\tilde{\textbf{p}}^{t-1}) +
\nabla_{\tilde{\textbf{p}}}\mathcal{V}
(\textbf{a}^{t-1},\tilde{\textbf{p}}^{t-1})^T.(\tilde{\textbf{p}}-
\tilde{\textbf{p}}^{t-1})
\triangleq 
\tilde{\mathcal{V}}(\textbf{a},\tilde{\textbf{p}}),
\label{F4-25}
\end{align}
where $\tilde{\textbf{p}}^{t-1}$ is the solution of the problem at the $(t-1)^{th}$ iteration, and $\nabla_{\square}$ represents the gradient with respect to ${\square}$.
Accordingly, we rewrite the optimization problem in (\ref{F2:20}) as follows
\begin{subequations}
\begin{align}
&\max_{\textbf{a},\textbf{p},\tilde{\textbf{p}}} \ 
\sum_{j \in \mathcal{J}}
\sum_{ k \in \mathcal{K_I}}
\mathscr{U}(\textbf{a},\tilde{\textbf{p}})-
\tilde{\mathcal{V}}(\textbf{a},\tilde{\textbf{p}})
\\
s.t.:~& \dot{C}_3: 
\mathscr{U}(\textbf{a},\tilde{\textbf{p}})-
\tilde{\mathcal{V}}(\textbf{a},\tilde{\textbf{p}})\geq R_{min},~
\forall j \in \mathcal{J},~
\forall k \in \mathcal{K_{I}},
\\
~& C_{1}-C_{2},C_{4}-C_{9}. 
\end{align} 
\label{F4-26}%
\end{subequations}
It can be perceived that optimization problem in (\ref{F4-26}) is still non-convex due to the integer subcarrier allocation variable, i.e., $a_{j,n,k}$.
This binary variable turns (\ref{F4-26}) into an MINLP problem, which makes it challenging to solve with a polynomial-time complexity.
To address this issue, we adopt an approach similar to chapter \ref{CHAP3}, and substitute the constraint $C_{5}$ with the following inequalities
\begin{align}
&\dot{C}_{5}:
0\leq a_{j,n,k}\leq 1,~
\forall j \in \mathcal{J},~
\forall n \in \mathcal{N},~
\forall k \in \mathcal{K},\\
&\ddot{C}_{5}:
\sum_{j\in \mathcal{J}}
\sum_{k\in\mathcal{K}}
\sum_{n\in \mathcal{N}}
a_{j,n,k}-(a_{j,n,k})^{2}\leq 0 \label{F4:26}. 
\end{align}
In this regard, the constraint $\dot{C}_{5}$ converts the binary variable $a_{j,n,k}$ into a continuous variable with values in the close interval $[0,1]$ which contains zero, one, and all real numbers in between.
However, in the constraint $\ddot{C}_{5}$, the value of $a_{j,n,k}$ is restricted to two possible values, i.e., zero and one. 
These two integer numbers are the only numbers that satisfy the constraint $\ddot{C}_{5}$, and thereby, belong to the set $\{0,1\}$.
In the same way as discussed in previous paragraphs, we can now write the constraint $\ddot{C}_{5}$ in the D.C. format as $\nu(\textbf{a})-\mu(\textbf{a})\leq 0$ where
\begin{align}
&\nu(\textbf{a})=
\sum_{j\in \mathcal{J}}
\sum_{k\in\mathcal{K}}
\sum_{n\in \mathcal{N}}a_{j,n,k},\\
&\mu(\textbf{a})=
\sum_{j\in \mathcal{J}}
\sum_{k\in\mathcal{K}}
\sum_{n\in \mathcal{N}}
(a_{j,n,k})^{2}.
\end{align}
Similar to the approach used for the data-rate function in (\ref{4F4:21}) and (\ref{4F4:22}), we employ an equivalent methodology based on the MM approach to make the constraint $\ddot{C}_{5}$ convex. 
This is done by taking the first-order Taylor approximation of $\mu(\mathbf{a})$.
Consequently,~$\mu(\mathbf{a})$ can be approximated as
\begin{align}\label{mu}
\mu(\mathbf{a}) 
\simeq 
\mu(\mathbf{a}^{(t-1)}) +
\nabla_{\mathbf{a}}{\mu}
(\mathbf{a}^{t-1})^T(\mathbf{a}-\mathbf{a}^{t-1})
\triangleq 
\tilde{\mu}(\mathbf{a}).
\end{align}
Therefore, the constraint $\ddot{C}_5$ can be restated as $\nu(\mathbf{a})-\tilde{\mu}(\mathbf{a})\leq 0$ which is now convex.
Finally, the optimization problem at hand can be reformulated as
\begin{subequations}
\begin{align}
&\max_{\textbf{a},\textbf{p},\tilde{\textbf{p}}} \ 
\sum_{j \in \mathcal{J}}
\sum_{ k \in \mathcal{K_I}}
\mathscr{U}(\textbf{a},\tilde{\textbf{p}})-
\tilde{\mathcal{V}}
(\textbf{a},\tilde{\textbf{p}})\\
s.t.:~
& \dot{C}_3: 
\mathscr{U}(\textbf{a},\tilde{\textbf{p}})-
\tilde{\mathcal{V}}(\textbf{a},\tilde{\textbf{p}})\geq R_{min}
\label{F4-18},~
\forall j \in \mathcal{J},~
\forall k \in \mathcal{K_{I}},\\
& \ddot{C}_5: 
\nu(\mathbf{a})-\tilde{\mu}(\mathbf{a})
\leq 0,\label{F4-19}\\
& C_{1}-C_{2},C_{4},\dot{C}_5,C_{6}-C_{9}.
\end{align}
\label{F4-26-}%
\end{subequations}
The optimization problem (\ref{F4-26-}) is convex and can be solved efficiently via D.C. approximation based on the interior point methods~\cite{MM}.~As a result, the solution of (\ref{F4-26-}) would be an approximation to the solution of the original problem given in (\ref{chap4:f6:main}).
However, in D.C. programming, the iteration begins from a feasible initial point and solves the optimization problem iteratively until it eventually approaches a close-to-optimal solution~\cite{che2014joint,DC,CL_Ata}.
Besides, it worth mentioning that the MM approach produces a sequence of improved feasible solutions with the adopted D.C. approximation, which will ultimately converge to a locally optimal solution $(\textbf{a}^{*},\textbf{p}^*,\tilde{\textbf{p}}^*)$ using standard convex program solvers such as CVX.

\begin{proposition}\label{proposition3}
The solution obtained from (\ref{F4-26-}) satisfies the constraints $\dot{C}_3$ and $\ddot{C}_5$ in (\ref{F4-20b}) respectively (\ref{F4:26}), by using the first-order concave function for $\mathcal{V}(\textbf{a},\tilde{\textbf{p}})$ and the first-order convex function for ${\mu}(\textbf{a})$.
\end{proposition}
\begin{proof}
It is straightforward to show that $\mathcal{V}(\textbf{a},\tilde{\textbf{p}})$ is a concave function.
Hence, the gradient of $\mathcal{V}(\textbf{a},\tilde{\textbf{p}})$ is a supergradient~\cite{KENNETH-LANGE}.~This yields 
\begin{equation} \label{F4-29}
\mathcal{V}(\textbf{a},\tilde{\textbf{p}})\leq\tilde{\mathcal{V}}(\textbf{a},\tilde{\textbf{p}}). 
\end{equation}
Subsequently, one may easily conclude from the inequality in (\ref{F4-25}) that 
\begin{equation}
\mathscr{U}(\textbf{a},\tilde{\textbf{p}})-
\mathcal{V}(\textbf{a},\tilde{\textbf{p}})
\geq 
\mathscr{U}(\textbf{a},\tilde{\textbf{p}})-
\tilde{\mathcal{V}}(\textbf{a},\tilde{\textbf{p}}).
\end{equation}
Similarly, since $\mu(\textbf{a})$ is a convex function, its gradient is a subgradient~\cite{KENNETH-LANGE}.~Thus, we have
\begin{equation}
\nu(\textbf{a})-\mu(\textbf{a})
\leq 
\nu(\textbf{a})-\tilde{\mu}(\textbf{a}).
\end{equation}
Nevertheless, it should be noted that the existence of subgradients and supergradients are closely tied to the fact that $\mu(\textbf{a})$ and $-\mathcal{V}(\textbf{a},\tilde{\textbf{p}})$ are locally Lipschitz \cite{KENNETH-LANGE}.
Accordingly, the constraint $\mathscr{U}(\textbf{a},\tilde{\textbf{p}})-{\mathcal{V}}(\textbf{a},\tilde{\textbf{p}})\geq R_{min}$ is satisfied if  $\mathscr{U}(\textbf{a},\tilde{\textbf{p}})-\tilde{\mathcal{V}}(\textbf{a},\tilde{\textbf{p}})$ is greater than $R_{min}$ for each user.
In the same way, if $\nu(\textbf{a})-\tilde{\mu}(\textbf{a}) \leq 0$, the inequality $\nu(\textbf{a})-{\mu}(\textbf{a}) \leq 0$ holds as well. 
This completes the proof. 
\end{proof}

\begin{proposition}
The surrogate functions that are employed to approximate the non-convex optimization problem in (\ref{F4-26-}) provide a tight lower bound for the objective function in (\ref{F2:20}).  
\end{proposition}

\begin{proof}
Borrowing the first part of the proof from the proof of \textbf{Proposition~\ref{proposition3}}, the following inequality holds for the objective function in (\ref{F4-26-})
\begin{equation}
\mathscr{U}(\textbf{a},\tilde{\textbf{p}})-
\mathcal{V}(\textbf{a},\tilde{\textbf{p}})
\geq 
\mathscr{U}(\textbf{a},\tilde{\textbf{p}})-
\mathcal{V}(\textbf{a}^{t-1},\tilde{\textbf{p}}^{t-1}) +
\nabla_{\tilde{\textbf{p}}}\mathcal{V}(\textbf{a}^{t-1},\tilde{\textbf{p}}^{t-1})^T.(\tilde{\textbf{p}}-\tilde{\textbf{p}}^{t-1}),
\end{equation}
where, the equality holds when $\textbf{a}=\textbf{a}^{t-1}$ and $\tilde{\textbf{p}}=\tilde{\textbf{p}}^{t-1}$.~Thus, the MM updates furnish a candidate subgradient in each iteration.~This demonstrates the tightness of the lower bound and completes the proof.
\end{proof}

\begin{proposition}
By incorporating D.C. approximation, the solution of (\ref{F4-26-}) improves after each iteration
\end{proposition}
\begin{proof}
For the objective function in (\ref{F2:20}), we have the following in the $t^{th}$ iteration
\begin{equation}
\mathscr{U}(\textbf{a}^t,\tilde{\textbf{p}}^t)-\mathcal{V}(\textbf{a}^t,\tilde{\textbf{p}}^t) 
\end{equation}
Subsequently, in the next iteration, we have
\begin{align}
\mathscr{U}(\textbf{a}^{t+1},\tilde{\textbf{p}}^{t+1})-
\mathcal{V}(\textbf{a}^{t+1},&\tilde{\textbf{p}}^{t+1}) 
\nonumber\\
& \geq 
\mathscr{U}(\textbf{a}^{t+1},\tilde{\textbf{p}}^{t+1})-
\mathcal{V}(\textbf{a}^{t},\tilde{\textbf{p}}^{t})-
\nabla_{\tilde{\textbf{p}}}\mathcal{V}(\textbf{a}^{t},\tilde{\textbf{p}}^{t})^{T}.(\tilde{\textbf{p}}-\tilde{\textbf{p}}^{t}) 
\nonumber \\ 
&=\max_{\textbf{a},\tilde{\textbf{p}}}\mathscr{U}(\textbf{a},\tilde{\textbf{p}})-
\mathcal{V}(\textbf{a}^{t},\tilde{\textbf{p}}^{t})-
\nabla_{\tilde{\textbf{p}}}\mathcal{V}(\textbf{a}^{t},\tilde{\textbf{p}}^{t})^{T}.(\tilde{\textbf{p}}-\tilde{\textbf{p}}^{t})
\nonumber\\
&\geq 
\mathscr{U}(\textbf{a}^{t},\tilde{\textbf{p}}^{t})-
\mathcal{V}(\textbf{a}^{t},\tilde{\textbf{p}}^{t}) -
\nabla_{\tilde{\textbf{p}}}\mathcal{V}(\textbf{a}^{t},\tilde{\textbf{p}}^{t})^{T}.(\tilde{\textbf{p}}^{t}-\tilde{\textbf{p}}^{t})
\nonumber\\
&= 
\mathscr{U}(\textbf{a}^{t},\tilde{\textbf{p}}^{t})-
{\mathcal{V}}(\textbf{a}^{t},\tilde{\textbf{p}}^{t}). 
\nonumber
\end{align} 
This completes the proof. 
\end{proof}
One can readily verify that the objective function of (\ref{F4-26-}) takes larger values as the iteration proceeds.~Therefore, the solution to the optimization problem improves gradually.
Hence, we adopt an iterative solution to tighten the obtained upper bound based on the \textbf{Algorithm~\ref{euclid}}.
Like so, the proposed iterative resource allocation scheme generates a monotonically non-decreasing sequence of feasible solution, i.e., $\textbf{a}^{t+1}$, $\textbf{p}^{t+1}$, and $\tilde{\textbf{p}}^{t+1}$, by solving the convex problem in (\ref{F4-26-}).

\begin{algorithm}[t]
\caption{Proposed Iterative Method via D.C. Programming Based on the MM
Approach}
\label{euclid}
\begin{algorithmic}[1]
\STATE {$\mathbf{Initialize}$} \\{
\begin{addmargin}[1em]{0em}
{MM iteration index $t=0$ with maximum number of MM iteration $T_{max}$\\
and feasible set vector $\mathbf{a}^{0}$, $\mathbf{p}^{0}$, and $\tilde{\mathbf{p}}^0$}.
\end{addmargin}}
\STATE{\textbf{repeat}}
\STATE{
\begin{addmargin}[1em]{0em}
Update $\tilde{\mathcal{V}}(\textbf{a},\tilde{\textbf{p}})$,~$\tilde{\mu}(\mathbf{a})$ as presented in (\ref{F4-25}) and (\ref{mu})~respectively. 
\end{addmargin}}
\STATE{
\begin{addmargin}[1em]{0em}
Solve optimization problem of (\ref{F4-26-}) and store the intermediate resource allocation policy $\mathbf{a}^t$, $\mathbf{p}^t$, and $\tilde{\mathbf{p}}^t$.
\end{addmargin}}
\STATE{
\begin{addmargin}[1em]{0em}
Set $t=t+1$.
\end{addmargin}}
\STATE{
\begin{addmargin}[1em]{0em}
Set \{$\mathbf{a}^t,\mathbf{p}^t$,$\tilde{\mathbf{p}}^t$\} $=$ \{$\mathbf{a},\mathbf{p}$,$\tilde{\mathbf{p}}$\}.
\end{addmargin}}
\STATE \textbf{until} Convergence or $t=T_{max}$
\STATE \textbf{return} \{$\mathbf{a}^{*},{\mathbf{p}}^{*}$,$\tilde{\mathbf{p}}^*$\} $=$ \{$\mathbf{a}^{t},{\mathbf{p}}^{t}$,$\tilde{\mathbf{p}}^t$\}
\end{algorithmic}
\end{algorithm}

\section{Computational Complexity}
In this section, we aim at investigating the computational complexity of the proposed algorithm. 
The optimization problem (\ref{F4-26-}) includes $NJK$ variables and $J(1+N+K)+5JNK$ linear convex constraints.
Therefore, the computational complexity is that of order $\mathcal{O}(NJK)^{3}(J(1+N+K)+5JNK)$.
Despite the polynomial-time computational complexity of our proposed algorithm, the computation cost is still high and may become unaffordable for resource allocators with predefined capabilities.
However, it should not be neglected that the computational complexity of \textbf{Algorithm~\ref{euclid}} is lower compared to the exhaustive search approaches.
Moreover, it is worth mentioning that \textbf{Algorithm~\ref{euclid}} provides a locally optimal solution that is closely approaching the optimal solution.
Nevertheless, in what follows, we provide a low complexity algorithm for designing the resource allocation policy.

\section{Low Complexity Algorithm Design (Lower Bound)}\vspace{-3mm}
In this section, we introduce another algorithm with even lower computational complexity to improve the practicality of \textbf{Algorithm~\ref{euclid}}.
In order to handle the non-convexity of data-rate functions, we first assume that there is an upper bound for the interference term and impose the following constraint on the optimization problem
\begin{align} 
I_{j,n,k} \leq 
I^{n}_{max},
\end{align}
where $I^{n}_{max}$ is the maximum tolerable inter-cell interference parameter.
In this way, we can derive an efficient resource allocation algorithm by trimming the solution design.
To improve performance, we can control the interference level in each subcarrier by the resource allocator policy through varying the value of $ I^{n}_{max}$~\cite{6294504}.
Also, we can have a concave function and therefore a computable data-rate, if we replace $I_{j,n,k}$ by $I^{n}_{max}$ in data-rate functions.
At this point, we describe a tractable solution methodology for the original problem in the following subsection.
\vspace{-3mm}

\subsection{Low Complexity Power Control and Subcarrier Assignment}\vspace{-2mm}
In this subsection, we seek to attain a low complexity suboptimal subcarrier assignment and power allocation algorithm.
To derive a cost-efficient resource allocation design, it is required to relax the binary subcarrier assignment constraint. 
This is done by transforming the subcarrier assignment variable $a_{j,n,k}$ into a continuous constraint with values within the close interval~$[0,1]$.
In this sense, the fraction of time that subcarrier $n$ is assigned to the user $k$ would be the physical interpretation of the continuous $a_{j,n,k}$.
Notwithstanding the non-convexity of the original optimization problem, strong duality still holds as a consequence of the time-sharing condition addressed in \cite{Time,6825834}.
Finally, by defining a new power allocation variable as $\tilde{q}_{j,n,k}=a_{j,n,k} p_{j,n,k}$, the modified optimization problem can be expressed as
\vspace{-2mm}
\begin{subequations}
\begin{align}
&\max_{\textbf{a},\textbf{p},\tilde{\textbf{q}}}
\underline{\mathcal{R}}^{\text{Total}}(\textbf{a},\tilde{\textbf{q}})
\label{F4-32}
\\
s.t.: &~C_{1}:
\sum_{k\in \mathcal{K_{I}}} a_{j,n,k}\leq 1,~~~~~~~~~~~~~~~~~~~~~~~~~~~~~~~~~~
\forall j \in \mathcal{J},~
\forall n \in \mathcal{N},
\\ 
&~C_{2}:
\sum_{k\in \mathcal{K}}~
\sum_{n\in \mathcal{N}}\tilde{q}_{j,n,k}\leq p_{max},~~~~~~~~~~~~~~~~~~~~~~~~
\forall j \in \mathcal{J},\\
&~C_{3}:
\sum_{n\in \mathcal{N}} 
\log_2
\Bigg(1+\frac{\tilde{q}_{j,n,k}\ |h^{ID}_{j,n,k}|^2}{{|\sigma_{j,n,k}^{{ID}}|^{^2}}+I^{n}_{max}}\Bigg)
\geq   
R_{min},~
\forall j \in \mathcal{J}, k \in \mathcal{K_{I}},
\\
&~C_{4}:
\epsilon_{j,k} 
\sum_{n \in \mathcal{N}}\tilde{q}_{j,n,k}|g^{EH}_{j,n,k}|^{2}\geq \text{EH}_{min},~~~~~~~~~~~~~~
\forall j \in \mathcal{J},~
\forall k \in \mathcal{K_{E}},\\
&~C_{5}:
0\leq a_{j,n,k}\leq 1,~~~~~~~~~~~~~~~~~~~~~~~~~~~~~~~~~~~
\forall j \in \mathcal{J},~
\forall n \in \mathcal{N},~
\forall k \in \mathcal{K},\\
&~C_{6}: 
\sum_{\substack{j'\neq j \\ j' \in \mathcal{J}}}
\sum_{\substack{k'\neq k \\ k' \in \mathcal{K}}}
\tilde{p}_{j',n,k'}|h^{ID}_{j',n,k}|^2\leq I^{n}_{max},~~~~~~~~~~~~
\forall n \in \mathcal{N},~
\forall k \in \mathcal{K},
\end{align}
\label{F4-32:main}%
\end{subequations}
where 
\vspace{-4mm}
\begin{align}
\underline{\mathcal{R}}^{\text{Total}}(\textbf{a},\tilde{\textbf{q}}) = 
\sum_{j\in \mathcal{J}}\sum_{k\in\mathcal{K_{I}}}\sum_{n\in \mathcal{N}} 
\log_2
\Bigg(1+\frac{\tilde{q}_{j,n,k}|h^{ID}_{j,n,k}|^2}
{|\sigma_{j,n,k}^{{ID}}|^{^2}+
I^{n}_{max}}\Bigg).
\end{align}
It is easy to verify that Slater's condition holds for the above convex optimization problem. Therefore, solving the dual problem is equivalent to solving the primal problem due to strong duality.
In order to obtain the corresponding resource allocation policy, the Lagrangian method is applied to the convex optimization problem in (\ref{F4-32:main}).
Hence, the {\text{Lagrangian function is as follows}}
\vspace{-2mm}
\begin{align}
\mathcal{L}(\textbf{a},\mathbf{p},\tilde{\textbf{q}},\boldsymbol{\chi},\boldsymbol{\phi},\boldsymbol{\zeta},\boldsymbol{\tau},\boldsymbol{\theta})= &  
\quad 
\underline{\mathcal{R}}^{\text{Total}}(\textbf{a},\tilde{\textbf{q}})
\nonumber 
\label{F4-38}\\ 
&-\boldsymbol{\chi}\bigg(\sum_{k\in \mathcal{K_{I}}} a_{j,n,k}- 1\bigg)
\nonumber \\
&-\boldsymbol{\phi}\bigg(\sum_{k\in \mathcal{K}}~\sum_{n\in \mathcal{N}}\tilde{q}_{j,n,k}
- p_{max}\bigg)\nonumber\\
&+\boldsymbol{\zeta}\bigg(\sum_{n\in \mathcal{N}}\log_2\big(1+\frac{\tilde{q}_{j,n,k}\ |h^{ID}_{j,n,k}|^2}{{|\sigma_{j,n,k}^{{ID}}|^{^2}}+I^{n}_{max}}\big)-R_{min}\bigg)
\nonumber \\ \nonumber
&+\boldsymbol{\tau}\bigg(\epsilon_{j,k}\sum_{n \in \mathcal{N}} \tilde{q}_{j,n,k}|g^{EH}_{j,n,k}|^{2}-\text{EH}_{min}\bigg)\\
&-\boldsymbol{\theta}\big(\sum_{\substack{j'\neq j \\ j' \in \mathcal{J}}}\sum_{\substack{k'\neq k \\ k' \in \mathcal{K}}}\tilde{p}_{j',n,k'}|h^{ID}_{j',n,k}|^2- I^{n}_{max}\big),
\end{align}
where $~\boldsymbol{\chi},~\boldsymbol{\phi},~\boldsymbol{\zeta},~\boldsymbol{\tau},~\boldsymbol{\theta}$ are the Lagrangian vectors associated with the constraints.
Specifically, the Lagrange multiplier vector with respect to the OFDMA constraint, i.e., the first constraint $C_{1}$, has its elements as $\chi_{j,n}$'s where $j \in \{1,2,...,J\}$ and $n \in \{1,2,..., N\}$.
The $\boldsymbol{\phi}$ is the Lagrange multiplier vector accounting for the maximum transmit power constraint $C_{2}$ with {$\boldsymbol{\phi}_{j}$'s where $j \in \{1,2,...,J\}$}.
The Lagrange multiplier vector $\boldsymbol{\zeta}$ that corresponds to the data-rate constraint $C_{3}$ posses the elements $\zeta_{j,k}$'s, where $j \in \{1,2,...,J\}$ and $k \in \{1,2,..., \mathcal{K}_{I}\}$.
The vector $\boldsymbol{\tau}$ is the Lagrange multiplier vector for the constraint $C_{4}$ with the elements $\tau_{j,k}$'s, where {$j \in \{1,2,...,J\}$ and $k \in \{1,2,..., \mathcal{K_{E}}\}$}.
Finally, the Lagrange multipliers vector for the interference threshold constraint $C_{6}$ is $\boldsymbol{\theta}$ that has components $\theta_{j,n,k}$'s, where ${j \in \{1,2,...,J\}}$,~${n \in \{1,2,..., N\}}$, and ${k \in \{1,2,..., \mathcal{K_{I}}\}}$.
It should be pointed out that the boundary constraints are absorbed into the Karush-Kuhn-Tucker (KKT) conditions when deriving the resource allocation policy.
Thus, the dual problem of (\ref{F4-32:main}) is given by
\begin{align}
\min_{\boldsymbol{\chi},\boldsymbol{\phi},\boldsymbol{\zeta},\boldsymbol{\tau},\boldsymbol{\theta}} \ \ 
\max_{\textbf{a},\mathbf{p},\tilde{\textbf{q}}}\ \
\mathcal{L}(\textbf{a},\mathbf{p},\tilde{\textbf{q}},\boldsymbol{\chi},\boldsymbol{\phi},\boldsymbol{\zeta},\boldsymbol{\tau},\boldsymbol{\theta}).
\label{F4-39}
\end{align}
In the following, we solve the above dual problem iteratively by decomposing it into two layers.
The first layer, Layer 1, consists of subproblems with identical structures while the second layer, Layer 2, is the master dual problem to be solved with the gradient method.

\textbf{\textit{Dual Decomposition and Layer 1 Solutions}:}
By dual decomposition, the first layer can be written as follows
\begin{align}
\mathcal{D}(\boldsymbol{\chi},\boldsymbol{\phi},\boldsymbol{\zeta},\boldsymbol{\tau},\boldsymbol{\theta})=
\max_{\textbf{a},\mathbf{p},\tilde{\textbf{q}}} \ \ 
\mathcal{L}(\textbf{a},\mathbf{p},\tilde{\textbf{q}},\boldsymbol{\chi},\boldsymbol{\phi},\boldsymbol{\zeta},\boldsymbol{\tau},\boldsymbol{\theta}).
\label{F4-40}
\end{align}
For a fixed set of Lagrange multipliers, (\ref{F4-40}) is a convex optimization problem, for which a unique optimal solution can be obtained using the Lagrange dual function.
Forming the Lagrangian, taking the derivative of the Lagrangian with respect to $\tilde{\textbf{q}}$ and setting the derivative equal to zero, the transmit power $\tilde{\textbf{q}}$ is obtained.
Using standard optimization techniques and the KKT conditions, the power allocation for user $k$ on subcarrier $n$ in the cell $j$ is obtained as
\begin{align}
\tilde{q}^*_{j,n,k}=
a_{j,n,k}{p}^*_{j,n,k}=a_{j,n,k}
\Bigg[\frac{1}{\ln{(2)}}
\Bigg(
\frac{1+\zeta_{j,k}}{\phi_{j} + 
\theta_{j,n,k}|h^{ID}_{j,n,k}|^2- 
\tau_{j,k}\epsilon_{j,k}|g^{EH}_{j,n,k}|^2}
\Bigg)-
\frac{{|\sigma_{j,n,k}^{{ID}}|^{^2}}+I^{n}_{max}}{|h^{ID}_{j,n,k}|^2}\Bigg]^{+}.
\label{F4-41}
\end{align}  
The power allocation has the form of multilevel water-filling.
It can be seen that the data-rate prevents energy inefficient transmission by truncating the water-levels~\cite{David-tse}.

In order to obtain the optimal subcarrier allocation, we take the derivative of the subproblem's objective function with respect to $a_{j,n,k}$, that is
\begin{equation}
\frac{\partial{}\mathcal{L}
(\textbf{a},\mathbf{p},\tilde{\textbf{q}},\boldsymbol{\chi},\boldsymbol{\phi},\boldsymbol{\zeta},\boldsymbol{\tau},\boldsymbol{\theta})}{\partial a_{j,n,k}}
\bigg|_{\tilde{q}^*_{j,n,k}}=
\mathscr{S}_{j,n,k},
\end{equation}
where $\mathscr{S}_{j,n,k}\geq 0$ can be interpreted as the marginal benefit, as discussed in \cite{975766}, for allocating subcarrier $n$ to user $k$ and is given by
\begin{align}
\mathscr{S}_{j,n,k}=~ &
(1+ \zeta_{j,k})
\Bigg[
\log_2\Bigg(1+\frac{ {p}^*_{j,n,k} |h_{{j,n,k}}|^2}
{{|\sigma_{j,n,k}^{{ID}}|^{^2}}+I^{n}_{max}}\Bigg) -
\frac{1}{\ln(2)}
\Bigg(\frac{ {p}^*_{j,n,k} |h_{j,n,k}|^2}{\tilde{p}^*_{j,n,k} |h_{{j,n,k}}|^2
+{|\sigma_{j,n,k}^{{ID}}|^{^2}}+I^{n}_{max}}\Bigg)
\Bigg]\nonumber\\&+ 
\tau_{j,k}\epsilon_{j,k} {p}^*_{j,n,k} |g^{EH}_{j,n,k}|^2- 
\theta_{j,n,k}{p}^*_{j,n,k} |h^{ID}_{j,n,k}|^2- 
\phi_{j}{p}^*_{j,n,k}- 
\chi_{j,n}.
\label{F4-42}
\end{align}
It should be noted that $ \mathscr{S}_{j,n,k}\geq 0$ has a physical meaning that users with negative scheduled data-rate on subcarrier $n$ are not selected as they can only provide a negative marginal benefit to the system.
Subsequently, the subcarrier allocation should satisfy the following region
\begin{align}
\label{F4-43}
a^{*}_{j,n,k} = \left\{ \begin{array}{ll}
1,~~~\textrm{if}~~\mathscr{S}_{j,n,k}\geq \chi_{j,n},\\
0,~~~ \textrm{otherwise}.
\end{array} \right.
\end{align}

\textbf{\textit{Solution of Layer 2 Master Problem}:}
To find the optimum subcarrier assignment (\ref{F4-43}), we must first determine the threshold $\mathscr{S}_{j,n,k}$.
However, the subcarrier assignment depends on the Lagrangian variables $\tilde{q}_{j,n,k}$.
Therefore, we employ the subgradient method to find the Lagrangian multipliers for a given $\tilde{\textbf{q}}$.~Hence, we have
\begin{align}
\phi_{j}^{i+1}&=
\bigg[\phi_{j}^{i}+
\alpha_{1}
\bigg(\sum_{k\in \mathcal{K}}~\sum_{n\in \mathcal{N}}\tilde{q}_{j,n,k}- p_{max}\bigg)\bigg]^{+},
\label{F4-44}\\
\zeta_{j,k}^{i+1}&=
\bigg[\zeta_{j,k}^{i}-
\alpha_{2}
\bigg(\sum_{n \in \mathcal{N}}
\log_2\Big(1+\frac{\tilde{q}_{j,n,k}\ |h_{j,n,k}|^2}{|\sigma_{j,n,k}^{{ID}}|^{^2}+
I^{n}_{max}}\Big)-R_{min}\bigg)\bigg]^{+},
\label{F4-45}\\
\tau^{i+1}_{j,k}&=
\bigg[\tau_{j,k}^{i}+
\alpha_{3}
\bigg(\epsilon_{j,k}\sum_{n \in \mathcal{N}}{\tilde{q}_{j,n,k}}|g^{EH}_{j,n,k}|^{2}-
\text{EH}_{min}\bigg)\bigg]^{+}, 
\label{F4-46-1}\\
\theta_{j,n,k}^{i+1}&=
\bigg[\theta_{j,n,k}^{i}-
\alpha_{4}
\bigg(\sum_{\substack{j'\neq j \\ j' \in \mathcal{J}}}
\sum_{\substack{k'\neq k \\ k' \in \mathcal{K}}}\tilde{p}_{j',n,k'}|h^{ID}_{j',n,k}|^2- I^{n}_{max}\bigg)\bigg]^{+},
\label{F4-46}
\end{align}
where index $i\geq 0$ is iteration index, and $\alpha_{q}$'s, $q \in \{1,2,3,4\}$, are positive step sizes.
The details of the low complexity algorithm are sketched in \textbf{Algorithm \ref{A_Low_Complexity_Algorithm_Design}}.

\begin{algorithm}[H]
\caption{Low Complexity Power Control and Subcarrier Assignment}
\label{A_Low_Complexity_Algorithm_Design}
\begin{algorithmic}[1]
\STATE {$\mathbf{Initialize}$} \\{
\begin{addmargin}[1em]{0em}
{iteration index $i=0$ with the maximum number of iteration $\mathcal{I}_{max}$\\
and Lagrangian variables vectors $\boldsymbol{\chi},\boldsymbol{\phi},\boldsymbol{\zeta},\boldsymbol{\tau},\boldsymbol{\theta}$} for a feasible set vector \{$\mathbf{a}^0,{\mathbf{p}}^0$,$\tilde{\mathbf{q}}^0$\}.
\end{addmargin}}
\STATE{\textbf{repeat} }
\STATE{
\begin{addmargin}[1em]{0em}
Update power allocation policy using (\ref{F4-41}).
\end{addmargin}}
\STATE{
\begin{addmargin}[1em]{0em}
Calculate $\mathscr{S}_{j,n,k}$ based on (\ref{F4-42}) and find optimal subcarrier assignment using (\ref{F4-43}). 
\end{addmargin}}
\STATE{
\begin{addmargin}[1em]{0em}
Update  Lagrangian variables vectors based on (\ref{F4-44})-(\ref{F4-46}).
\end{addmargin}}
\STATE{
\begin{addmargin}[1em]{0em}
Set $i=i+1$.
\end{addmargin}}
\STATE \textbf{until} Convergence or $i=\mathcal{I}_{max}$
\STATE{\textbf{return}~\{$\mathbf{a}^{*},{\mathbf{p}}^{*}$,$\tilde{\mathbf{q}}^*$\}}
\end{algorithmic}
\end{algorithm}

\section{Simulation Results}
In this section, the performance gain of the proposed subcarrier and power allocation algorithms for SWIPT in the DL direction of a multi-cell multi-user OFDMA system is evaluated through extensive simulations.
There are $J=3$ cells in the network topology with $K_{j}=4$ users in each cell.
From the considered four users in each cell with ring-shaped boundary regions, two are uniformly and randomly located inside the inner-circle while two are in the outer-zone, i.e., ${K}^{EH}_j=2={K}^{ID}_j=2$ $\forall j \in \mathcal{J} = \{1,2,3\}$.
We set the radius of a cell, $d_{max}$, as 20 meters, with a reference distance, $d_{0}$, of 5 meters, where the EH users are placed in the interval (0, $d_{0}$] while ID used are inside ($d_{0}$, $d_{max}$).
Additionally, we consider a frequency-selective fading channel and further assume the central carrier frequency is set to be 3 GHz.
The number of subcarriers is $N = 16$, where the bandwidth of each subcarrier is set to 180 kHz.
It should be noted that as the power of the background noise for both EH and ID receivers is rather small compared to maximum transmit power, $p_{max}$, it is assumed to have {$|\sigma_{j,n,k}^{{EH}}|^{^2} = |\sigma_{j,n,k}^{{ID}}|^{^2} = \sigma^{2} = $ -120 dBm} in all simulations.
Since a line-of-sight (LoS) signal is expected in the received signal, the small-scale fading channel is modeled as Rician fading with Rician factor $\rho=3$ dB.
\begin{table}[!b]
{\caption{Simulation Parameters}
\label{chap:4:Simulation_Parameters}
\centering
\begin{tabular}{|c|c|}\hline
{\bf Parameter} & {\bf Value} \\ \hline \hline
{Cell coverage ($d_{max}$)} & {$20$ m} \\ 
{Reference distance ($d_{0}$)} & {$5$ m} \\ 
{The number of cell ($J$)} & {$3$}\\
{The number of ID user in each cell ($K^{ID}_{j}$)} & {$2$}\\
{The number of EH user in each cell ($K^{EH}_{j}$)} & {$2$}\\
{The number of subcarrier (N)} & {$16$} \\
{Noise power ($\sigma^{2}$)} & {$-120$} dBm \\
{The bandwidth of each subcarrier} & {$180$ kHz}\\
{Path loss exponent ($\alpha$)} & {$2.76$} \\
{Path loss model for cellular links} & {$31.7+27.6 \log(\frac{d}{d_0})$} \\
{Multi-path fading distribution}& {Rician fading with factor 3~dB}\\
{Power conversion efficiency ($\epsilon$)}& {30\%}\\ 
The maximum transmit power of the SBS ($p_{\textrm{max}}$) & {$30$ dBm} \\
The minimum data-rate requirement for $k^{th}$ ID user~($R_{\textrm{min}}$) & $1$ bps/Hz \\
{The minimum harvested requirement ($\text{EH}_{min}$)}& {$0$~dBm}\\
{The maximum interference threshold ($I^n_{max}$)}&$-70$~dBm\\
Channel realization number & $100$\\\hline
\end{tabular}}
\end{table}
Moreover, the Rician flat fading channel gains include a distance-dependent path loss component of $31.7+10 \alpha \log(\frac{d}{d_0})$ [dB] (where $d$ is the distance between the transmitter and the receiver) and a log-normal shadowing component with~$8$ dB standard deviation where the path loss exponent is equal to $\alpha = 2.8$~\cite{339880}.
These parameters for propagation modeling and simulations follow the suggestions in 3GPP evaluation methodology~\cite{chap3:3GPP}.
The power conversion efficiency of all EH users, $\epsilon_{j,k}$, is assumed to be the same and is equal to  $\epsilon_{j,k} = \epsilon = 0.3 $.
The target transmission rate $R_{min}=1$ bps/Hz for each ID used unless otherwise stated.
The minimum harvested energy is EH$_{min}$ for each EH user.
Besides, a maximum interference threshold of -70 dBm is considered for the low complexity design algorithm~\cite{6294504}.
Furthermore, we conduct Monte Carlo simulations by generating random realizations of the channel gains to obtain the average data-rate of the network.
In fact, the channel gain between a transmitter and a receiver is calculated using independent and identically distributed Rician flat fading and the figures shown in this section are obtained by estimating the average of results over different realizations of path-loss as well as multi-path fading.
The rest of the simulation parameters are given in \textbf{Table}~(\ref{chap:4:Simulation_Parameters}) unless otherwise is specified.

\begin{figure}[!b]
\centering
\includegraphics[width=12cm]{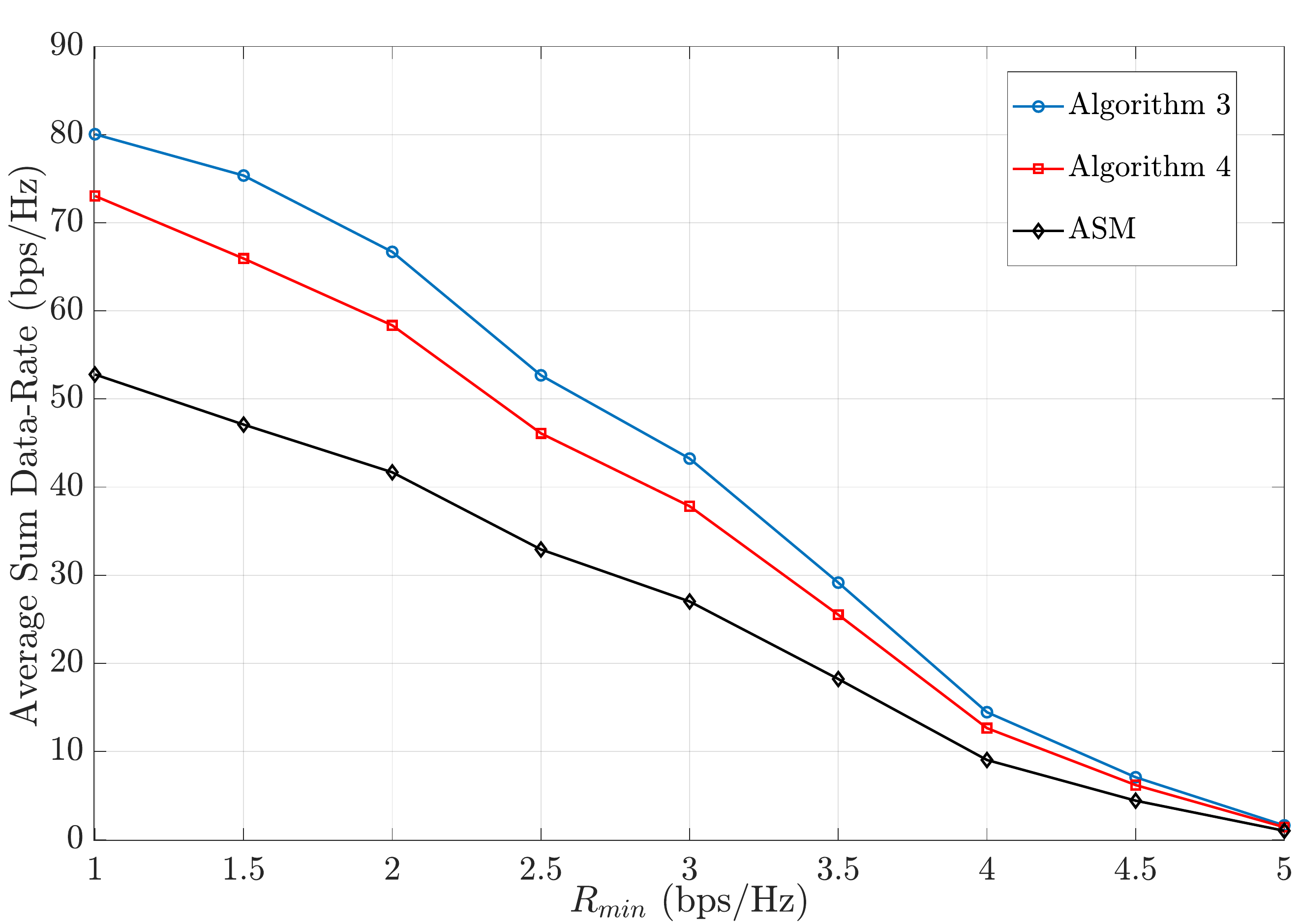}
\caption{Average sum data-rate versus minimum data-rate requirement.}
\label{plot:4.1}
\end{figure}
\subsection{Average Sum Data-rate versus Minimum Data-rate Requirement}
In figure (\ref{plot:4.1}), the average sum data-rate versus minimum data-rate requirement is depicted. 
It is illustrated that as $R_{min}$ increases, the average sum data-rate decreases. 
The reason is that when $R_{min}$ is high, more subcarrier have to be assigned to ID users to satisfy the minimum data-rate requirement, especially to those users with poor channel conditions in an extremely deep fade.
This is in conjunction with the necessity of higher transmit power for reaching a certain data-rate.
However, more transmit power also means more interference that obstructs the data transmission.
This excess interference adversely shows itself in data-rate functions that would substantially decrease the achievable data-rate.
In fact, as the maximum transmit power increases, the interference power arising from co-channel becomes more severe, degenerating the received user signals.
We also observe that \textbf{Algorithm~\ref{euclid}} outperforms both \textbf{Algorithm~\ref{A_Low_Complexity_Algorithm_Design}} and the alternative search method (ASM).
It can be concluded that the proposed \textbf{Algorithm~\ref{euclid}} has considerably better performance due not only to performing a joint resource allocation policy, but also acknowledging the interference term as a variable in data-rate functions, which in turn stresses the dependency between power and subcarrier allocation.
For the ASM, we employ a heuristic search method in which we decouple the problem of joint subcarrier assignment and power allocation to maximize the system throughput based on~\cite{jalal}.

\begin{figure}[!b]
\centering
\includegraphics[width=12cm]{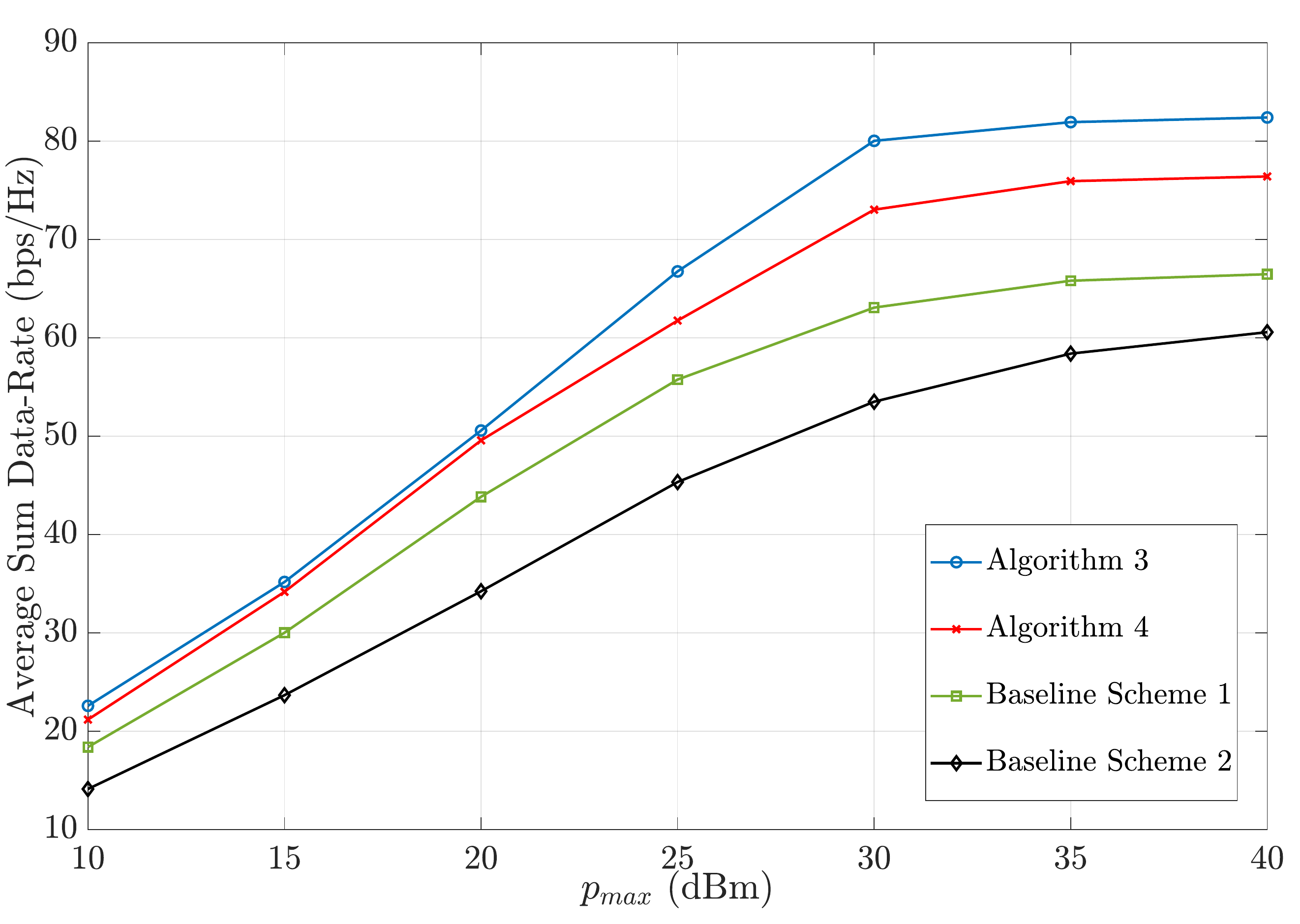}
\caption{Average sum data-rate versus maximum transmit power.}
\label{plot:4.2}
\end{figure}
\subsection{Average Sum Data-rate versus Maximum Transmit Power}
Figure (\ref{plot:4.2}) shows the average system data-rate versus maximum allowance transmit power, $p_{max}$, for SBSs.
It can be seen from this figure that the average sum data-rate increases by raising the maximum transmit power. We also observe that the average sum data-rate grows monotonically up to 35 dBm.
Nevertheless, the slope of the curve in the average sum data-rate is declined as the maximum transmit power gets larger values.
Particularly, the average sum data-rate starts to saturate when $p_{max}\geq 35$ dBm.
However, there exists an important point that should be kept in mind: increasing the transmit power does not improve data-rate perpetually.
More transmit power also means more intensified interference level that hampers the information transmission. 
These excess interference terms from the co-channel show itself in the data-rate functions negatively.
Therefore, the slope of the sum data-rate tends to decrease at some point, which bounds the maximum achievable data-rate to an almost constant value.
Furthermore, for comparison and evaluation of our proposed method, we consider two baseline schemes.
Baseline scheme 1 is based on the decoupling of the subcarrier assignment and power allocation variables in which the original problem is divided into two disjoints optimization problems~\cite{jalal}.
For baseline scheme~2, only power allocation is performed while the subcarrier assignment is done randomly.
It can be seen that the proposed \textbf{Algorithm~\ref{euclid}} outperforms the other methods due to solving the optimization problem jointly based on the MM approach. 
This yields a close-to-optimal solution.

\subsection{Average Sum Data-rate Versus Number of Iterations}
In figure (\ref{plot:4.4}), we examine the convergence behavior of our proposed iterative method via the MM approach under different initialization of power.
It can be observed that \textbf{Algorithm~\ref{euclid}} has a quick convergence at an equal power allocation over all subcarriers, that is $\mathbf{p}^0(i) = \frac{p_{max}}{N}$, whereas it requires a little more number of iterations for zero power, i.e., the extreme case with $\mathbf{p}^0(i) = 0$.
This figure also demonstrates even though the speed of convergence differs from one case to another, our proposed method quickly converges to a stationary point only after a small limited number of iterations. 

\begin{figure}[!b]
\centering
\includegraphics[width=12cm]{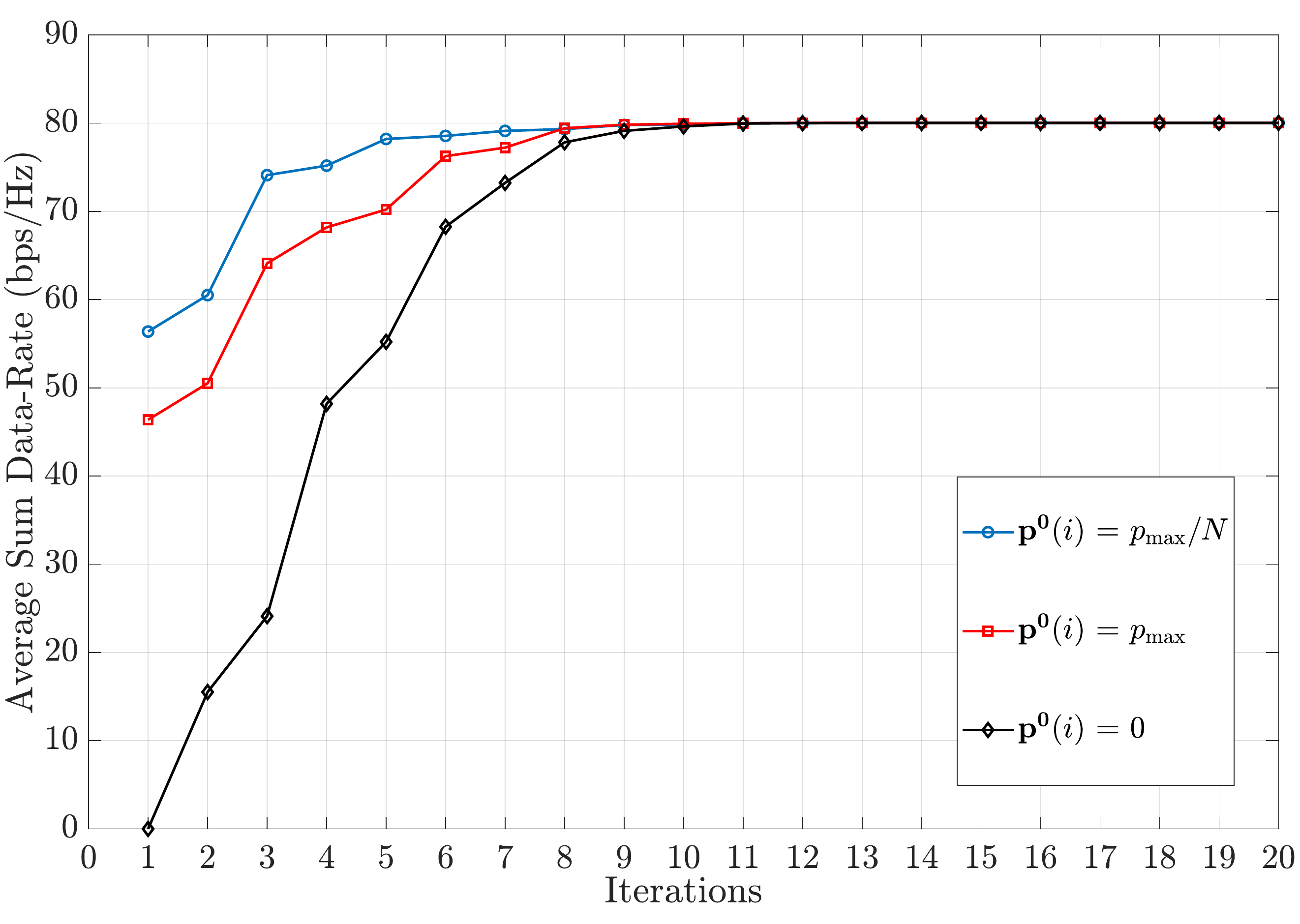}
\caption{Average sum data-rate versus number of iteration under different initialization of the power.}
\label{plot:4.4}
\end{figure}
\begin{figure}[!t]
\centering
\includegraphics[width=12cm]{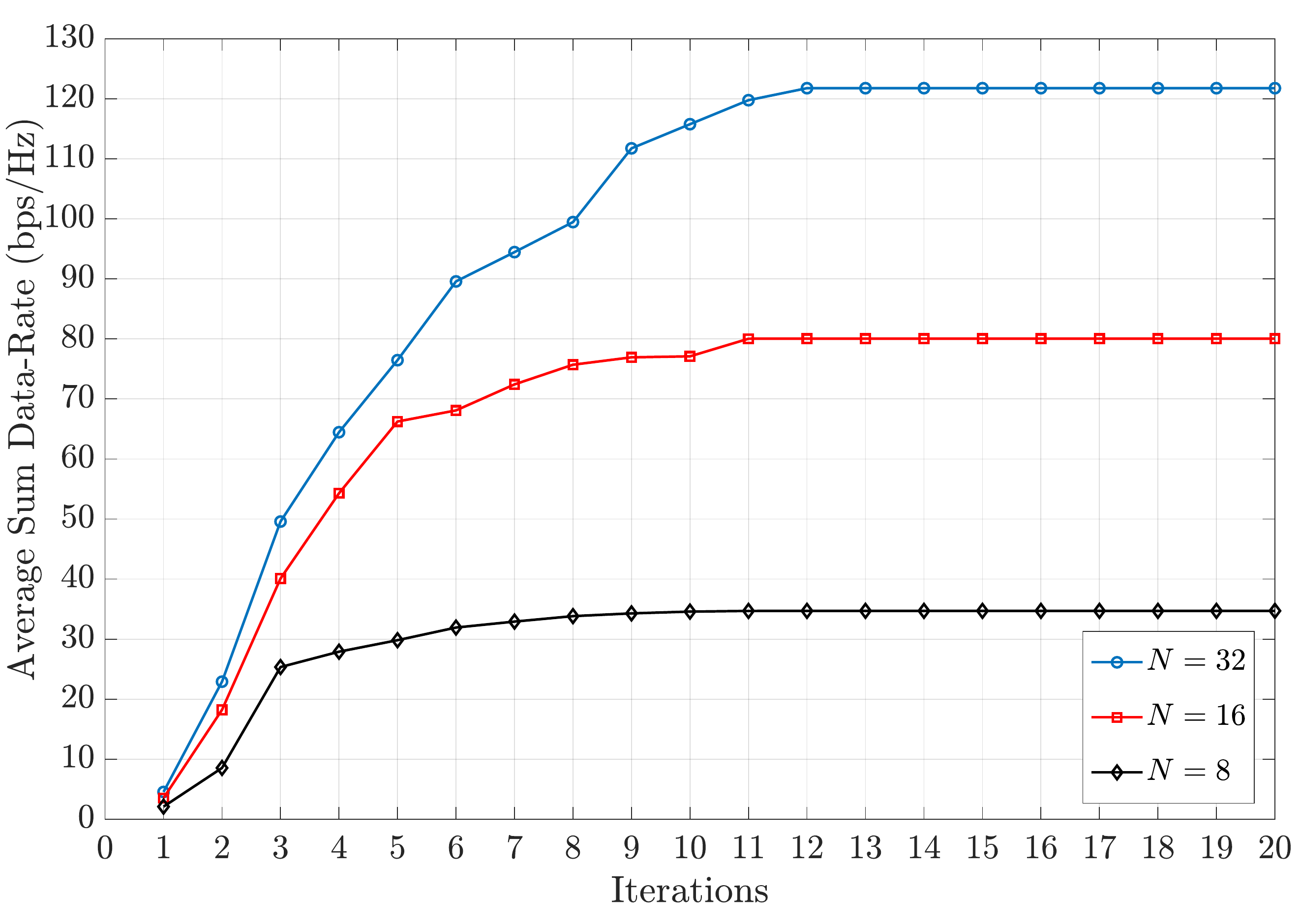}
\caption{Average sum data-rate versus number of iteration under different number of subcarriers.}
\label{plot:4.3}
\end{figure}
To give another perspective, figure (\ref{plot:4.3}) is plotted that shows the average sum data-rate versus the number of iteration for a different number of subcarriers.
It can be seen that as the number of subcarrier increases, the average sum data-rate increases as well.
This is because with increasing the number of subcarriers, an ID user is more likely to choose subcarriers with higher channel quality gains, leading to more significant system throughput in agreement with the effect of channel diversity.
Another interesting observation in this figure is the convergence behavior of our proposed algorithm.
It should be emphasised that the number of subcarriers affects convergence behavior of \textbf{Algorithm~\ref{euclid}}.
The more the number of subcarriers, the more iterations are needed for our proposed algorithm to converge.
The reason for this can be attributed to an expansion in the search region of the optimization problem in (\ref{F4-26-}).
Due to the existence of more binary subcarriers allocation variables that must accordingly be assigned to ID users, extra iterations are required for our proposed algorithm based on the MM approach to settle in a stationary point.  

\subsection{Average Sum Data-rate versus Number of Cells}
Figure (\ref{plot:4.5}) investigates the sum data-rate of $J$ small-cells when the number of small-cells varies from 3 to 7.
With the number of small-cells going up, the average sum data-rate improves even more, although the total number of subcarrier is fixed in the entire network.
This is because each subcarrier has more candidate users of small-cells to choose from (according to the channel quality between the SBS and each ID user), when the number of small-cells increases.
This is known as multi-user diversity.
Consequently, each small-cell throughput improves, which results in an enhancement of the average sum data-rate of the whole network.
Moreover, we can further conclude from the same figure that \textbf{Algorithm~\ref{euclid}} performs better in comparison to other algorithms due to designing a joint resource allocation framework. 

\begin{figure}[!t]
\centering
\includegraphics[width=12cm]{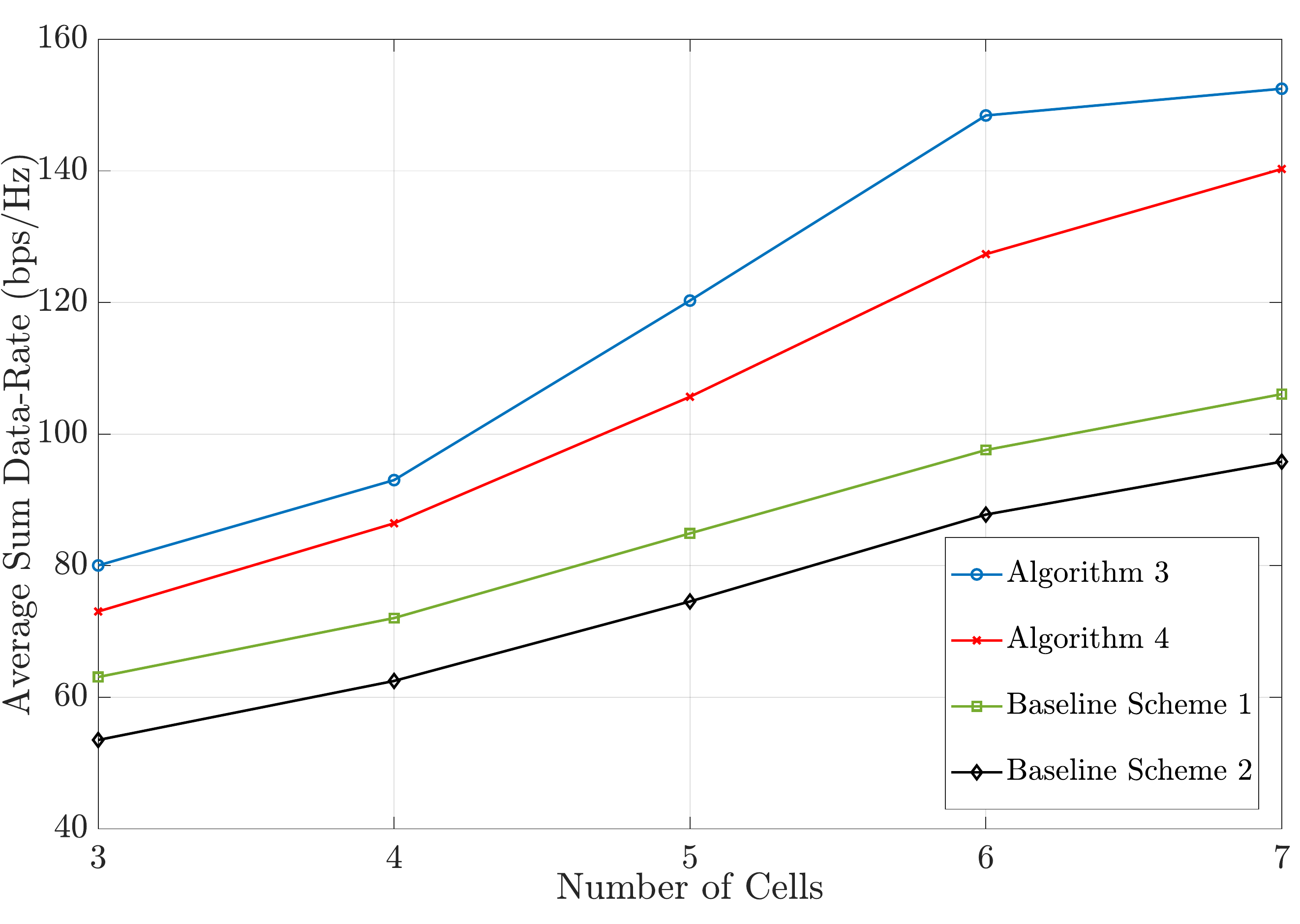}
\caption{Average sum data-rate versus number of cells.}
\label{plot:4.5}
\end{figure}

\section{Summary}
In this chapter, we studied throughput or data-rate maximization for indoor SWIPT-enabled OFDMA multi-user multi-cell networks.
In particular, we considered the separated receiver architecture in which the ID and EH users are separated in the coverage area of a SBS and located in two distinct regions.
Taking into account the subcarrier assignment and power allocation, a resource allocation problem was formulated to maximize data-rate while respecting the minimum required data-rate for each ID user and minimum harvesting energy for each EH user.
The underlying problem was non-convex MINLP.
We employed the MM approach, where a surrogate function serves to approximate the non-convex term.
Since the computational complexity of MM approach was high, we also proposed a suboptimal subcarrier assignment and power allocation algorithm to solve the problem with lower complexity .
Through simulation results, we demonstrated the excellent performance of our proposed algorithms compared to state-of-the-art algorithms that have been addressed in the literature.
Furthermore, numerical results clearly demonstrated that our proposed scheme would substantially improve the system throughput, although it converges to a stationary point even after small number of iterations.
\newpage\thispagestyle{empty}\mbox{}  
\chapter{Generalized Antenna Switching Technique in SWIPT}
\label{CHAP5}
\fancyhf{}
\renewcommand{\headrulewidth}{2pt}
\fancyhead[LE,RO]{\thepage}
\fancyhead[RE]{\textit{ \nouppercase{\leftmark}} }
\fancyhead[LO]{\textit{ \nouppercase{\rightmark}} }
\renewcommand{\footrulewidth}{0.1pt}
\fancyfoot[CE,CO]{\nouppercase{\leftmark}}
\fancyfoot[LE,RO]{JFMJ}

\vspace{15mm}

Most works about simultaneous wireless information and power transfer (SWIPT) focus either on maximizing energy harvesting, which we discuss in chapter~\ref{CHAP3}, or maximizing throughput, in chapter~\ref{CHAP4}.
Nevertheless, deploying algorithms to reap more harvested energy in the overall network topology adversely affects information transfer, and this in turn causes the quality of service (QoS) of the system to degenerate.
Besides, global commitments to sustainable development are contravened if the system design merely seeks to maximize the spectral efficiency (SE) by improving throughput in light of the inexorable increase of network power consumption.
Therefore, energy efficiency (EE) maximization is the focus of this chapter.
It should be noted that EE, conventionally defined as the quantity of transmitted information bits per unit energy (bits/joule), and the need for a high SE are key performance indicators in communication networks.

More features can be added to SWIPT networks by employing multiple antennas at the transmitter and receiver – commonly referred to as multiple-input multiple-output (MIMO) systems.
With respect to receivers, multiple receive antennas can help users harvest more energy because of the broadcast nature of wireless transmission.
Furthermore, the efficiency of information and energy transfer can be significantly improved by using multiple transmit antennas.
That is why it is not surprising that MIMO is regarded as a promising technology – not only for improving network throughput and radio communication system reliability, but also because of the distinct features of SWIPT systems.
However, pre-processing and post-processing are required at both transmitters and receivers in multiple antenna systems, which increases the cost and complexity of system designs.
Various antenna selection methods have been proposed in~\cite{jalal,13,21} as low-cost and simpler solutions for exploiting the performance gain, by achieving a diversity gain, that is promised by multiple antenna systems. 

\begin{figure}[!t]
\vspace*{10mm} 
\centering
\hspace*{-7mm} 
\includegraphics[width=17.1cm,trim=4 4 4 4,clip]
{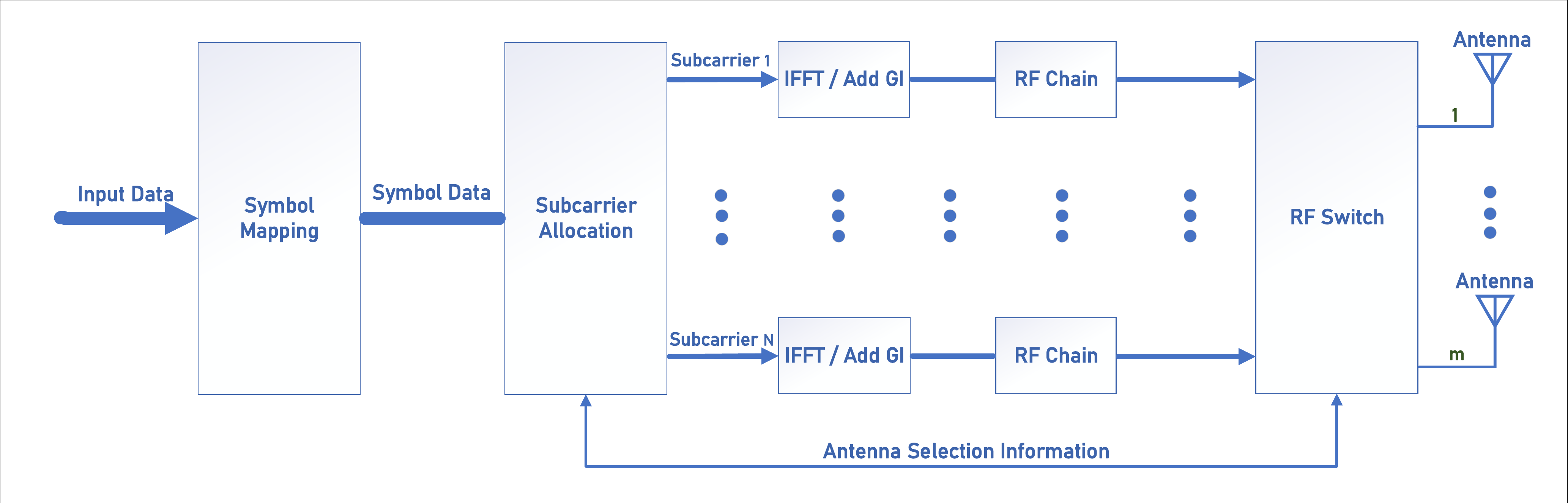}
\caption{Antenna selection architecture.}
\label{fig:5-1}
\end{figure}

The fundamental idea of antenna selection is that the limited available radio frequency (RF) chains, which provide wireless links with the most reliable signal-to-noise ratio (SNR), are assigned to transmit and receive antennas.
The building block of an antenna selection unit is depicted in figure (\ref{fig:5-1}).
Since antenna selection is considered a novel technique that provides higher flexibility to the system operator while also offering a better policy for resource allocation, we briefly explain its basic working principle below.
As can be seen, the input data stream is mapped to the corresponding symbol data at the transmitter by the symbol mapping block.
A symbol data frame, generated by the symbol mapping block, goes to the subcarrier allocation block, which allocates the data to the antenna selected by the RF switch associated with a specific subcarrier.
Then the output sequences from the subcarrier allocation block are applied to inverse fast Fourier transform (IFFT) blocks, and a guard interval (GI) is added to each time-domain signal being transmitted by its respective transmit antenna.

There are two basic approaches to deploying antenna selection in OFDMA systems: bulk selection and per-subcarrier selection. 
In bulk selection, the same antennas are chosen for all subcarriers, whereas in per-subcarrier selection the antennas for each subcarrier are selected independently~\cite{6111188}.
The same principle holds for antenna selection at the receiver.
This subject is developed in \cite{7100915}.
The uplink (UL) of the fourth generation (4G) standard of long term evolution, known as LTE-Advanced, uses the antenna selection technique because of its low implementation cost and the small amount of feedback required – compared to existing techniques such as beamforming and precoding \cite{6316788}.

To the authors' best knowledge, antenna selection in SWIPT networks with a focus on resource allocation design has not yet been studied.
We therefore present a theoretical study of an EE optimization problem that considers an antenna selection approach in a SWIPT system.
In the introduction we defined the antenna switching (AS) scheme as a SWIPT enabler, whereby the user is equipped with independent antennas for energy harvesting (EH) and information decoding (ID) operations.
The antenna selection technique at the receiver can be viewed as a generalization of the AS scheme in a co-located SWIPT network with distinct subsets of antennas that could be assigned to EH and ID operations if permitted by channel quality and the channel state information (CSI).
We refer to this observation as a “generalized AS technique” in SWIPT.
By this we mean that the generalized AS acts as a “switch” in the operation mode of the antenna; each antenna is capable of both ID and EH operations. 
To further explore the EE of the generalized AS-based SWIPT system, we study this scheme in a downlink (DL) of a multi-user multi-cell orthogonal frequency division multiple access (OFDMA) network. 
Intuitively, a higher data-rate and more harvested energy would be expected if more receive antennas were activated. 
The objective function of this chapter's optimization problem is under the condition of satisfying the minimum data-rate requirement and respecting the maximum power transfer constraints. This is achieved by jointly optimizing subcarrier assignment and power allocation with the receive active antenna set selection. 

This particular EE optimization problem is complicated because it is non-convex and fractional-combinatorial.
We begin by dividing the main problem into two subproblems.
The first subproblem concerns the joint small base station (SBS)-subcarrier assignment and power allocation while the second subproblem seeks to determine the best antennas based on the scheduling (joint SBS-subcarrier assignment and power allocation) chosen for ID or EH operations.
We confirm the validity of our theoretical findings for the generalized AS-based SWIPT systems by providing simulation results to draw design insights and demonstrate how our proposed algorithm achieves excellent performance while also reducing the computational cost at the receiver.

\section{System Model}
In this section, we consider a DL of an OFDMA network in a multi-user multi-cell network using SWIPT, as shown in figure (\ref{fig:5.1}).
We assume that the coverage area of the specific region is provided via a set of $j \in {\mathcal{J} = \{1,2,...,J\}}$ cells with one serving SBS in each cell.
Moreover, only one single transmit antenna is handling the corresponding users in the associating cell, even though all SBSs are equipped with multiple antennas.
We additionally assume that the entire frequency band of $\mathscr{B}$ is divided into $N$ orthogonal subcarrier, each having a bandwidth of $\mathscr{W}$.
All cells share the subcarrier set $\mathcal{N}=\{1,2,..,N\}$, where $|\mathcal{N}|=N$ indicates the total number of subcarriers.
The set of all users in a given cell $j$ is represented by $\mathcal{K}_{j}=\{1,2,...,K_j\}$, where the total number of users in that cell is $|\mathcal{K}_{j}|=K_j$.
Furthermore, $\mathcal{K}=\sum_{j\in\mathcal{J}}\mathcal{K}_j$ gives the total number of users in the network.
Besides, each user is equipped with multiple antennas, where the set of antennas is represented by $m \in \mathcal{M}=\{1,2,...,M\}$ with $|\mathcal{M}|=M$ for each user.
It is also assumed that the perfect CSI is available at the resource allocator to design the resource allocation policy.
Note that there is a centralized controller that is connected to all the SBSs.
Specifically, it is presumed that each SBS broadcasts orthogonal preambles, pilot signals, in the DL to the users.
Then, through a feedback channel, each user estimates the CSI and transfers this information back to the associated SBS.
Afterward, the corresponding SBS listens to the sounding reference signals communicated by the users and sends the CSI to the centralized controller for resource allocation design.
\begin{figure}[p]
\centering
\includegraphics[width=21cm,trim=4 4 4 4,clip,angle=90]
{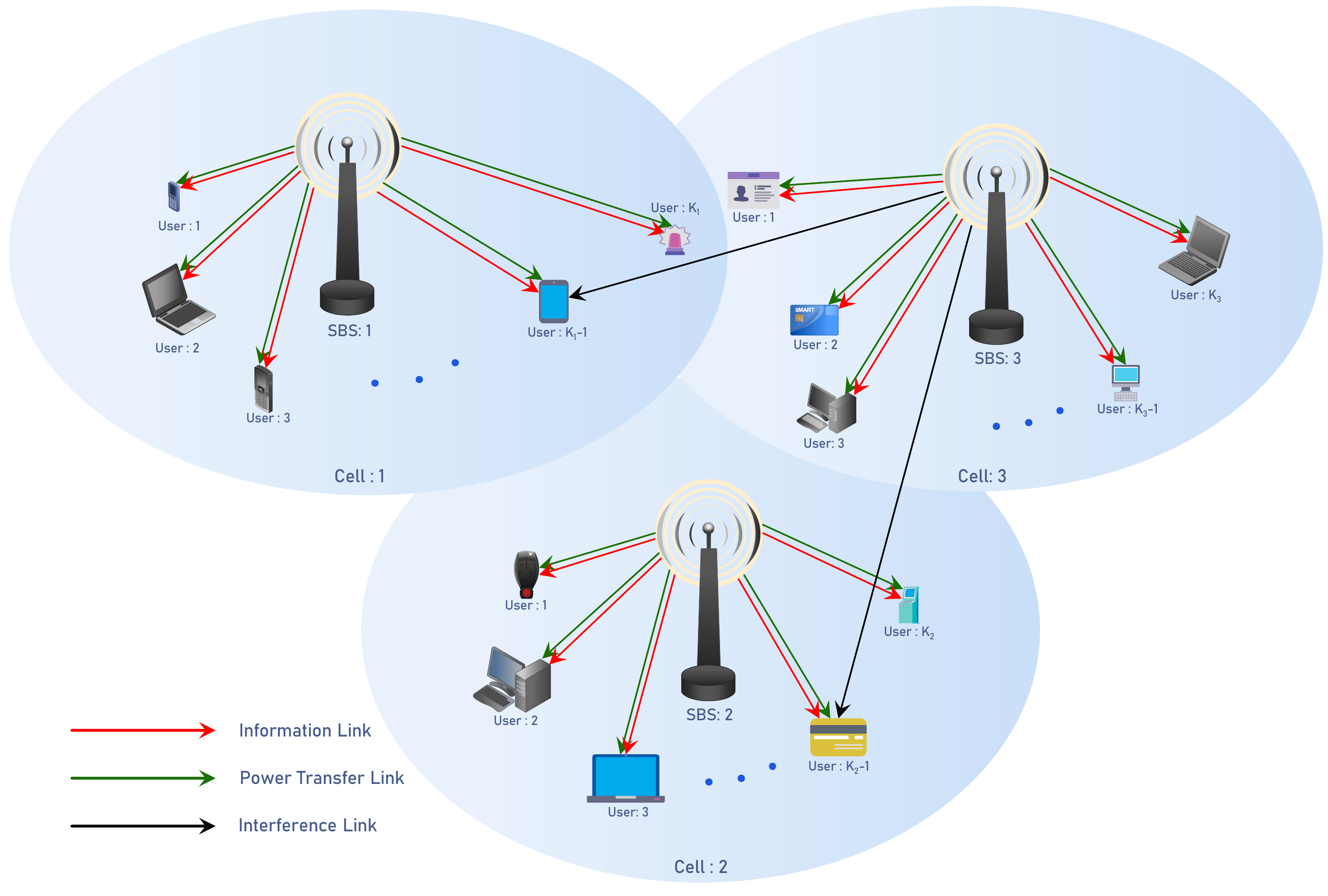}
\caption{SWIPT in a DL of a multi-user multi-cell OFDMA network with generalized AS receivers.}
\label{fig:5.1}
\end{figure}

Furthermore, with multiple antenna setting in each user, the best antenna can be selected for both ID and EH operation in each subcarrier based on the optimization problem. 
However, these operation modes cannot be done over the same antenna at the same time.
Beside, we assume a per-subcarrier selection method for the generalized AS approach in this chapter. 
In this regard, all the assigned subcarriers can be used to receive signals from different antennas resulting in a better degree of freedom by increasing the throughput of the system.~For the sake of readability, we first introduce some of the essential parameters that are used to describe the system model:   
\begin{itemize}
\item $h^{m}_{j,n,k}$: The DL channel gain for the wireless information transfer from the $j^{th}$ SBS to the user $k$ using its $m$ antenna over the $n^{th}$ subcarrier.
\item $g^{m}_{j,n,k}$: The DL channel gain for the wireless power transfer from the $j^{th}$ SBS to the user $k$ using its $m$ antenna over the $n^{th}$ subcarrier.
\item $x^{m}_{j,n,k}$: Binary antenna selection indicators from the $j^{th}$ SBS to the $k^{th}$ user over the $n^{th}$ subcarrier when the $m^{th}$ antenna is selected for harvesting energy.
\item $a_{j,n,k}$: Binary subcarrier indicators from the $j^{th}$ SBS to the $k^{th}$ user when the $n^{th}$ subcarrier is selected.
\item $p_{j,n,k}$: The corresponding transmit power from the $j^{th}$ SBS to the $k^{th}$ user in the $n^{th}$ subcarrier.
\end{itemize}

The generalized AS technique is performed to distinguish between the information and power transfer signals.
Through this methodology, the receiver antennas can be separated into two groups.
One group of antennas is used for harvesting energy, whereas the other group handles wireless information reception.~In fact, each antenna of each user can switch its operation mode in different subcarriers for either ID or EH.~Then, the received ID signal from the $j^{th}$ SBS to the $k^{th}$ user over the $n^{th}$ subcarrier is given by
\begin{align}
y^{\text{ID}}_{j,n,k} & =  
\sum_{m \in \mathcal{M}}
(1-x^{m}_{j,n,k})a_{j,n,k}\sqrt{p_{j,n,k}}h^{m}_{j,n,k} 
\nonumber \\ 
& \quad +
\sum_{\substack{j'\neq j \\ j' \in \mathcal{J}}}
\sum_{\substack{k'\neq k \\ k' \in \mathcal{K}}}
\sum_{\substack{m'\neq m \\ m' \in \mathcal{M}}}
(1-x^{m'}_{j,n,k})a_{j',n,k'}\sqrt{p_{j',n,k'}}h^{m'}_{j',n,k} 
+z^{\text{ID}}_{j,n,k}, 
\end{align}
where $z^{ID}_{j,n,k}$ is the additive white Gaussian noise at the $k^{th}$ user when its ID antennas are switched on.
More specifically, $z^{ID}_{j,n,k}$ is an additive white Gaussian noise (AWGN) random variable with zero mean and variance $|\sigma_{j,n,k}^{{ID}}|^{^2}$ denoted by $z^{ID}_{j,n,k} \sim  \mathcal{CN}(0,|\sigma_{j,n,k}^{{ID}}|^{^2})$.
Moreover,~the EH signal from the $j^{th}$ SBS for the user $k$ over the subcarrier $n$ is given by
\vspace{5mm}
\begin{align}
&y^{EH}_{j,n,k}=
\sum_{m \in \mathcal{M}}(x^{m}_{j,n,k})
\sqrt{p_{j,n,k}}g^{m}_{j,n,k} \nonumber\\& \quad \quad \quad +
\sum_{\substack{j'\neq j \\ j' \in \mathcal{J}}}
\sum_{\substack{k'\neq k \\ k' \in \mathcal{K}}}
\sum_{\substack{m'\neq m \\ m' \in \mathcal{M}}}
(x^{m'}_{j,n,k})\sqrt{p_{j',n,k'}}g^{m'}_{j',n,k}+z^{EH}_{j,n,k}, 
\label{F5-2}
\end{align}
where $z^{EH}_{j,n,k}$ is the additive white Gaussian noise with a circularly symmetric Gaussian distribution, referred to as $z^{EH}_{j,n,k} \sim  \mathcal{CN}(0,|\sigma_{j,n,k}^{{ED}}|^{^2})$, at the $k^{th}$ user when its EH antennas are activated.
According to the famous Shannon capacity formula, the data-rate of the user $k$ using its $m$ antenna over the subcarrier $n$ inside the cell $j$ can be written as
\begin{equation}
R^{m}_{j,n,k}=
\log_2
\bigg(
1+
\frac{a_{j,n,k}p_{j,n,k}|h^{m}_{j,n,k}|^{2}}
{{|\sigma^m_{j,n,k}|}^2+I_{j,n,k}}
\bigg),
\end{equation}
where 
\vspace{5mm}
\begin{align}
I_{j,n,k}=
\sum_{\substack{j'\neq j \\ j' \in \mathcal{J}}}
\sum_{\substack{k'\neq k \\ k' \in \mathcal{K}}}
a_{j',n,k'}{p_{j',n,k'}}|h^{m}_{j',n,k}|^{2},    
\end{align}
is the interference term arising from the co-channel effect on the subcarrier $n$ which is emitted by unintended transmitters sharing the same frequency channel.
For facilitating the presentation, we denote $\textbf{p} \in \mathbb{R}^{1\times JNK}$, $\textbf{a} \in \mathbb{Z}^{1\times JNK}$, and $\textbf{x} \in \mathbb{Z}^{1\times JNK}$ as vectors of optimization problem for power allocation, subcarrier assignment, and antenna selection, respectively.~Consequently, the data-rate of the user $k$  when using all its active ID antennas can be stated as 
\begin{align}
R_k(\mathbf{a},\mathbf{x},\mathbf{p}) = 
\sum_{n\in \mathcal{N}}
\sum_{m\in \mathcal{M}}
(1-x^{m}_{j,n,k})R^{m}_{j,n,k}.
\end{align}
\begin{figure}[p]
\centering
\vspace{-5mm}
\includegraphics[width=20cm,trim=4 4 4 4,clip,angle=90]
{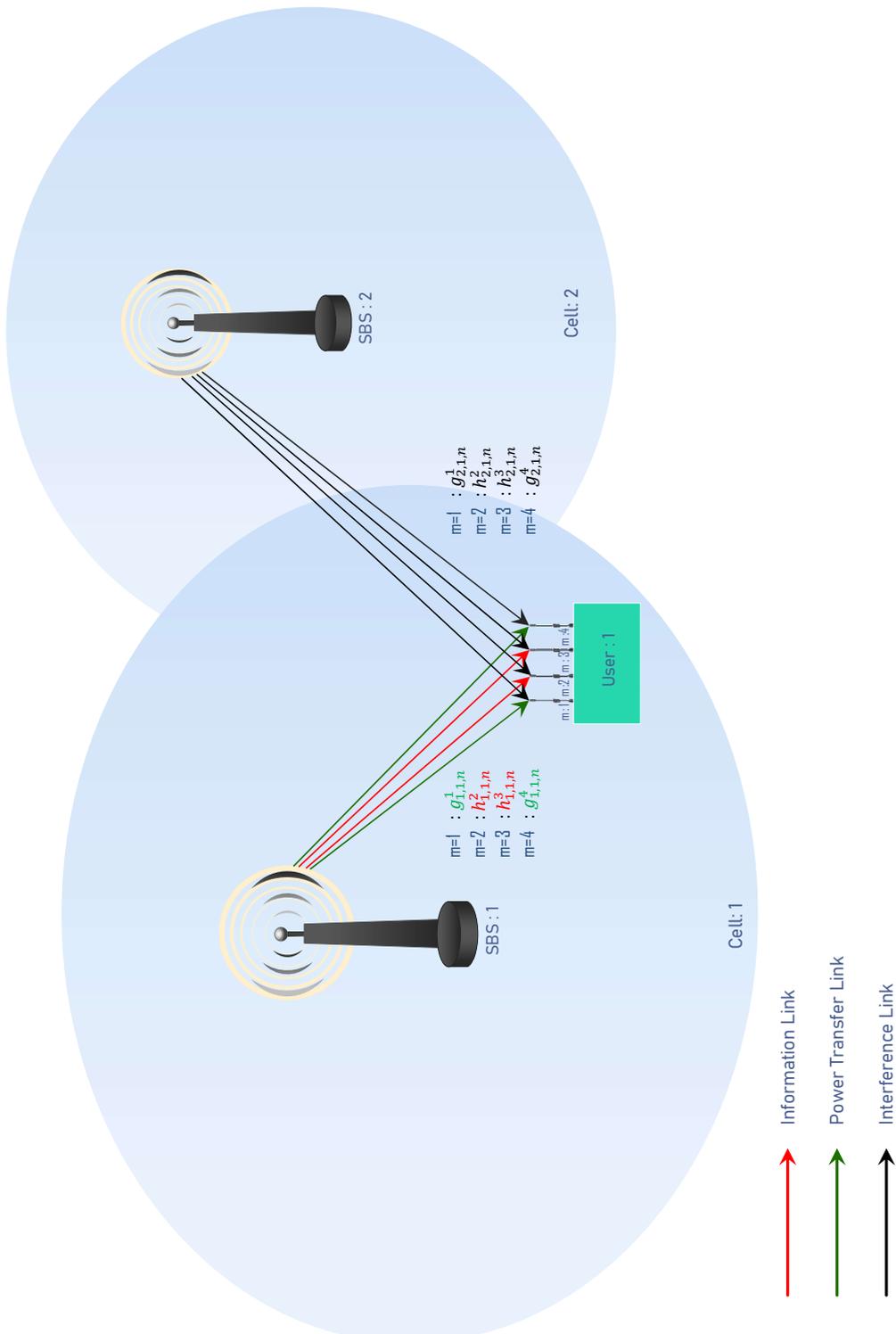}
\caption{SWIPT in a DL of an OFDMA network consisting of $J=2$ small cells, where there is one user in the intersection of the two cells, i.e., $\mathcal{K}_1=$1. The user has for antennas; two of which are employed four ID and the rest for EH. Also, the user is served in the first cell and receives interference from the second cell.}
\label{fig:5.2}%
\end{figure}
Hence, the total system throughput, denoted by $R^{\textrm{Total}}(\mathbf{a},\mathbf{x},\mathbf{p})$, is obtained as
\begin{equation}
R^\textrm{Total}(\mathbf{a},\mathbf{x},\mathbf{p}) = 
\sum_{j \in \mathcal{J}}
\sum_{k\in \mathcal{K}}
R_k(\mathbf{a},\mathbf{x},\mathbf{p}).
\end{equation}
On the other hand, to compute the total power consumption of the network, we use the following linear model in which the transmit power consumption, circuit energy consumption, and the harvested energy are taken into account. 
In this model, there exist coefficients that represent the efficiency of power amplifiers in network devices as well as the power efficiency of EH antennas, which will be explained later on.
In particular, the total power consumption of the considered system $P^{\textrm{Total}}(\textbf{a},\textbf{x},\textbf{p})$ consists of three major terms and can be expressed as
\begin{align}
P^{\textrm{Total}}(\textbf{a},\textbf{x},\textbf{p}) = 
\bigg(
\sum_{j \in \mathcal{J}}
\sum_{k\in \mathcal{K}}
\sum_{n\in \mathcal{N}}
\big(
\frac{a_{j,n,k}p_{j,n,k}}
{\kappa_j}+P_c^\mathrm{SBS}
\big)
\bigg)
-P^{\textrm{EH}}(\textbf{x},\textbf{p}),
\end{align}
where
\begin{equation}\label{Chap5:5-5}
P^{\textrm{EH}}(\textbf{x},\textbf{p})=
\sum_{j \in \mathcal{J}}
\sum_{k\in \mathcal{K}}
\sum_{n\in \mathcal{N}}
\sum_{m\in \mathcal{M}}
\epsilon^m_{j,k}x^{m}_{j,n,k}{p_{j,n,k}}|g^{m}_{j,n,k}|^{2},
\end{equation}
is the total harvested energy in the network topology using the active EH antenna set of each user.
In the above equation, $P^\mathrm{SBS}_\mathrm{c}$ denotes the circuit energy consumption of SBSs, where $\kappa_j$ is the power amplifier of SBSs that takes its values from the interval of $0< \kappa_j < 1$.
Moreover, $0<\epsilon^m_{j,k}<1$ is the power conversion efficiency for the $m^{th}$ active EH antenna of the $k^{th}$ receiver in cell $j$ as introduced in the previous chapters.
It should also be noted that the contribution of the noise power, i.e., ${|\sigma^m_{j,n,k}|}^2$, to $P^{\text{EH}}$ formula in (\ref{Chap5:5-5}) is ignored, since its value is very small compared to the other existing term in (\ref{Chap5:5-5}).
In this chapter, we define the EE as the ratio of system throughput to the corresponding network energy consumption in bits/joule, and denote it by $\mathscr{E}_{eff}(\mathbf{a},\mathbf{p},\mathbf{q})$, where
\begin{align} \label{F5-6}
\mathscr{E}_{eff}(\mathbf{a},\mathbf{x},\mathbf{p}) =
\frac{R^\textrm{Total}(\mathbf{a},\mathbf{x},\mathbf{p})}
{P^\textrm{Total}(\mathbf{a},\mathbf{x},\mathbf{p})}.
\end{align}
In what follows, we first formulate the problem to maximize the EE while considering the feasibility of the transmitted power, minimum data-rate requirement as well as the OFDMA constraint in multi-user multi-cell network with SWIPT.
Then, we propose a solution to solve the optimization problem.

\section{Optimization Problem Formulation}
Notwithstanding the fact that EE is an effective resource allocation metric in cellular networks, some limitations might turn it into an undesired metric in many of the modern applications.
In energy-efficient algorithms, the main goal is to maximize system throughput and minimize the corresponding energy consumption simultaneously, without differentiating between the priority of these competing objectives. 

In this section, we first formulate the optimization problem of joint SBS-subcarrier assignment and power allocation together with the antenna selection for the EE maximization problem of a multi-user multi-cell OFDMA network with a generalized AS-based SWIPT framework. 
Afterward, we present our proposed algorithm to solve the stated problem.
Hence, we introduce the following optimization problem to maximize the system EE 
\begin{subequations}
\begin{align} 
&\max_{\textbf{a},\textbf{x},\textbf{p}}
\mathscr{E}_{eff}(\mathbf{a},\mathbf{x},\mathbf{p})
\label{F5-7}\\
s.t.: 
&~C_{1}:
\sum_{k\in \mathcal{K}} 
a_{j,n,k}\leq 1,~~~~~~~~~~~~~~~~~~~~
\forall {j \in \mathcal{J},~ 
\forall n\in \mathcal{N}},
\label{F5-8}\\ 
&~C_{2}:
\sum_{k\in \mathcal{K}}~
\sum_{n\in \mathcal{N}}
a_{j,n,k}\ p_{j,n,k}\leq p_{max},~
\forall j \in \mathcal{J}, 
\label{F5-9}\\
&~C_{3}:
\sum_{j\in \mathcal{J}}
R_k(\mathbf{a},\mathbf{x},\mathbf{p})
\geq  R_{min},~~~~~~~~
\forall k \in \mathcal{K},\label{F5-10}\\
&~C_{4}:
\sum_{m\in \mathcal{M}} 
x^{m}_{j,n,k}=1,~~~~~~~~~~~~~~~~~~
\forall j\in \mathcal{J},~
\forall n \in \mathcal{N},~
\forall k \in \mathcal{K},
\label{F5-11}\\
&~C_{5}:
a_{j,n,k}\in\{0,1\} ,~~~~~~~~~~~~~~~~~~~~
\forall {j \in \mathcal{J},~ 
\forall n\in \mathcal{N},~ 
\forall  k\in \mathcal{K}},
\label{F5-12}\\
&~C_{6}:
x^{m}_{j,n,k}\in\{0,1\},~~~~~~~~~~~~~~~~~~~~
\forall {j \in \mathcal{J},~ 
\forall n\in \mathcal{N},~ 
\forall  k\in \mathcal{K},  ~ 
\forall m\in \mathcal{M}},
\label{F5-13}\\
&~C_{7}:
a_{j,n,k} + a_{j',n',k}\leq 1,~~~~~~~~~~~~~~
\forall {j\neq j' \in \mathcal{J},~ 
\forall n,n' \in \mathcal{N}},~ 
\forall  k\in \mathcal{K}.
\end{align}
\label{chap5:5:main}%
\end{subequations}
In the optimization problem (\ref{chap5:5:main}), constraint $C_1$ indicates that each subcarrier can be allocated to at most one user.
The constraint $C_2$ makes sure that the total transmit power of SBSs should not exceed their maximum threshold denoted by $p_{max}$.~In constraint $C_3$, the minimum data-rate requirement, $R_{min}$, is guaranteed for each user in DL in each cell.
Constraint $C_4$ indicates that each user utilizes only one antenna in each subcarrier.
$C_5$ and $C_6$ indicate that the subcarrier assignment and the antenna selection variables take only binary values.
Finally,~$C_7$ presents that each user can be assigned to at most one SBS in each subcarrier.
It is worth mentioning that this last constraint is the so-called user association in the literature \cite{6845056}. 

Due to the binary subcarrier assignment and antenna selection variables as well as the interference included in the data-rate function, and the fractional form of the objective function, the optimization problem~(\ref{chap5:5:main}) is a mixed-integer non-linear programming (MINLP) problem, which is generally complicated. 
These challenges that make the above optimization problem difficult to solve are further explained below:
\begin{itemize}
\item Fractional form of the objective function: As a fractional objective function, $\mathscr{E}_{eff}(\mathbf{a},\mathbf{x},\mathbf{p})$ is non-convex.
\item Multiplication of two variables: Since the multiplication of two variables is non-convex, the term $a_{j,n,k}p_{j,n,k}$ in constraint $C_2$, poses a challenge in tackling the optimization problem in (\ref{chap5:5:main}). 
Because the maximum total transmit power constraint in $C_2$, the minimum data-rate QoS requirement constraint in $C_3$, and also the total data-rate objective function are multiplied by a function of joint SBS-subcarrier assignment and the transmit power variables (as given in (\ref{F5-9}), (\ref{F5-10}), and (\ref{F5-7})), these mentioned constraints together with the objective function are non-convex.
\item Interference: The inter-cell interference incorporated in data-rate functions,~makes both constraints $C_3$ and the objective function non-convex. 
\item Binary antenna selection and SBS-subcarrier assignment variable: Discrete antenna selection and SBS-subcarrier assignment variables turn (\ref{chap5:5:main}) into a complex MINLP problem.
\end{itemize}

\section{Solution to the Optimization Problem}
To cope with the complexity of the MINLP problem, we decompose the optimization problem (\ref{chap5:5:main}) into two subproblems: 
1) joint SBS-subcarrier assignment and power allocation, and 
2) antenna selection.
In order to find the locally optimal SBS-subcarrier allocation, power assignment, and antenna selection, the following iterative procedure is employed.
At the beginning of each iteration, the optimal SBS-subcarrier allocation and power assignment are obtained from an optimal antenna selection of the previous iteration, i.e., $\mathbf{x}^{t-1}$, using the results of layer \ref{Chap5:layer1}.\textbf{A}.
Knowing the best SBS-subcarrier allocation and power assignment, the best antenna for either EH or ID operation is selected by incorporating the results of layer \ref{Chap5:layer2}.\textbf{B}.
The corresponding update rule is summarized as follows
\begin{align}\label{update}
& \underbrace{(\mathbf{a}^{0},\mathbf{p}^{0}) 
\rightarrow 
\mathbf{x}^{0}}_{\text{Initialization}}
\rightarrow ...\rightarrow  
\underbrace{(\mathbf{a}^{t-1},\mathbf{p}^{t-1}) \rightarrow\mathbf{x}^{t-1}}_{\text{Iteration \textit{t-}1}}
\nonumber\\ 
\rightarrow 
& \underbrace{(\mathbf{a}^{t},\mathbf{p}^{t}) \rightarrow\mathbf{x}^{t}}_{\text{Iteration \textit{t}}}
\rightarrow ...\rightarrow
\underbrace{(\mathbf{a}^{opt},\mathbf{p}^{opt})
\rightarrow\mathbf{x}^{opt}}_{\text{Optimal Solution}}.
\end{align}
After solving each subproblem by algorithms that are discussed in the following two layers methodology, an iterative procedure is employed in which the solution of the previous subproblem is used as the input of the current problem.
Through this iterative procedure, we are able to enhance the accuracy of our obtained solution.
The iterative algorithm in (\ref{update}) needs an initial setting for $\mathbf{a,x}$, and $\mathbf{p}$.
In order to converge to a locally optimal solution, these initial settings must be appropriately selected.
To this end, we first note that the optimization problem (\ref{chap5:5:main}) is feasible, as we can find initial settings $\mathbf{a}^0, \mathbf{p}^0$, and $\mathbf{x}^0$, satisfying the constraints.
Such a setting that respects all the constraints can be found as follows. First, for the initial SBS-subcarrier assignment $\mathbf{a}^0$ and power allocation $\mathbf{p}^0$, we assume each SBS-subcarrier is assigned to the small-cell user with the highest channel gain, where equal power is allocated to all small-cell users across all subcarriers. 
Lastly, for the antenna selection $\mathbf{x}^0$, a set of antennas are selected for each of the small-cell users based on quasi-stationary channel statistics. 
The selected antennas in each user can perform either EH or ID operation accordingly.

Having equipped with the necessary background, a tractable solution procedure to the original problem (\ref{chap5:5:main}) is described in a two-layer format in the next following subsections.

\subsection{A. Joint SBS-subcarrier Assignment and Power Allocation}\label{Chap5:layer1}
In this subsection, we carry out the subproblem of the SBS-subcarrier assignment and power allocation in order to maximize EE. 
Assuming the antenna set selection is given from the previous iteration, we aim to tackle the non-convexity of the multiplication of a binary variable with the transmit power in constraints $C_2$ and $C_3$ along with the objective function.
In order to handle this issue,~we adopt the big-M formulation \cite{big_M} to decouple the product terms.
Therefore, we introduce the following additional constraints into the optimization problem in (\ref{chap5:5:main}) 
\begin{align}
&C_{8}:
\tilde{p}_{j,n,k}\leq p_{max}a_{j,n,k},~~~~~~~~~~~~~~~~~~~~
\forall j \in \mathcal{J},~
\forall n \in \mathcal{N},~
\forall k \in \mathcal{K},\\
&C_{9}:
\tilde{p}_{j,n,k}\leq p_{j,n,k},~~~~~~~~~~~~~~~~~~~~~~~~~~
\forall j \in \mathcal{J},~
\forall n \in \mathcal{N},~
\forall k \in \mathcal{K},\\
&C_{10}:
\tilde{p}_{j,n,k}\geq  
p_{j,n,k}-(1-a_{j,n,k})p_{max},~
\forall j \in \mathcal{J},~
\forall n \in \mathcal{N},~
\forall k \in \mathcal{K},\\
&C_{11}:
\tilde{p}_{j,n,k}\geq 0,~~~~~~~~~~~~~~~~~~~~~~~~~~~~~~
\forall j \in \mathcal{J},~
\forall n \in \mathcal{N},~
\forall k \in \mathcal{K},
\end{align}  
where $\tilde{\textbf{p}} \in \mathbb{R}^{1\times JNK}$ is the collection of all $\tilde{p}_{j,n,k}$'s.
Through this method, we can easily deal with the non-convexity of the multiplication of two variables, both in the objective and also the constraints. 
Next, we relax the integer SBS-subcarrier assignment variable by converting it into a continuous variable between zero and one as
\begin{equation}
\dot{C}_{5}:~ 0\leq  a_{j,n,k} \leq 1,~\forall j \in \mathcal{J},~\forall n \in \mathcal{N},~\forall k \in \mathcal{K}.
\end{equation}
Moreover,~by inspiring the same approach in previous chapters, we impose the following region to the optimization problem
\begin{align}
\ddot{C}_{5}:~
\sum_{j \in \mathcal{J}}
\sum_{ n \in \mathcal{N}}
\sum_{ k \in \mathcal{K}} 
\left( a_{j,n,k} - (a_{j,n,k})^2 \right) \leq 0.
\end{align}
Subsequently , the original optimization problem in (\ref{chap5:5:main}) can be reformulated as follows
\begin{subequations}
\begin{align}
(\mathbf{a}^{t},\tilde{\mathbf{p}}^{t}) =
& \arg \max_{\mathbf{a},\mathbf{p},\tilde{\mathbf{p}}}
\frac{{\overline{{R}}}^\textnormal{Total}
(\mathbf{a}^{t},\mathbf{x}^{t-1},\tilde{\mathbf{p}}^{t})}
{{\overline{{P}}}^\textnormal{Total}
(\mathbf{a}^{t},\mathbf{x}^{t-1},\tilde{\mathbf{p}}^{t},\textbf{p}^{t})} 
\label{F5-13_new}\\
s.t.: 
&~C_{1},C_{4},\dot{C}_{5},\ddot{C}_{5},{C}_{7}-C_{11},\\
&~C_{2}:
\sum_{k\in \mathcal{K}}~
\sum_{n\in \mathcal{N}}
\tilde{p}_{j,n,k}\leq p_{max},~~
\forall j \in \mathcal{J}, \\ 
&~C_{3}: \sum_{j\in \mathcal{J}} \overline{R}_k(\mathbf{a},\mathbf{x},\tilde{\mathbf{p}}) 
\geq R_{min},~
\forall k \in \mathcal{K}.
\end{align}
\label{F5:18}%
\end{subequations}
Furthermore, the numerator and denominator of the objective function in the above optimization problem (\ref{F5:18}) are as follows
\begin{equation}
{\overline{R}}^\textrm{Total}(\mathbf{a},\mathbf{x},\tilde{\mathbf{p}}) = 
\sum_{j\in \mathcal{J}}
\sum_{k\in \mathcal{K}}
\underbrace{
\sum_{n\in \mathcal{N}}
\sum_{m\in \mathcal{M}}(1-x^{m}_{j,n,k})
\log_2
\bigg(
1+\frac{\tilde{p}_{j,n,k}|h^{m}_{j,n,k}|^{2}}{{|\sigma^m_{j,n,k}|}^2+I_{j,n,k}}
\bigg)}_{\text{\normalsize $\overline{R}_k(\mathbf{a},\mathbf{x},\tilde{\mathbf{p}})$}},
\end{equation}
\begin{align}
{\overline{P}}^{\textrm{Total}}(\textbf{a},\textbf{x},\tilde{\mathbf{p}},\textbf{p}) 
= & \bigg(
\sum_{j \in \mathcal{J}}
\sum_{k\in \mathcal{K}}
\sum_{n\in \mathcal{N}}
(\frac{\tilde{p}_{j,n,k}}{\kappa_j}+
P_c^\mathrm{SBS})\bigg)-
P^{\textrm{EH}}(\textbf{x},\textbf{p}),
\end{align}
where
\begin{equation}
I_{j,n,k}=
\sum_{\substack{j'\neq j \\ j' \in \mathcal{J}}}
\sum_{\substack{k'\neq k \\ k' \in \mathcal{K}}}
\tilde{p}_{j',n,k'}|h^{m}_{j',n,k}|^{2}.
\end{equation}
Although we addressed the issue of the coupling variables in constraints $C_2$ (\ref{F5-9}), $C_3$ (\ref{F5-10}), and the objective function, the main problem in (\ref{F5:18}) is still non-convex due to existence of the interference in logarithmic data-rate functions. 
In order to handle this problem, we employ the majorization minimization (MM) algorithm \cite{MM} to make the constraint $C_{3}$ in (\ref{F5-13_new}) and the data-rate function in the objective function of (\ref{F5-13_new}) convex, as becomes clear in what follows.
To do so, we first restate the constraint $C_3$ as a difference of convex (D.C.) functions \cite{SCA,DC} as
\begin{equation}\label{F5-17}
\dot{C}_{3}:~
\sum_{j\in \mathcal{J}}
\sum_{n\in \mathcal{N}}
\sum_{m\in \mathcal{M}}
\mathcal{U}(\mathbf{a},\mathbf{x},\tilde{\mathbf{p}}) - 
\sum_{j\in \mathcal{J}}
\sum_{n\in \mathcal{N}}
\sum_{m\in \mathcal{M}}
\mathcal{V}(\mathbf{a},\mathbf{x},\tilde{\mathbf{p}}) 
\geq 
R_{min},~
\forall k \in \mathcal{K},
\end{equation}
where $\mathcal{U}(\mathbf{a},\mathbf{x},\tilde{\mathbf{p}})$ and $\mathcal{V}(\mathbf{a},\mathbf{x},\tilde{\mathbf{p}})$ are
\begin{align}
\mathcal{U}(\mathbf{a},\mathbf{x},\tilde{\mathbf{p}}) 
&= 
(1-x^{m}_{j,n,k})\log_2\bigg(\tilde{p}_{j,n,k}h^{m}_{j,n,k}
+ {|\sigma^m_{j,n,k}|}^2+I_{j,n,k}\bigg),
\\
\mathcal{V}(\mathbf{a},\mathbf{x},\tilde{\mathbf{p}})
&= (1-x^{m}_{j,n,k})\log_2\bigg({|\sigma^m_{j,n,k}|}^2+I_{j,n,k}\bigg).
\end{align}
To obtain a concave approximation for the constraint $\dot{C}_3$, we apply the MM algorithm \cite{MM}, and we construct a surrogate function for $\mathcal{V}(\mathbf{a},\mathbf{x},{\tilde{\mathbf{p}}})$ using first-order Taylor approximation as
\begin{align}
\label{F5-18}
{{\mathcal{V}}}(\mathbf{a},\mathbf{x},\tilde{\mathbf{p}}) \simeq &
\sum_{j\in \mathcal{J}}
\sum_{n\in \mathcal{N}}
\sum_{m\in \mathcal{M}}
\mathcal{V}(\mathbf{a}^{(t-1)},\mathbf{x}^{(t-1)},\tilde{\mathbf{p}}^{(s-1)}) \nonumber \\+ &
\sum_{j\in \mathcal{J}}
\sum_{n\in \mathcal{N}}
\sum_{m\in \mathcal{M}}
\nabla_{\tilde{\mathbf{p}}}{\mathcal{V}^T}(\mathbf{a}^{(t-1)},\textbf{x}^{(t-1)},\tilde{\mathbf{p}}^{(s-1)}) (\tilde{\mathbf{p}}-\tilde{\mathbf{p}}^{(s-1)}) \triangleq \tilde{{\mathcal{V}}}(\mathbf{a},\mathbf{x},\tilde{\mathbf{p}}),
\end{align}
where ${\tilde{\mathbf{p}}}^{(s-1)}$ is the solution of the problem at the $(s-1)^{th}$ iteration, and $\nabla_{\square}$ represents the gradient operation with respect to ${\square}$.
In a similar manner, we handle the total data-rate function in the nominator of the objective function in the optimization problem (\ref{F5-13_new}) as
\begin{align}\label{F5-16_new}
\overline{R}^\textrm{Total}(\mathbf{a},\mathbf{x},\tilde{\mathbf{p}}) = 
\mathscr{U}(\mathbf{a},\mathbf{x},{\tilde{\mathbf{p}}})- \mathscr{V}(\mathbf{a},\mathbf{x},{\tilde{\mathbf{p}}}),
\end{align}
where
\begin{equation}
\mathscr{U}(\mathbf{a},\mathbf{x}, {\tilde{\mathbf{p}}})=
\sum_{j \in \mathcal{J}} 
\sum_{n \in \mathcal{N}}
\sum_{k \in \mathcal{K}}
\sum_{m\in \mathcal{M}}
\mathcal{U}(\mathbf{x},\tilde{{\mathbf{p}}}),   
\end{equation}
and 
\begin{align}
\mathscr{V}(\mathbf{a},\mathbf{x}, {\tilde{\mathbf{p}}})=
\sum_{j \in \mathcal{J}}
\sum_{n \in \mathcal{N}}
\sum_{k \in \mathcal{K}}
\sum_{m\in \mathcal{M}}
\mathcal{V}(\mathbf{x},\tilde{{\mathbf{p}}}).   
\end{align}
With the same reasoning as before, it can be concluded that numerator of the objective function in (\ref{F5-16_new}) is not a concave function.
However, equation (\ref{F5-16_new}) belongs to the class of D.C. functions.
To approximate a concave function in (\ref{F5-16_new}), the MM approach is applied again to make a concave approximation via the first-order Taylor approximation as 
\begin{align}\label{F5-30}
{\mathscr{V}}(\mathbf{a},\mathbf{x},{\tilde{\mathbf{p}}}) & \simeq
\mathscr{V}\big(\mathbf{a}^{(s-1)},\mathbf{x}^{(t-1)} , {\tilde{\mathbf{p}}}^{(s-1)}\big) \nonumber \\& +
\nabla_{\tilde{\mathbf{p}}}
{\mathscr{V}^T\!\big(\mathbf{a}^{(s-1)},\mathbf{x}^{(t-1)} , {\tilde{\mathbf{p}}}^{(s-1)}\big)} 
\big({\tilde{\mathbf{p}}} - 
{\tilde{\mathbf{p}}}^{(s-1)}\big) \triangleq \tilde{\mathscr{V}}(\mathbf{a},\mathbf{x},{\tilde{\mathbf{p}}}), 
\end{align}
where the total data-rate of the network can be rewritten as
\begin{equation}
\widehat{\overline{R}}^\textnormal{Total}
(\mathbf{a},\mathbf{x},\tilde{\mathbf{p}}) =
\mathscr{U}(\mathbf{a},\mathbf{x}, \tilde{\mathbf{p}})- 
\tilde{\mathscr{V}}(\mathbf{a},\mathbf{x}, \tilde{\mathbf{p}}).
\end{equation}
By using approximation (\ref{F5-30}), the MM principles is satisfied. 
This makes a tight lower bound of equation (\ref{F5-16_new}) \cite{che2014joint,DC}.
Now, the optimization problem in (\ref{chap5:5:main}) can be restated as 
\begin{align}\label{F5-20}
(\mathbf{a}^{t},\mathbf{p}^{t}) = 
& \arg \max_{\mathbf{a},\mathbf{p},\tilde{\mathbf{p}}}
\frac{\widehat{\overline{R}}^\textnormal{Total}(\mathbf{a}^{t},\mathbf{x}^{t-1},\tilde{\mathbf{p}}^{t}) }
{\overline{P}^\textnormal{Total}(\mathbf{a}^{t},\mathbf{x}^{t-1},\tilde{\mathbf{p}}^{t},\mathbf{p}^t)} \\
s.t.: &~C_{1}-C_{2},\dot{C}_{3},C_{4},\dot{C}_{5},\ddot{C}_{5},{C}_{7}-C_{11}. \nonumber
\end{align}
However, the optimization problem in (\ref{F5-20}) is still non-convex due to non-convexity of $\ddot{C}_{5}$. 
Nevertheless, let us rewrite $\ddot{C}_{5}$ as the difference of two convex functions as $\mu(\mathbf{a})-\nu(\mathbf{a})$, where 
\begin{align}
&\mu(\mathbf{a})=\sum_{j \in \mathcal{J}}
\sum_{ n \in \mathcal{N}}\sum_{ k \in \mathcal{K}}( a_{j,n,k}),\\
&\nu(\mathbf{a})=\sum_{j \in \mathcal{J}}
\sum_{ n \in \mathcal{N}}\sum_{ k \in \mathcal{K}}(a_{j,n,k})^2. 
\end{align}
Similar to the approach used for data-rate functions in (\ref{F5-17}) and (\ref{F5-16_new}), we employ the same methodology based on the MM algorithm to convexify the constraint by taking the first-order Taylor approximation from  $\nu(\mathbf{a})$.
Consequently,~$\nu(\mathbf{a})$ can be approximated as
\begin{align}
\nu(\mathbf{a}) 
\simeq
{\nu}(\tilde{\mathbf{a}}^{(s-1)}) +
\nabla_{\tilde{\mathbf{a}}}{\nu^T}(\tilde{\mathbf{a}}^{(s-1)}) 
(\tilde{\mathbf{a}}-\tilde{\mathbf{a}}^{(s-1)}) 
\triangleq 
\tilde{\nu}(\mathbf{a}).
\label{F5-34}
\end{align}
Hence, the constraint $\ddot{C}_{5}$ can be restated as $\mu(\mathbf{a})-\tilde{\nu}(\mathbf{a})$ that would be a  convex function.~Finally, the optimization problem at hand can be recast as 
\begin{subequations}
\begin{align}
(\mathbf{a}^{t},\mathbf{p}^{t}) = 
& \arg \max_{\mathbf{a},\mathbf{p},\tilde{\mathbf{p}}}
\frac{\widehat{\overline{R}}^\textnormal{Total}(\mathbf{a}^{t},\mathbf{x}^{t-1},\tilde{\mathbf{p}}^{t}) }
{\overline{P}^\textnormal{Total}(\mathbf{a}^{t},\mathbf{x}^{t-1},\tilde{\mathbf{p}}^{t},\mathbf{p}^t)} \\s.t.: 
&~\ddot{C}_{5}:~\mu(\mathbf{a})-\tilde{\nu}(\mathbf{a}) \leq 0,\\
&~C_{1}-C_{2},\dot{C}_{3},C_{4},\dot{C}_{5},{C}_{7}-C_{11}. 
\end{align}
\label{F5-26_new}%
\end{subequations}
So far, we tackled the issues with the multiplication of two variables, the interference in data-rate functions, and the binary SBS-subcarrier assignment variable.
As for the last step in our solution design, we must take care of the issue with the fractional form of the objective function that makes the optimization problem non-convex. 
Thus, we now address the fractional objective function in the optimization problem in (\ref{F5-26_new}) by describing a technique to treat fractional programming problems.
It can be concluded that the optimization problem in (\ref{F5-26_new}) can be solved via a well-known algorithm, namely, the Dinkelback method \cite{Dinkelbach}.
For this matter, let us denote $\mathscr{E}_{eff}^{^{*}}$ as the optimal EE point of the optimization problem (\ref{F5-26_new}) that belongs to the set of feasible solutions spanned by its constraints.
Therefore, we can solve the following non-fractional optimization problem given the data from the previous round of the main loop for the chosen set of antennas, i.e., $t-1$, and the $i^{th}$ iteration for the Dinkelback algorithm.
Thus, by considering the \textbf{Algorithm~\ref{Dinkelbach-algorithm}}, an optimization problem with a transformed objective function is introduced to obtain the resource allocation policy as follows
\begin{align}\label{F5-36_new}
(\mathbf{a}^{i},\tilde{\mathbf{p}}^{i}) =
& \arg \max_{\mathbf{a},\mathbf{p},\tilde{\mathbf{p}}}
\widehat{\overline{R}}^\textnormal{Total}
(\mathbf{a}^i,\mathbf{x}^{{t-1}},\tilde{\mathbf{p}}^i)-
\mathscr{E}_{eff}^{^i} \overline{P}^\textnormal{Total}
(\mathbf{a}^{*^i},\mathbf{x}^{t-1},\tilde{\mathbf{p}}^{*^i},\mathbf{p}^{*^i}) \\
s.t.: &~C_{1}-C_{2},\dot{C}_{3},C_{4},\dot{C}_{5},\ddot{C}_{5},{C}_{7}-C_{11}. \nonumber 
\end{align}
In the above optimization problem, we have {$\mathscr{E}_{eff}^{^i} = \frac{\widehat{\overline{R}}^\textnormal{Total}(\mathbf{a}^i,\mathbf{x}^{{t-1}},\tilde{\mathbf{p}}^i)}{\overline{P}^\textnormal{Total}(\mathbf{a}^i,\mathbf{x}^{{t-1}},\tilde{\mathbf{p}}^i,{\mathbf{p}}^i)}$} in which {$\{\mathbf{a}^i,\mathbf{x}^{{t-1}},\tilde{\mathbf{p}}^i,{\mathbf{p}}^i\}$} are the corresponding resource allocation parameters. 
It is easy to demonstrate that the optimization problem in (\ref{F5-36_new}) is now convex with respect to all variables.
Therefore, the \text{\textbf{Algorithm \ref{Dinkelbach-algorithm}}} terminates when $\mathscr{E}_{eff}^{^i}$ converges and so does the solution to problem (\ref{F5-36_new}). 
That is \{$\mathbf{a}^{*},\tilde{\mathbf{p}}^{*}$,$\mathbf{p}^*$\} are eventually achieved.

\begin{proposition}\label{fractional_programming}
The optimal EE, i.e., $\mathscr{E}_{eff}^{^*}$, can be used to obtain the resource allocation policy if and only if
\begin{align}
\max_{\mathbf{a},\tilde{\mathbf{p}},\mathbf{p} \in \mathcal{S_F}}\quad 
\widehat{\overline{R}}^\textnormal{Total}(\mathbf{a}^{i}, \mathbf{x}^{t-1}, \tilde{\mathbf{p}}^{i})-& 
\mathscr{E}_{eff}^{^i}\overline{P}^\textnormal{Total}(\mathbf{a}^{i}, \mathbf{x}^{t-1}, \tilde{\mathbf{p}}^i,{\mathbf{p}}^i) = \nonumber \\
\widehat{\overline{R}}^\textnormal{Total}(\mathbf{a}^{*},\mathbf{x}^{t-1},\tilde{\mathbf{p}}^{*})-& 
\mathscr{E}_{eff}^{^*}\overline{P}^\textnormal{Total}(\mathbf{a}^{*},\mathbf{x}^{t-1},\tilde{\mathbf{p}}^*,{\mathbf{p}}^*) =0 ,
\end{align}
for $\widehat{\overline{R}}^\textnormal{Total}(\mathbf{a}^{*},\mathbf{x}^{t-1},\tilde{\mathbf{p}}^{*}) \geq 0$ and $\overline{P}^\textnormal{Total}(\mathbf{a}^{*},\mathbf{x}^{t-1},\tilde{\mathbf{p}}^*,{\mathbf{p}}^*) \geq 0$, where $\mathbf{a^*}$,~$\mathbf{x}^{t-1}$, $\tilde{\mathbf{p}}^*$,~and $\mathbf{{p}^*}$ yield the optimal solution to the convex optimization problem in (\ref{F5-36_new}).
\end{proposition}
\begin{proof}
We define $\mathscr{E}_{eff}^{^*}$ and $\{\mathbf{a}^{*},\mathbf{x}^{t-1},\mathbf{p}^{*},\tilde{\mathbf{p}}^{*}\}\in \mathcal{S_F}$ as the optimal EE and the optimal resource allocation policy of the original objective function in (\ref{chap5:5:main}), respectively, which can be obtained as
\begin{align}
\label{F5-14}
\mathscr{E}_{eff}^{^*} & = 
\max_{\mathbf{a},\mathbf{{p},\tilde{\mathbf{p}}}\in \mathcal{S_F}}
{ \frac{\widehat{\overline{R}}^\textnormal{Total}(\mathbf{a},\mathbf{{x}}^{t-1},{\mathbf{p}})}
{\overline{P}^\textnormal{Total}(\mathbf{a},\mathbf{{x}}^{t-1},\tilde{\mathbf{p}},\mathbf{p})} }.
\end{align}
Then, the optimal EE would be
\begin{align}
\mathscr{E}_{eff}^{^*} = &
\frac{\widehat{\overline{R}}^\textnormal{Total}
(\mathbf{a}^{*},\mathbf{x}^{t-1},\tilde{\mathbf{p}}^{*})}
{P^\textrm{Total}(\mathbf{a}^{*},\mathbf{x}^{t-1},\tilde{\mathbf{p}}^{*},\mathbf{p}^{*})} 
\nonumber
\\
\geq &
\frac{\widehat{\overline{R}}^\textnormal{Total}(\mathbf{a},\mathbf{x}^{t-1},\tilde{\mathbf{p}})}
{P^\textrm{Total}(\mathbf{a},\mathbf{x}^{t-1},\tilde{\mathbf{p}},\mathbf{p})}, 
\forall \{\mathbf{a}^{*},\mathbf{x}^{t-1},\tilde{\mathbf{p}}^{*},\mathbf{p}^{*}\}\in \mathcal{S_F}.
\end{align}
Therefore, it can easily be seen that 
\begin{align}
& \widehat{\overline{R}}^\textnormal{Total}(\mathbf{a},\mathbf{x}^{t-1},\tilde{\mathbf{p}})-
\mathscr{E}_{eff}^{^*}\overline{P}^\textnormal{Total}(\mathbf{a},\mathbf{x}^{t-1},\tilde{\mathbf{p}},\mathbf{p})
\leq 0,
\end{align}
Hence, one can conclude that 
\begin{align}
\widehat{\overline{R}}^\textnormal{Total}(\mathbf{a}^{*},\mathbf{x}^{t-1},\tilde{\mathbf{p}}^{*})-
\mathscr{E}_{eff}^{^*}\overline{P}^\textnormal{Total}(\mathbf{a}^{*},\mathbf{x}^{t-1},\tilde{\mathbf{p}}^{*},\mathbf{p}^{*})=0.
\end{align}
Thus, we have
\begin{align}
\max_{\mathbf{a},\tilde{\mathbf{p}},\mathbf{{p}} \in \mathcal{S_F}} \quad 
{\widehat{\overline{R}}^\textnormal{Total}(\mathbf{a},\mathbf{x}^{t-1},\tilde{\mathbf{p}}) -\mathscr{E}_{eff}^{^*}\overline{P}^\textnormal{Total}(\mathbf{a},\mathbf{x}^{t-1},\tilde{\mathbf{p}},\mathbf{p})}=0,    
\end{align}
and this would be achievable by the resource allocation policy. 
This completes the proof.
\end{proof}

Since the problem in (\ref{F5-36_new}) is a convex optimization problem at each iteration, it can be solved efficiently using the optimization packages, including interior point methods such as CVX \cite{boyd2004convex}.
Moreover, it is worth mentioning that the proposed \textbf{Algorithm \ref{Dinkelbach-algorithm}} can obtain an optimal solution to the problem (\ref{chap5:5:main}), if the inner loop of the optimization problem in (\ref{F5-36_new}) can be solved optimally in each iteration.
For solving the inner problem of (\ref{F5-36_new}), we have to check the approximations for the data-rate functions and check a condition for their convergence.
Let us consider the data-rate approximation at the $s^{th}$ iteration of the MM method as $\widehat{\overline{R}}^{(s)}$. 
In order to make sure that the Taylor approximation is a tight lower bound, we investigate the difference of data-rate functions for two consecutive iterations of the MM method.
The proof of convergence of the MM procedure in the inner loop of \textbf{Algorithm~\ref{Dinkelbach-algorithm}} is similar to the proof given in {\text{\textbf{Proposition~\ref{proposition3}}}} of chapter~\ref{CHAP4}.
We also note that the optimal solution of the Dinkelback algorithm can be considered as the solution for the $(t-1)^{th}$ iteration of the main loop for obtaining the antenna selection set for either an ID or EH operation.
In the next subsection, we discuss the solution methodology for solving the antenna selection problem. 

\begin{algorithm}[ht]
\caption{Resource Allocation Algorithm for Solving Joint Subcarrier and Power Assignment}
\label{Dinkelbach-algorithm}
\centering
\begin{algorithmic}[1]
\STATE {$\mathbf{Initialize}$} \\
{\begin{addmargin}[1em]{0em}
{iteration index of resource allocation policy $i=0$ with maximum allowed tolerance $\Theta>0$, \\
MM iteration index $s=0$ with maximum number of MM iteration $T_{max}$ and $\mathbf{\Psi}>0$, \\
feasible set vector $\mathbf{a}^{0}$, $\mathbf{x}^{t-1}$, $\tilde{\mathbf{p}}^{{0}}$, and $\mathbf{p}^{0}$},\\
and the penalty factor $\lambda \gg 1$.
\end{addmargin}}
\STATE Set maximum EE for resource allocation policy $\mathscr{E}_{eff}^{^0}=0$.
\STATE \textbf{while} {$|\mathscr{E}_{eff}^{^i}- \mathscr{E}_{eff}^{^{i-1}}| > \Theta$} \textbf{do}
\STATE{
\begin{addmargin}[1em]{0em}
{\textbf{repeat}}
\end{addmargin}}
\STATE{
\begin{addmargin}[2em]{0em}
{Update $\tilde{\mathcal{V}}(\mathbf{a},\mathbf{x},{\tilde{\mathbf{p}}})$, $\tilde{\mathscr{V}}(\mathbf{a},\mathbf{x},{\tilde{\mathbf{p}}})$, and  $\tilde{\nu}(\mathbf{a})$
using equations (\ref{F5-18}), (\ref{F5-30}), and (\ref{F5-34}), respectively.}  
\end{addmargin}}
\STATE{
\begin{addmargin}[2em]{0em}
{Set $s=s+1$}. 
\end{addmargin}}
\STATE{
\begin{addmargin}[2em]{0em}
{Solve $\mathscr{U}(\mathbf{a},\mathbf{x},{\tilde{\mathbf{p}}})- \tilde{\mathscr{V}}(\mathbf{a},\mathbf{x},{\tilde{\mathbf{p}}})$ to obtain the data-rate as well as  $\mathbf{a}^{s}$~and $\tilde{\mathbf{p}}^{{s}}$}.
\end{addmargin}}
\STATE{
\begin{addmargin}[1em]{0em}
\textbf{until} $|\widehat{\overline{R}}^{(s)}-\widehat{\overline{R}}^{(s-1)}| \leq \mathbf{\Psi}$ or $s =T_{max}$
\end{addmargin}}
\STATE{
\begin{addmargin}[1em]{0em}
Set \{$\mathbf{a}^i,\tilde{\mathbf{p}}^i$,$\mathbf{p}^i\} =$ $\{\mathbf{a}^{{s}^*},\tilde{\mathbf{p}}^{{s}^*},\mathbf{p}^{{s}^*}\}$.
\end{addmargin}}
\STATE{
\begin{addmargin}[1em]{0em}
Set $i=i+1$.
\end{addmargin}}
\STATE{
\begin{addmargin}[1em]{0em}
Set {$\mathscr{E}_{eff}^{^i} = \frac{\widehat{\overline{R}}^\textnormal{Total}(\mathbf{a}^i,\mathbf{x}^{{t-1}},\tilde{\mathbf{p}}^i)}{\overline{P}^\textnormal{Total}(\mathbf{a}^i,\mathbf{x}^{{t-1}},\tilde{\mathbf{p}}^i,{\mathbf{p}}^i)}$}.
\end{addmargin}}
\STATE \textbf{end while}
\STATE Set \{$\mathbf{a}^{*},\tilde{\mathbf{p}}^{*},\mathbf{p}^*\} =$ $\{\mathbf{a}^{i-1},\tilde{\mathbf{p}}^{i-1},\mathbf{p}^{i-1}\}$.
\STATE \textbf{return} \{$\mathbf{a}^{*},\tilde{\mathbf{p}}^{*}$,$\mathbf{p}^*$\}
\end{algorithmic}
\end{algorithm} 

\subsection{B. Antenna selection}\label{Chap5:layer2}
Since the system throughput is by itself a function of the transmit power, it can be perceived that the effect of total energy consumption on the EE is generally much more substantial than that of the system throughput. 
Here, we aim at investigating the optimal antenna selection for each user among all available antennas to perform ID and EH appropriately.
Assuming the SBS-subcarrier assignment and the power allocation are obtained from a previous iteration {\text{($\mathbf{a}^{t-1}$, $\mathbf{p}^{t-1}$)}} in the first layer, we formulate the following optimization problem to find the optimal antenna set {\text{selection}} for a generalized AS-based SWIPT in the DL direction of a multi-user multi-cell OFDMA network
\begin{subequations}
\begin{align}\label{max_prob.6}
\mathbf{x}^t =
& \arg \max_{\mathbf{x}}
\frac{R^\textnormal{Total}(\mathbf{a}^{t-1},\mathbf{x}^t,\mathbf{p}^{t-1})}
{P^\textnormal{Total}(\mathbf{a}^{t-1},\mathbf{x}^t,\mathbf{p}^{t-1})} \\
\propto
& \arg \min_{\textbf{x}}P^{\textrm{Total}}(\textbf{a}^{t-1},\textbf{x}^{t},\textbf{p}^{t-1}) \\
= & \arg  \min_{\textbf{x}}\bigg(
\sum_{j \in \mathcal{J}}
\sum_{k\in \mathcal{K}}
\sum_{n\in \mathcal{N}}
\big(\frac{{a}_{j,n,k}{p}_{j,n,k}}{\kappa_j}+ P_c^\mathrm{SBS}\big)\bigg)-P^{\textrm{EH}}(\textbf{x}^{t},\textbf{p}^{t-1})\\
s.t.: 
&~C_{3}: 
\sum_{j\in \mathcal{J}} 
R_k(\mathbf{a},\mathbf{x},\mathbf{p}) 
\geq R_{min},~
\forall k \in \mathcal{K},\\
&~C_{4}:
\sum_{m\in \mathcal{M}} 
x^{m}_{j,n,k}=1,~~~~~~~~~~~~
\forall j\in \mathcal{J},~
\forall n \in \mathcal{N},~
\forall k \in \mathcal{K},\\
&~C_{6}:
x^{m}_{j,n,k}\in\{0,1\},~~~~~~~~~~~~~
\forall {j \in \mathcal{J}},~ 
\forall {n\in \mathcal{N}},~  
\forall {k\in \mathcal{K}},~  
\forall {m\in \mathcal{M}}.
\end{align}
\label{chap5:F5_35-new}%
\end{subequations}
It can be noticed that the first term in (\ref{chap5:F5_35-new}) is constant.~Hence, the objective function in (\ref{chap5:F5_35-new}) can be restated as follows 
\begin{align}\label{Chap5:F5-35}
\mathbf{x}^t=
& \arg \max_{\textbf{x}}P^{\textrm{EH}}
(\textbf{x}^{t},\textbf{p}^{t-1})\\
s.t.: &~C_{3}-C_{4},C_{6}. \nonumber
\end{align}
It is worth mentioning that when the power and the corresponding SBS-subcarrier is fixed, the problem can be decomposed into M subproblems, where each subproblem in (\ref{Chap5:F5-35}) can be solved using well-established optimization packages including CVX.
One may conclude that the solution of (\ref{Chap5:F5-35}) would require the knowledge of all branch SNRs.
However, there are various techniques to address this issue based on the quasi-stationary property of the channel gains despite the difficulty of knowing all SNRs simultaneously.
For instance, one may use a training signal in a preamble. 
During this preamble, when the receiver scans the antennas, the highest channel gain is selected for receiving the next data burst or power signal. 

Finally, the pseudo-code of the iterative solution for the antenna selection with a joint SBS-subcarrier assignment and power control is given in \textbf{Algorithm~\ref{last-alg}}.

\begin{algorithm}[t]
\caption{Proposed Iterative Method}
\label{last-alg}
\begin{algorithmic}[1]
\STATE {$\mathbf{Initialize}$} \\
{\begin{addmargin}[1em]{0em}
{iteration index $t = 1$ with the maximum number of iterations $\Delta_{max}$. }
\end{addmargin}}
\STATE{\textbf{repeat} \{Main Loop\}
\\ \textbf{\quad Joint SBS-Subcarrier Assignment and Power Control:}}
\STATE{
\begin{addmargin}[1em]{0em}
For a given antenna set $\textbf{x}^{t-1}$, find the optimal subcarrier assignment $\textbf{a}^{t}$ and power allocation based on (\ref{F5-36_new}) using \textbf{Algorithm \ref{Dinkelbach-algorithm}}. \\
\textbf{Antenna Selection Policy:}
\end{addmargin}}
\STATE{
\begin{addmargin}[1em]{0em}
For a fixed subcarrier assignment $\textbf{a}^{t-1}$ and power allocation $\textbf{p}^{t-1}$,~determine the best antenna based on~(\ref{Chap5:F5-35}).
\end{addmargin}}
\STATE {\quad Set $t = t + 1$.}
\STATE \textbf{until} Convergence with $\mathscr{E}_{eff}^{^*}$ or $t = \Delta_{max}$
\STATE \textbf{return} \{$\mathbf{a}^{*},{\mathbf{x}}^{*}$,$\mathbf{p}^*$\}
\end{algorithmic}
\end{algorithm}

\section{Complexity Analysis}
Our proposed iterative algorithm includes two subproblems: 
1) joint SBS-subcarrier assignment and power allocation and 2) 
antenna selection.
The solution for the first subproblem includes a two-layer approach.
The outer layer provides a solution according to the Dinkelbach algorithm, whereas the inner layer's solution is based on the MM algorithm.
It can be seen that the inner solution to the optimization problem (\ref{F5-36_new}) includes $JNK$ variables and $J+K+JN+6JNK+J^{2}N^{2}K$ linear convex constraints.
As a result, the computational complexity of the first subproblem is $\mathcal{O}(JNK)^{2}(J+K+JN+6JNK+J^{2}N^{2}K)$. 
This can be asymptotically estimated as $\mathcal{O}(JN)^{4}(K)^{3}$, showing a polynomial-time complexity.
Moreover, the complexity of the outer layer is $\mathcal{O}(T_{\text{Dinkelbach}})$, where $T_{\text{Dinkelbach}}$ is the number of iterations needed for the convergence of the outer layer. 
Consequently, the overall complexity of our proposed scheme becomes $\mathcal{O}(T_{\text{Dinkelbach}}T_{\text{MM}}(JN)^{4}(K)^{3})$ in which $T_{\text{MM}}$ is the number of iterations required for reaching convergence in the MM method.
Now, we aim at calculating $T_{\text{MM}}$ for the D.C. programming that includes the interior point method. 
It should be noted that when CVX is adopted based on the interior point method to solve the optimization problem (\ref{F5-36_new}), it requires $\log \frac{JN+J+K+2JNK+JNK+J^{2}N^{2}K}{t^{0}\phi \xi}$ iterations, where $t^{0}$ is the initial point for approximating the accuracy of the interior point method, $0<\phi\ll 1$ is the stopping criterion, and $\xi$ shows the accuracy of the method \cite{Ata_2020}.
For the antenna selection subproblem in (\ref{Chap5:F5-35}), the computational complexity is $\mathcal{O}(JNKM)$.
This computational complexity is substantially lower than the algorithm for the joint power allocation and subcarrier assignment subproblem.
Note that this subproblem can be solved via MOSEK or Gurobi solver, which provides a polynomial-time complexity. 
Besides, it is worth noting that this algorithm reduces the number of variables by half in each optimization subproblem and changes the original problem (\ref{chap5:5:main}) into a mathematically tractable form to be solved.
Furthermore, the number of iterations for this subproblem is $T_{AS}=\log \frac{K+NK+JNKM}{t^{0}\phi \xi}$.
Finally, the total complexity order of the optimization problem, including a two-layer approach, is $\mathcal{O}(T_{\text{Dinkelbach}}T_{\text{MM}}T_{\text{AS}}(JN)^{4}(K)^{3})$. 
This denotes the total number of iteration that is required for reaching convergence of the EE optimization problem.

\section{Simulation Results}
In this section, the performance gain of the proposed scheme for a generalized AS-based SWIPT in the DL direction of a multi-user multi-cell OFDMA system is evaluated under various system parameters.
There are $J=3$ cells in the network topology. 
The radius of a cell, $d_{max}$, is 20 meters, with a reference distance, $d_{0}$, of 5 meters.
Moreover, there are $K_j = 4$ users in each cell uniformly located between the reference distance, $d_{0}$, and maximum coverage of the cell, $d_{max}$.
Also, each user is equipped with two antennas ($M=2$), where the receiver antennas are capable of both ID and EH operations.
Additionally, we consider a frequency-selective fading channel and further assume the central carrier frequency is set to 3 GHz.
The number of subcarriers is $N = 16$, where the bandwidth of each subcarrier is set to 180 kHz.
It should be noted as the power of the background noise on all antennas of each receiver is rather small compared to maximum transmit power, $p_{max}$, it is assumed {$|\sigma_{j,n,k}^{m}|^{^2} = \sigma^{2} = $ -120 dBm} in all simulations.
Since a line-of-sight (LoS) signal is expected in the received signal, the small-scale fading channel is modeled as Rician fading with Rician factor $\rho=3$ dB.
Moreover, the Rician flat fading channel gains include a distance-dependent path loss component of $31.7+10 \alpha \log(\frac{d}{d_0})$ [dB] (where $d$ is the distance between the transmitter and the receiver) and a log-normal shadowing component with~$8$ dB standard deviation, where the path loss exponent is equal to $\alpha = 2.8$~\cite{339880}.
These parameters for propagation modeling and simulations follow the suggestions in 3GPP evaluation methodology~\cite{chap3:3GPP}.
The power conversion efficiency of all active EH antennas, $\epsilon^m_{j,k}$, is assumed to be the same and is equal to  $\epsilon^m_{j,k} = \epsilon = 0.3 $.
For the power consumption model, a constant consumed circuit power, $P_{\textrm{c}}^{\textrm{SBS}}$, is considered for all SBSs and is equal to~23 dBm. 
The power amplifier efficiency of all SBSs is also supposed to be the same and is $\kappa_j = \kappa = 0.2$.   
The target transmission rate is $R_{min}=1$ bit/second/Hz (bps/Hz) for each user.
Furthermore, we conduct Monte Carlo simulations by generating random realizations of the channel gains to obtain the average data-rate of the network.
In fact, the channel gain between a transmitter and a receiver is calculated using independent and identically distributed Rician flat fading and the figures shown in this section are obtained by estimating the average of results over different realizations of the path-loss and the multi-path fading.
The rest of the simulation parameters are given in \textbf{Table}~(\ref{chap:5:Simulation_Parameters}) unless otherwise is specified.

\begin{table}[t]
{\caption{Simulation Parameters}
\label{chap:5:Simulation_Parameters}
\centering
\begin{tabular}{|c|c|}\hline
{\bf Parameter} & {\bf Value} \\ \hline \hline
{Coverage cell ($d_{max}$)} & {$20$ m} \\ 
{Reference distance ($d_{0}$)} & {$5$ m} \\ 
{The number of cell ($J$)} & {$3$}\\
{The number of user in each cell ($K_{j}$)} & {$4$}\\
{The number of antenna pf each user ($M$)} & {$2$}\\
{The number of subcarrier (N)} & {$16$} \\
{Noise power ($\sigma^{2}$)} & {$-120$} dBm \\
{The bandwidth of each subcarrier} & {$180$ kHz}\\
{Path loss exponent ($\alpha$)} & {$2.76$} \\
{Path loss model for cellular links} & { $31.7+27.6 \log(\frac{d}{d_0})$} \\
{Multi-path fading distribution}& {Rician fading with factor 3~dB}\\ 
{Power conversion efficiency of EH antennas ($\epsilon$)}& {30\%}\\ 
{Power amplifier efficiency of SBSs ($\kappa$ )}& {20\%}\\
The maximum transmit power of the SBS ($p_{\textrm{max}}$) & {$30$ dBm} \\
The circuit power consumption of SBSs ($P_{\textrm{c}}^{\textrm{SBS}}$) & {$23$ dBm} \\
The minimum data-rate requirement for each user~($R_{\textrm{min}}$) & $1$ bps/Hz \\
Channel realization number & $100$\\\hline
\end{tabular}}
\end{table}

\subsection{Convergence Speed}
Figure (\ref{plot:5.1}) depicts the average system EE versus the number of iterations of the proposed algorithm under different initialization of the power control; $\mathbf{p}^0(i) = \frac{p_{max}}{N}$, i.e., equal power is assigned to all small-cell users over all subcarriers, $\mathbf{p}^0(i) = p_{max}$, i.e., all subcarriers has the maximum power, and $\mathbf{p}^0(i) = 0$, no power is assigned to the subcarriers initially.
As can be seen, our proposed iterative algorithm runs until it converges to a fixed value. 
Moreover, the convergence rate of the proposed algorithm is considered fast as it reaches to a specific value only after a small number of iterations.
This figure also demonstrates that although the speed of convergence differs from one case to another, in all cases, our proposed algorithm converges to a stationary point only after a small limited number of iterations.

\begin{figure}[t]
\centering
\includegraphics[width=12cm]{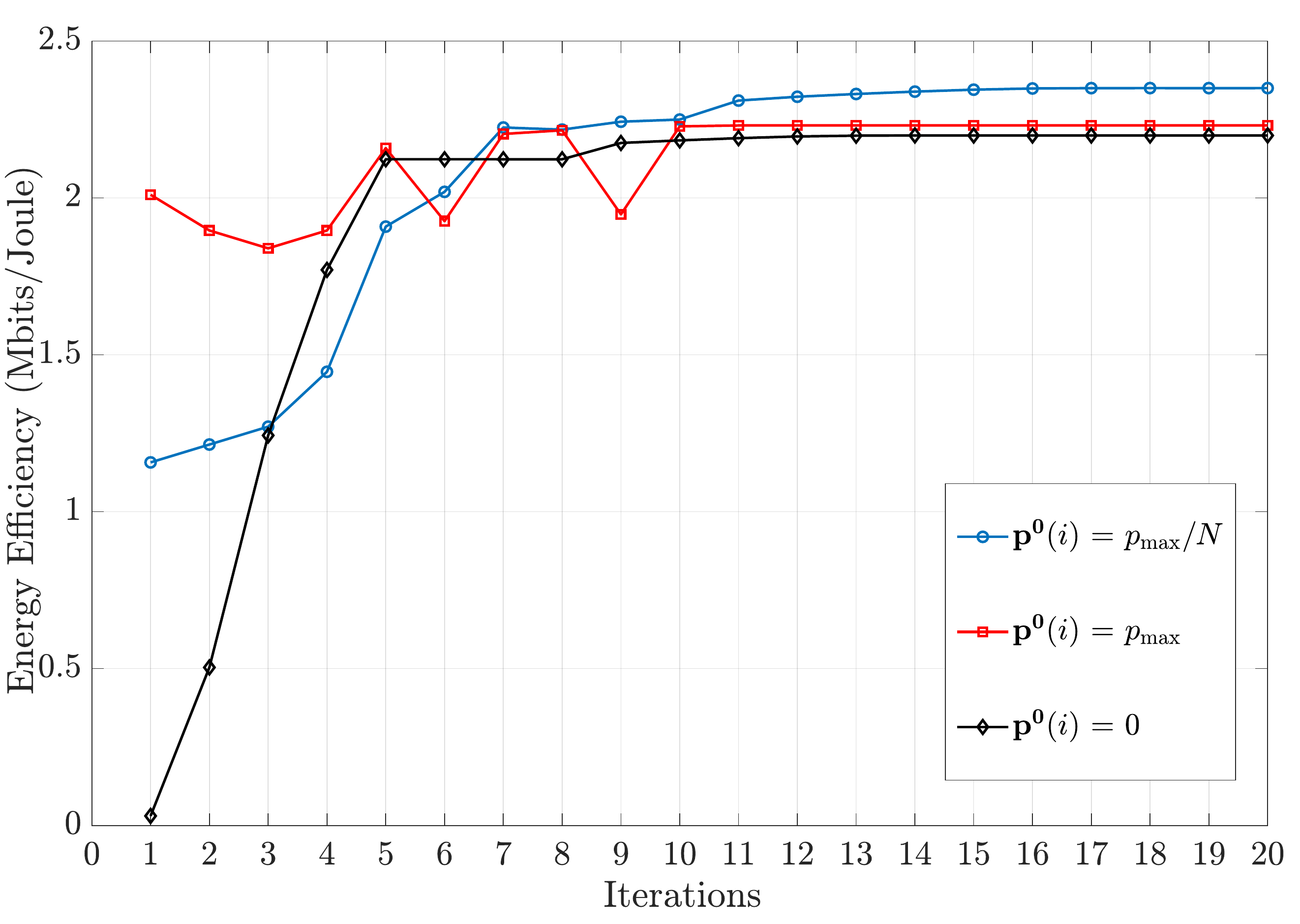}
\caption{Convergence speed.}
\label{plot:5.1}
\end{figure}

\subsection{Energy Efficiency versus Maximum Allowed Transmit Power}
In figure (\ref{plot:5.2}), we present the average EE versus the maximum transmit power $p_{max}$.
As can be observed from this figure, the average system EE for the resource allocation schemes is monotonically non-decreasing with the maximum allowed transmit power. 
This is because the received SINR at the users can be enhanced by allocating the additional available transmit power via the solution of the problem, which leads to an improvement of the system EE. 
However, there is a diminishing return in the average system EE when $p_{max}$ is higher than $30$ dBm. 
As a matter of fact, with an increase in maximum transmit power, the interference power level arising from the other SBSs becomes more severe, resulting in a degradation of the received users' signals. 
Consequently, the throughput of the users deteriorates, which results in a reduction of the EE. 
In particular, by increasing the value of $p_{max}$, the system EE quickly increases at first, and then starts to saturate when $p_{max}$ is higher than $30$ dBm. 
The reason for this is quite evident as the resource allocator is not willing to consume more power, once the maximum EE is obtained.

This figure also consists of four baseline schemes for EE maximization, i.e., Methods B-E, and compares their performance with the proposed iterative \textbf{Algorithm~\ref{last-alg}}, that we call Method A. 
For Method B, we optimize the system EE maximization problem by considering the power splitting receiver architecture. 
In particular, Method B considers receivers with two antennas ($M=2$). Each antenna in this architecture is capable of both EH and ID operations at the same time through a power splitting method with fixed power splitting ratios. 
Moreover, Method~C examines the proposed method in~\cite{TWC_Ata}. 
It considers the EE maximization with respect to the subcarrier assignment, power allocation, and antenna selection for users without the energy harvesting capability. 
Furthermore, Method D is our proposed algorithm based only on the power allocation optimization when random scheduling of the subcarrier assignment and antenna selection variables is performed to obtain the resource allocation policy. 
Finally, Method E is the full power approach in the sense that equal power is allocated across subcarriers for each user.

\begin{figure}[!b]
\centering
\includegraphics[width=12cm]{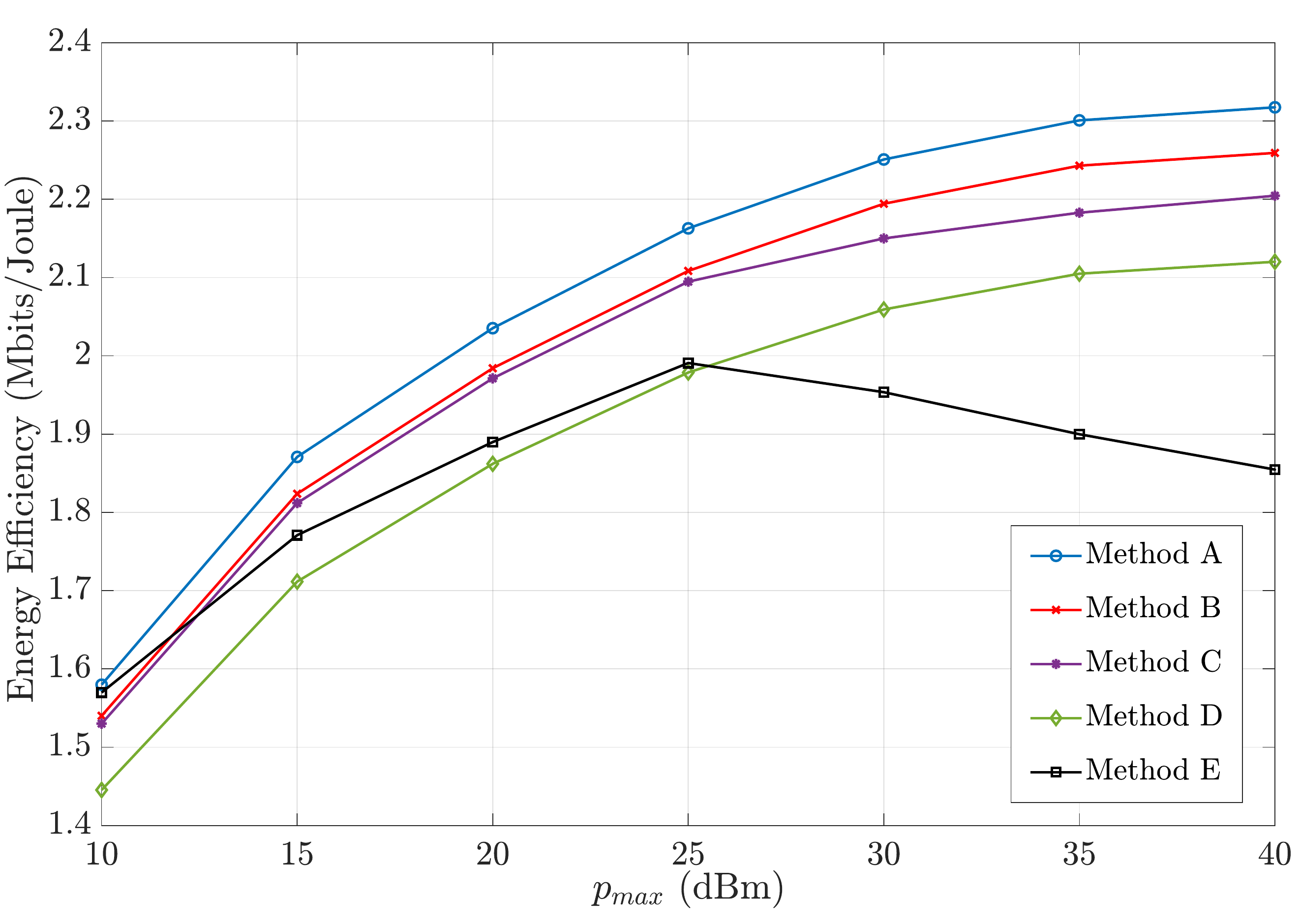}
\caption{Energy efficiency versus maximum allowed transmit power.}
\label{plot:5.2}
\end{figure}

It can be concluded that our proposed iterative algorithm has a better performance in comparison with other methods as we optimize the resource allocation jointly and use a generalized AS-based harvesting technique at the receiver based on the antenna selection architecture. 
We also observe that power control has a significant impact on system EE. As it can be seen in the low transmit power regime, the received power of the desired signal at receivers may not be sufficiently large for simultaneous information decoding and energy harvesting.
It can be realized that in the low transmit power regime, the received power of the desired signal may not be sufficiently large for simultaneous information decoding and energy harvesting.
We also note that for the higher value of the maximum allowed transmit power, the maximum EE achieved by the system via EH active antennas or EH users cannot be attained by a system without energy harvesting capabilities through increasing the transmit power. 
This confirms that energy harvesting contributes to the average system EE. 

\begin{figure}[!b]
\centering
\includegraphics[width=12cm]{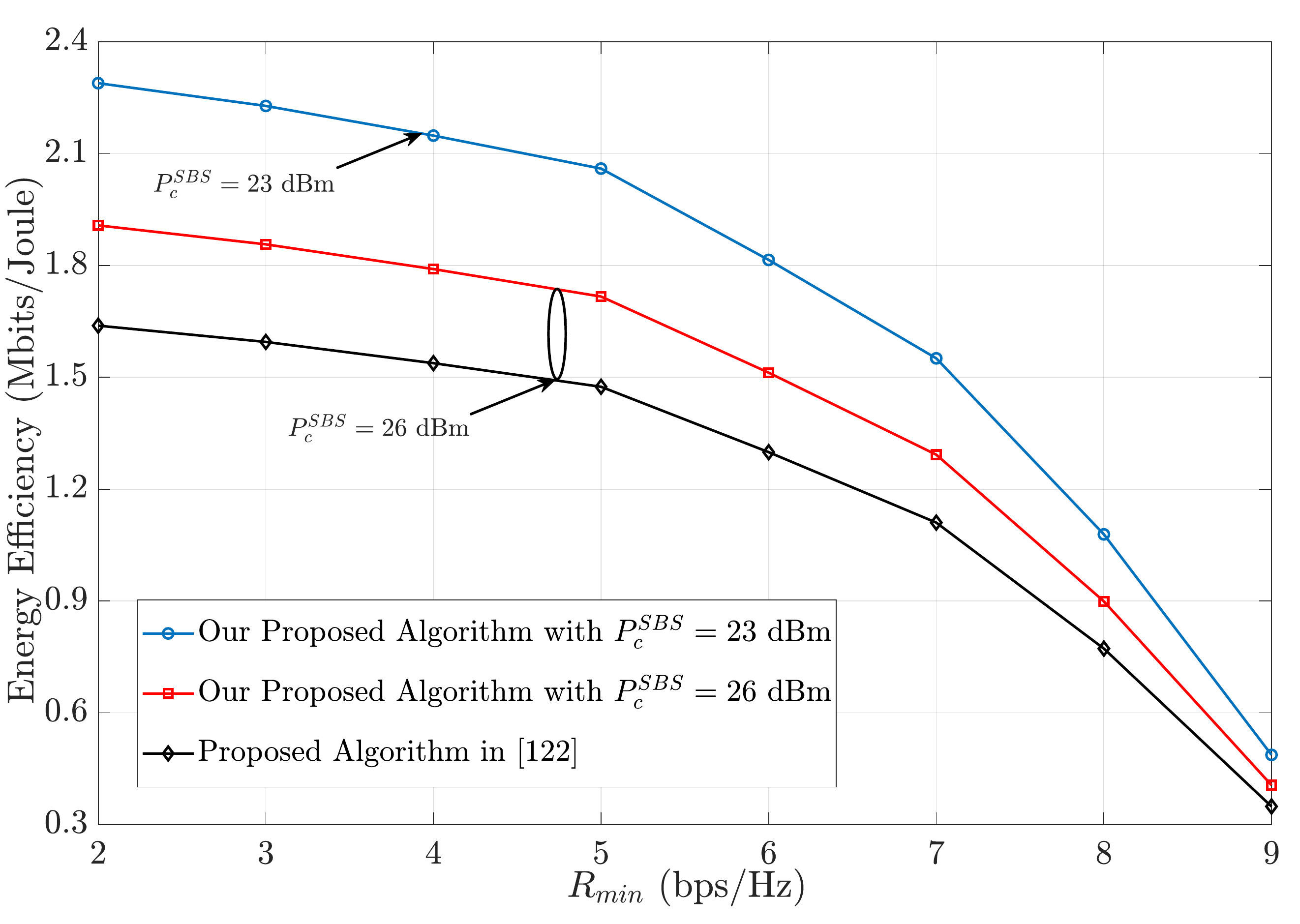}
\caption{Energy efficiency versus minimum data-rate requirement.}
\label{plot:5.3}
\end{figure}

\subsection{Energy Efficiency versus Minimum Data-rate Requirement}
In this section, we show the maximum average EE under different data-rate requirements for different values of circuit power consumption, $P_{\textrm{c}}^{\textrm{SBS}}$, in figure (\ref{plot:5.3}). 
It can be acknowledged that by increasing the data-rate, the system EE stays almost the same up to a specific minimum data-rate requirement value, but starts to decline afterward. 
This is because the required transmit power is similarly low in order to satisfy the QoS provisioning for a lower value of the minimum data-rate requirement. 
However, by increasing the data-rate requirement, more subcarrier are needed and more ID antennas have to be proportionally activated in users with a more mediocre channel quality to meet the QoS requirement. 
This is in addition to a need for an energy-efficient design that operates at a higher transmit power for achieving the optimal system EE.
Moreover, we can observe that our proposed iterative \textbf{Algorithm~\ref{last-alg}} outperforms the baseline scheme algorithm in \cite{jalal} due to performing a joint resource allocation policy as well as obtaining a locally optimal solution.
Figure~(\ref{plot:5.3}) also compares the effect of static circuit power consumption on the system EE. 
From there, we can see that EE decreases with increased circuit power due to higher total power consumption in the network.

\subsection{Energy Efficiency versus Distance}
We investigate the average EE versus different reference distances.
As can be observed from figure~(\ref{plot:5.5}), the system EE decreases with increasing the distance.
Consequently, as the channel strength deteriorates by increasing the distance between transmitter and the receiver due to the effect of the path-loss, more ID antennas have to be activated, and more power must be allocated to users to meet the minimum required data-rate. 
Hence, the average system EE would decline due to increasing the total power consumption in addition to decreasing the total data-rate of the network.
However, the system EE achieved by our proposed iterative \textbf{Algorithm~\ref{last-alg}} still has superior performance as compared to the other algorithms.
It should be noted that the Methods A-D are the same as defined earlier in figure (\ref{plot:5.2}). 

\begin{figure}[t]
\centering
\includegraphics[width=12cm]{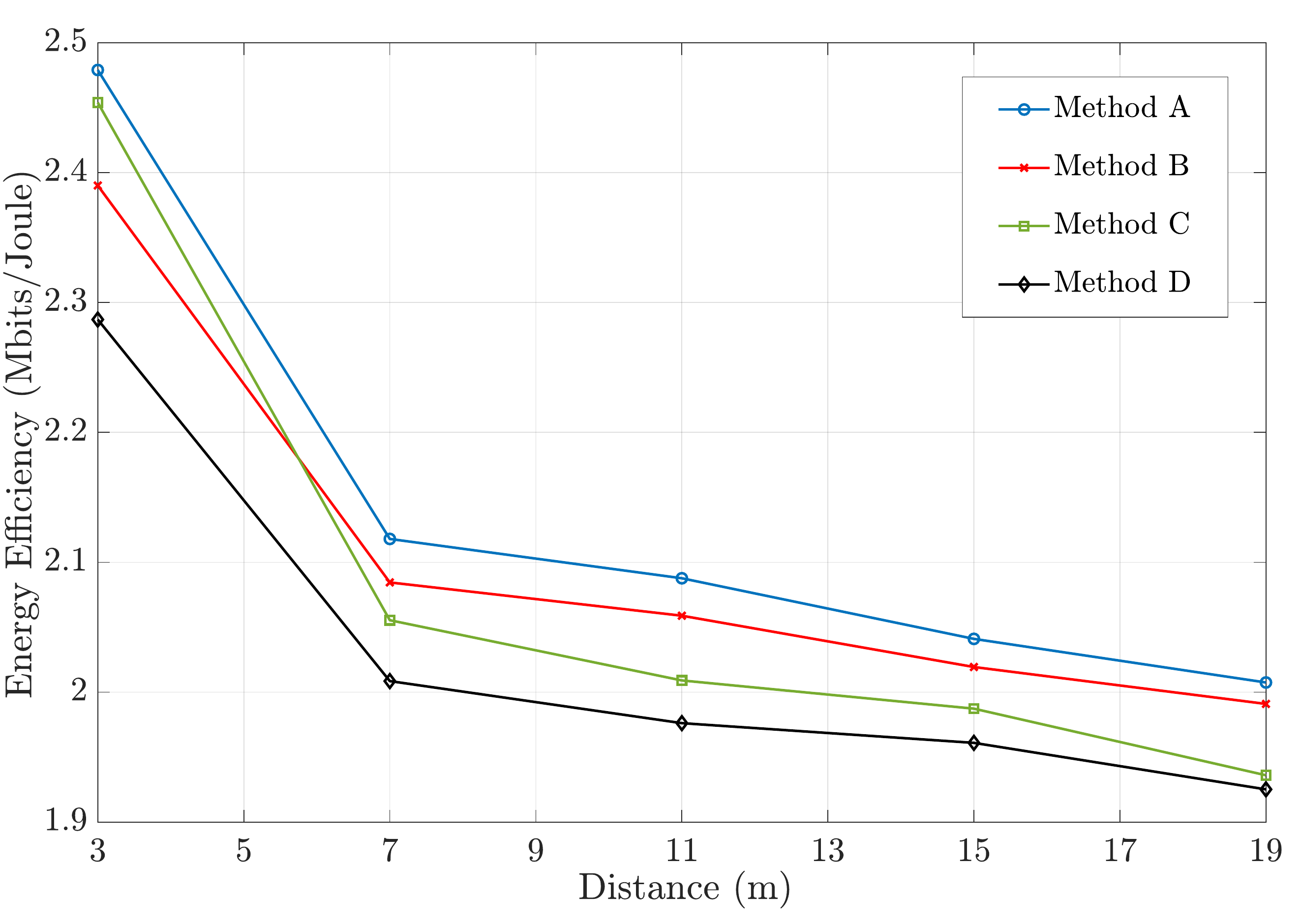}
\caption{Energy efficiency versus distance.}
\label{plot:5.5}
\end{figure}

\subsection{Average System Throughput versus Maximum Allowed Transmit Power}
In figure (\ref{plot:5.6}), we plot the average throughput or data-rate of the network versus the maximum allowed transmit power, $p_{max}$, for different schemes. 
For $p_{max}\leq 30~\text{dBm}$, it can be perceived that the average system throughput of the proposed iterative \textbf{Algorithm~\ref{last-alg}}, that is the Method A in the figure, raises with the maximum transmit power allowance. 
However, the slope of the curve of the system throughput starts to decline and reach a saturation in the high transmit power regime, i.e.,~$p_{max}\geq 30~\text{dBm}$.
In other words, there is a diminishing return in the average system throughput when the maximum transmit power is higher than 30 dBm. 
In fact, as the maximum transmit power increases, the interference power level arising from the other SBSs becomes more severe, which degrades the received users' signal quality.
To compare, we also plot the curve based on the alternative search method (ASM), i.e., Method B, in which the original problem is divided into three disjoint subproblems. 
In this method, the SBS-subcarrier assignment, power allocation, and antenna selection are selected based on the values that are determined in the previous round \cite{jalal}. 
Method C is the proposed algorithm based only on the power allocation when random scheduling of the subcarrier allocation and antenna selection variables is performed to obtain the resource allocation policy. 
Method D is the full power allocation approach with equal power across subcarriers for each user in which the SBS-subcarrier assignment and antenna selection variables are chosen based on our proposed algorithm. It can be perceived that power allocation can significantly increase the performance gain due to controlling interference term arising from other SBSs.

\begin{figure}[!t]
\centering
\includegraphics[width=12cm]{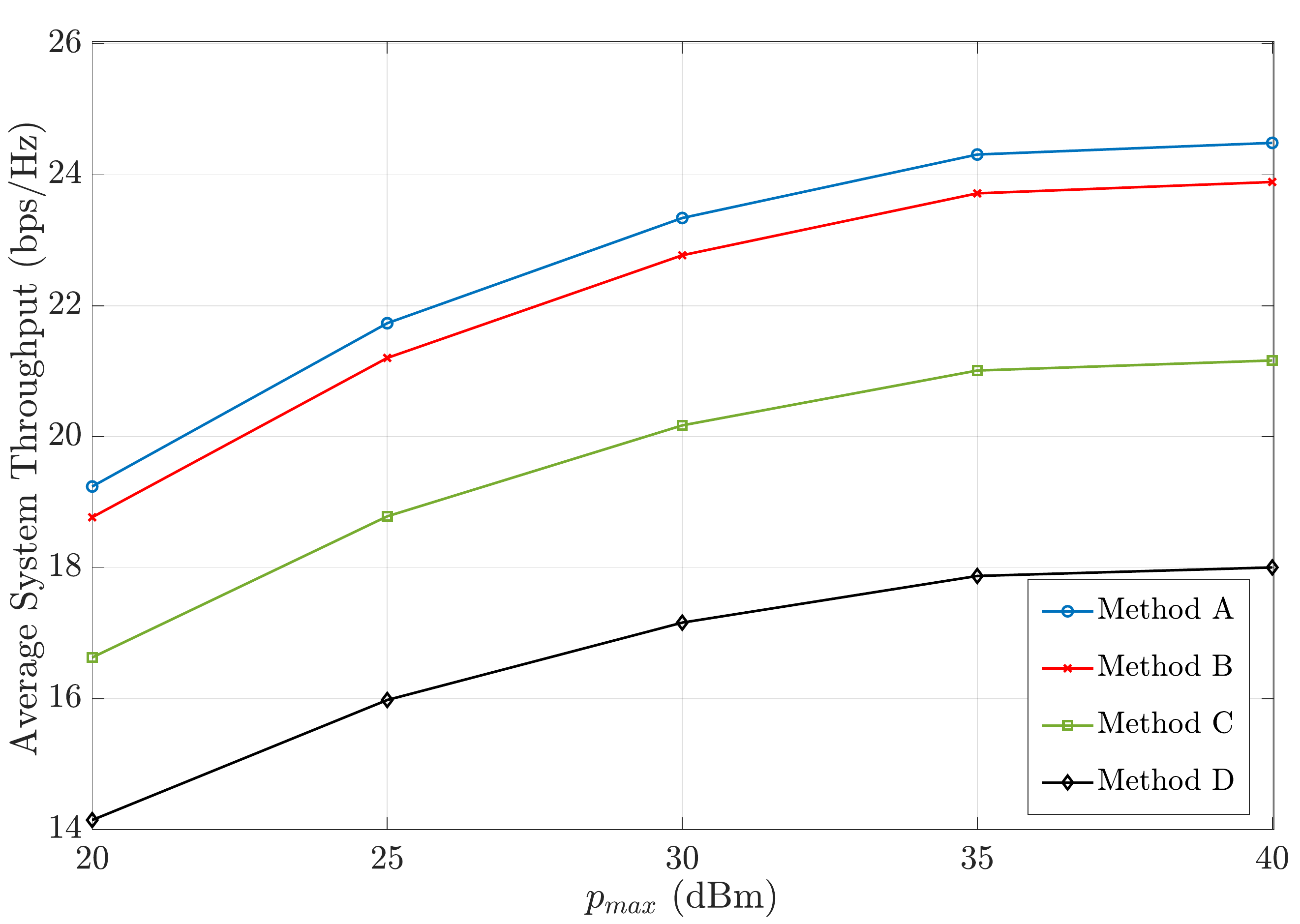}
\caption{Average system throughput versus maximum allowed transmit power.}
\label{plot:5.6}
\end{figure}

\subsection{Average Harvested Power versus Maximum Allowed Transmit Power }
Figure (\ref{plot:5.7}) illustrates the average harvested power versus the maximum allowed transmitting power, $p_{max}$.
As the $p_{max}$ increases, the harvested energy also increases in all considered Methods~A-D. 
However, it is noticeable that for a large value of $p_{max}$, the amount of average harvested power gets saturated.
The reason for this inclination is that the transmitter stops to increase the transmit power for the system EE maximization.
In order to evaluate our performance gain, we also compare our results from the proposed iterative \textbf{Algorithm~\ref{last-alg}}, i.e., Method A, with three baseline Methods B-D.
In Method B, we consider the proposed algorithm in \cite{Wireless_Information}. 
Method C splits the received signal into two power streams for a finite discrete set of power splitting ratios. 
Method A works better than Method B and C due to using a different antenna in each subcarrier, leading to a better degree of freedom of the network.
Finally, Method D is our proposed algorithm with a different objective.
Specifically,~we maximize the system throughput with respect to the same constraints and also further consider minimum energy harvesting requirements for each user.
It can be observed that our proposed algorithm reaches higher performance gain due to employing the antenna selection strategy via the multiplexing gain.
Moreover, by using the antenna selection, we enhance the performance gain due to increasing flexibility in the resource allocation design.

\begin{figure}[t]
\centering
\includegraphics[width=12cm]{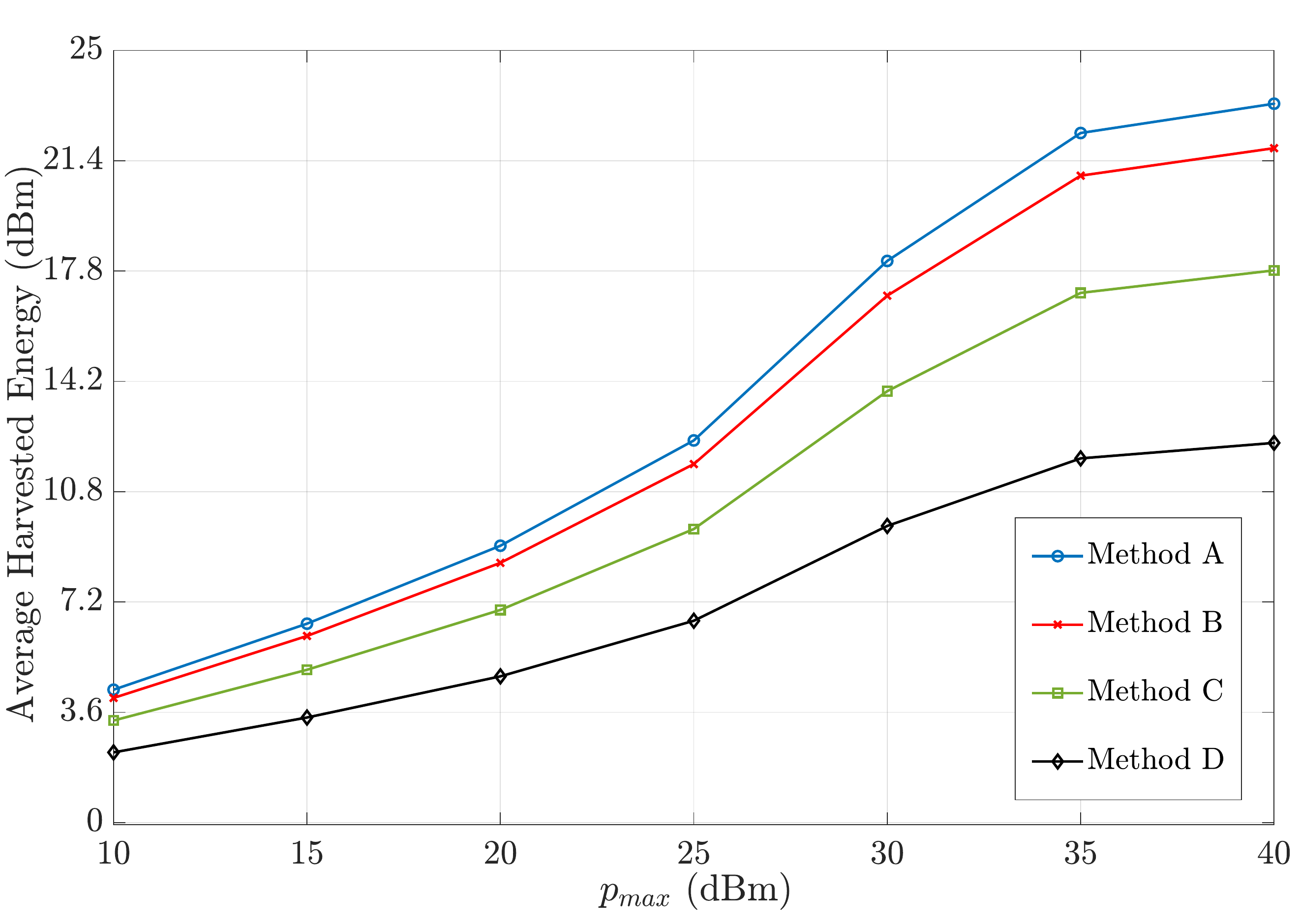}
\caption{Average harvested power versus maximum allowed transmit power.}
\label{plot:5.7}
\end{figure}

\section{Summary}
In this chapter, we addressed the EE optimization problem for the DL of a multi-user multi-cell OFDMA network with the generalized AS-based co-located receivers using SWIPT.
Considering a practical linear power model in which the transmit power consumption, circuit energy consumption, and the harvested energy (by active receiver EH antennas) are taken into account, our goal was to maximize the EE whilst satisfying the minimum data-rate requirement for each user.
The EE optimization problem, which involves a joint optimization of the SBS-subcarrier assignment and power allocation along with an optimal antenna selection, was non-convex and non-linear. 
This made the optimization problem extremely difficult to tackle directly. 
Hence, to obtain a feasible solution for this problem, we employed the MM approach by constructing a sequence of surrogate functions to approximate the non-convex optimization. 
In particular, based on the Dinkelbach method, an optimization problem with a transformed objective function was designed that uses the MM method in its inner loop.
Simulation results revealed the superiority of our proposed method over existing works.
Furthermore, the proposed antenna selection scheme demonstrated that our algorithm provides a good balance of improving in terms of the throughput as well as EE.
\chapter{Conclusion and Future Work}
\label{CHAP6}

\vspace{15mm}
In this chapter, we summarize our overall conclusions.~We also propose some of the future research directions that emerge from this work.

\section{Conclusion}
In chapter \ref{CHAP3}, we introduced a new approach to harvesting ambient energy.
In this approach, a designated portion of the spectrum was used for information decoding (ID) and the rest for energy harvesting (EH), with two separate filters used at the receivers.
We used neither splitters nor switches, which significantly simplified the complexity of the receiver.
Furthermore, we formulated an optimization problem to maximize the harvested energy via a joint subcarrier assignment and power allocation using the simultaneous wireless information and power transfer (SWIPT) scheme for a downlink (DL) of a multi-user small single-cell orthogonal frequency division multiple access (OFDMA) network, fulfilling each user's minimum data-rate requirement.
Our extensive simulation results indicate that our proposed algorithm outperformed other algorithms in the literature by complying with the policy designed for resource allocation.

In chapter \ref{CHAP4}, we extended the system model for SWIPT-enabled single-cell OFDMA proposed in the previous chapter to a multi-cell network based on separated receiver architecture with the goal of maximizing system throughput while respecting the maximum power transfer allowed, the minimum of energy harvested for EH receivers, and the minimum amount of data-rate required for ID receivers.
The resulting problem, which jointly optimizes subcarrier assignment and power allocation, was mixed-integer non-linear programming (MINLP). 
This is intractable because of the multiplication of two variables, the binary subcarrier assignment variable, and the intertwingled interference term in the data-rate function.
We applied the majorization minimization (MM) approach to manage resource allocation policy in this complex problem.
We also analyzed the design of a low complexity algorithm with an upper bound for the interference term by imposing a limiting interference threshold in each subcarrier that can be controlled by the resource allocator to improve system performance.
Simulation results showed the excellent performance gain of our suggested algorithms.

Finally, in chapter \ref{CHAP5}, we described a novel harvesting technique at the receiver that is based on the receiver antenna selection for a multi-user multi-cell SWIPT OFDMA system with a co-located architecture. 
This we named a “generalized antenna switching technique”.
We then maximized energy efficiency (EE) as a performance metric in a two-layered approach to determine a resource allocation policy that optimizes subcarrier assignment, power allocation, and antenna selection with few constraints.
The underlying problem in this chapter was neither linear nor convex due to the fractional form of the objective function, the multiplication of variables, incorporating interference, and integer variables.
We relaxed the integer variable and applied the big-M formulation to make sure that relaxed variables take binary values.
After that, we used the MM approach based on difference of two convex functions (D.C.) programming to find a feasible solution for the inner problem, employing a first-order Taylor approximation to convexify the non-convex functions. 
Next we applied the Dinkelback algorithm to transform the objective function into a non-fractional function.
We observed that generalized antenna switching, also known as “antenna selection strategy”, increases the EE of the system by providing more degrees of freedom as a result of assigning resources via different antennas.
We concluded the chapter with simulation results that demonstrate the superiority of our proposed method. 

\vspace{-3mm}
\section{Future Work}
\vspace{-3mm}
The field of SWIPT in energy-constrained communication systems is a remarkably rich research discipline with great potential.
The following fundamental research directions could be pursued in future work.

\textbf{NOMA-based SWIPT:} 
 
Non-orthogonal multiple access (NOMA) has been suggested as one of the fundamental techniques for beyond fifth generation (5G) and the forthcoming sixth generation (6G) in order to enhance spectral efficiency (SE) while permitting some degree of multiple access interference at receivers by having users share the same spectrum.
NOMA schemes are designed to concurrently serve two or more users at the same base station (BS) or access point (AP) in a single orthogonal resource block.
They can be categorized into main two classes – single-carrier NOMA and multi-carrier NOMA.
The primary principle of single-carrier NOMA is allocating the same time/frequency resources to multiple users by utilizing distinct power levels (“power-domain” NOMA).
In the multi-carrier NOMA scheme, multiple users are multiplexed on different subcarriers by using different codes (sparse code multiple access) and patterns (pattern division multiple access) for each subcarrier~\cite{8452949}.
The successive interference cancellation technique is employed at the receiver end to eliminate expected interference and guarantee improved overall fairness and throughput. 
Most importantly, SE can be achieved in NOMA networks.
It is worth noting that the resource allocation design based on SWIPT-enabled NOMA cellular networks could provide both SE improvement through NOMA and EE improvement via SWIPT.
The very first attempt to maximize EE under multiple restraints and achieve energy-efficient resource management in SWIPT-enabled NOMA was conducted in \cite{8891923}.
Research on this topic has just begun to flourish, with many more potentials awaiting discovery~\cite{jalal2}.

\textbf{Massive MIMO- and mmWave-based SWIPT}:

Propagation losses of broadcasted radio waves may deteriorate SWIPT efficiency, which makes SWIPT particularly applicable to short-distance uses.
Massive multiple-input multiple-output (MIMO) and millimeter-wave (mmWave) technologies would enable a BS or an AP to reliably transfer power to energy-constrained users through ultra-sharp energy beams (that only concentrate transmission energy at certain points) as well as provide  diversity and multiplexing gain~\cite{6736761,6894453,7593259}.
Integrating SWIPT with massive MIMO and mmWave technologies can therefore help overcome SWIPT's deficiencies and further enhance performance in terms of achievable data-rate, overall SE, and EE.
Since the beam-width of the massive MIMO and mmWave based SWIPT system is very small due to the shorter wavelengths (and wider bandwidths), effective initial beam association, beam selection, and beam alignment algorithms are desirable. 
These topics were investigated in~\cite{7491259,8680660,8718514,8828094,8485639,8907878}.
Moreover, considering that the antenna selection technique in multiple antenna systems would significantly decrease the system's power consumption, an exciting research direction would be verifying its applicability to massive MIMO- and mmWave-based SWIPT-enabled networks. 

\textbf{UAV-based SWIPT}:

Unmanned Aerial Vehicles (UAVs) or remotely piloted aircraft have attracted significant attention.~UAVs can improve traditional cellular communication as relays and BSs.
Compared to conventional cellular communication, some of the advantages of UAVs include increased wireless connectivity, high maneuverability, extensive coverage, simple implementation, and cost-effective communication – with increasing affordability.
An unobstructed line-of-sight (LoS) link might be able to help further improve the reliability of wireless communication systems.
UAVs could be effectively integrated with SWIPT by virtue of the dominant presence of LoS connections – acting as flying BSs or APs – to provide a variety of services in areas lacking infrastructure.
However, a major drawback to UAV-based applications is that UAV devices are typically power-hungry devices with limited energy storage performing operations in flight. 
Deployment, trajectory design, and resource allocation must be improved for the efficient utilization of UAVs within a reasonable range of energy-constrained  receivers~\cite{8876702,8892608,8915758,8918344,8941314,8976230,9045325}.

\textbf{MEC-based SWIPT}:

Mobile edge computing (MEC) is a new paradigm for facilitating the real-time implementation of computationally heavy tasks for massive low-power devices (e.g., sensors) by providing cloud-like computing at the edge of mobile networks.
MEC offers solutions for resource-limited wireless devices with high-quality wireless services and multi-media applications because it is capable of offloading some or all the computing tasks of resource-limited devices to nearby APs or BSs, where integrated MEC servers can remotely handle their tasks~\cite{8016573}.
Although MEC covers the computational aspect of massive computation-intensive devices via local computing or offloading, one big challenge remains: providing a sustainable and cost-effective energy supply to resource-limited energy-constrained wireless devices.
Implementing resource allocation policies for mobile users with high computing tasks via MEC and SWIPT would be an interesting topic to explore, with MEC-based SWIPT permitting ubiquitous computing~\cite{8642372,8845134,8904282,8907790,8952620,8970293,9057476}.
\newpage\thispagestyle{empty}\mbox{}  
\newpage\thispagestyle{empty}\mbox{}  
\bibliography{referenties}
\newpage\thispagestyle{empty}\mbox{}  
\newpage\thispagestyle{empty}\mbox{}  
\backmatter
\end{document}